\documentclass[
aps,
prx,
reprint,   
superscriptaddress, 
amsmath,
amssymb,
longbibliography,
footinbib
]{revtex4-2}
\usepackage{amsmath,amsfonts,amssymb,amstext,amscd,dsfont,bbm,natbib}
\usepackage{physics}
\usepackage{color,comment}
\usepackage{soul,xcolor}
\usepackage{gensymb}
\usepackage[colorlinks=true,citecolor=violet,urlcolor=violet,linkcolor=blue]{hyperref}
\usepackage{graphicx}
\usepackage{csquotes}
\usepackage{bbold}
\usepackage{braket}
\usepackage{placeins,comment}
\usepackage{mathtools}
\usepackage{enumitem}
\usepackage[normalem]{ulem}
\usepackage{thm-restate}
\usepackage{cancel}
\usepackage[utf8]{inputenc}
\usepackage[T1]{fontenc}

\usepackage{lipsum}

\usepackage[margin=1in]{geometry}
\usepackage{tikz}
\usetikzlibrary{
	arrows.meta,
	positioning,
	calc,
}

\tikzset{
  lattice node/.style={circle, draw, minimum size=8mm, inner sep=0pt},
  special node/.style={lattice node, fill=yellow!40},
  lattice edge/.style={line width=0.8pt},
  special edge/.style={lattice edge, red, very thick},
}

\usepackage{amsthm}

\newtheorem{lemma}{Lemma}

\newtheorem{definition}{Definition}



\theoremstyle{definition} 
\newtheorem{example}{Example} 
\newtheorem{proposition}{Proposition}

\newtheorem*{statement*}{}

\newtheorem{principle}{Principle}

\definecolor{js}{rgb}{0.578125,0.23828125,0.9609375}
\definecolor{green}{rgb}{0.2,0.8,0.1}
\definecolor{red}{rgb}{1,0,0}

\newcommand{\TTT}{\textbf{T}}
\newcommand{\FFF}{\textbf{F}}
\newcommand{\mTT}{\mbox{\TTT}}
\newcommand{\mFF}{\mbox{\FFF}}
\newcommand{\mR}{\mathbb{R}}
\newcommand{\mC}{\mathbb{C}}

\newcommand{\LL}{\mathtt{L}}

\newcommand{\property}{valuation }

\newcommand{\PrI}[1]{{\small\downarrow}\left(#1\right)}
\newcommand{\GG}[1]{\mathsf{G}\left(#1\right)}

\newcommand{\GGn}[2]{\mathsf{G}^{(#1)}\left(#2\right)}

\newcommand{\SL}[1]{{\mathcal{S}}_{\mathcal{L}}(#1)}

\newcommand{\Ded}{\mathsf{D}}
\newcommand{\Val}{\mathcal{V}}

\usepackage{etoolbox}

\newcommand{\examplesymbol}{\ensuremath{\diamond}} 

\AtBeginEnvironment{example}{\normalfont}
\AtEndEnvironment{example}{\unskip\hfill\examplesymbol}

\begin{document}

\title{Reconstruction of finite Quasi-Probability and Probability from Principles: \\The Role of Syntactic Locality}
\author{Jacopo Surace}
\email{jacopo.surace@gmail.com}
\affiliation{Aix-Marseille University, CNRS, LIS, Marseille, France}

\begin{abstract}
Quasi-probabilities appear across diverse areas of physics, 
but their conceptual foundations remain unclear: they are often treated merely as computational tools, and operations like conditioning and Bayes' theorem become ambiguous. We address both issues by developing a principled framework that derives quasi-probabilities and their conditional calculus from structural consistency requirements on how statements are valued across different universes of discourse, understood as finite Boolean algebras of statements.
We begin with a universal valuation that assigns definite (possibly complex) values to all statements. The central concept is \emph{Syntactic Locality}: every universe can be embedded within a larger ambient one, and the universal valuation must behave coherently under such embeddings and restrictions. From a set of structural principles, we prove a representation theorem (Theorem~1) showing that every admissible valuation can be re-expressed as a finitely additive measure on mutually exclusive statements, mirroring the usual probability sum rule. We call such additive representatives \emph{pre-probabilities}. This representation is unique up to an additive regraduation freedom. When this freedom can be fixed canonically, pre-probabilities reduce to \emph{finite quasi-probabilities}, thereby elevating quasi-probability theory from a computational device to a uniquely determined additive representation of universal valuations. Classical finite probabilities arise as the subclass of quasi-probabilities \textit{stable under relativisation}, i.e., closed under restriction to sub-universes.
Finally, the same framework enables us to define a coherent theory of conditionals, yielding a well-defined \emph{generalized Bayes' theorem} applicable to both pre-probabilities and quasi-probabilities. We conclude by discussing additional regularity conditions, including the role of rational versus irrational probabilities in this setting.
\end{abstract}

\maketitle

\section{Introduction}
What is a probability, and what is it a probability of? In the standard classical paradigm, probabilities express uncertainty: a statement about the world is assumed to be either true or false, and when an agent remains ignorant of the underlying fact, probability represents this incomplete knowledge. In this view, probability is not a quality of the world itself, but a reflection of an agent’s epistemic situation. However, the structure of quantum theory creates a tension with this intuition. Even with a complete specification of preparation and measurement, the theory maintains that no refinement of the description can return certain answers for all questions one might reasonably ask.

We assume that this is not a technical limitation, but it is a reflection of the fact that quantum mechanics describes a world of intrinsic randomness. Recent advancements have demonstrated that observed randomness in quantum phenomena can be rigorously certified as fully intrinsic \cite{dhara2014, acin2016, meng2024}, a foundational insight that has enabled the development of quantum random number generators \cite{Herrero-Collantes2017} and novel methods for quantifying the inherent unpredictability of measurements \cite{senno2022}. If this limit on knowledge is absolute, it becomes unclear what it means to interpret quantum probabilities as degrees of \enquote{missing information}. Missing information about what, exactly~\footnote{This question echoes the inquiry posed by Fuchs regarding the nature of information in the quantum realm in section $4$ of \cite{fuchs2002a}, where an alternative point of view is explored. In particular one can interpret quantum probabilities as \enquote{doxastic quantities} \cite{stacey2019}, from the Greek \enquote{doxa} translatable as \enquote{opinion}, versus the greek \enquote{episteme} translatable as \enquote{knowledge}.}?

The struggle to define probability is not unique to the quantum domain. For over a century, the foundations of the field have remained the subject of ongoing debate \footnote{As noted by A. Hájek \cite{hajek2023}, referencing E.T. Bell who quotes B. Russell: \enquote{\textit{Probability is the most important concept in modern science, especially as nobody has the slightest notion what it means.}}}. 
To bypass these metaphysical stalemates, the modern canon shifted its focus towards internal coherence. This culminated in Kolmogorov’s axiomatic formalisation \cite{kolmogorov1933}, which successfully decoupled mathematical probability from its physical interpretations. While this allowed the theory to be applied universally, it also spurred a tradition of reconstructions: attempts to identify the conceptually motivated principles that single out specific probabilistic formalisms \cite{fine2014}. Such reconstructions have proven to be fundamental instruments for resolving practical problems where pure mathematical tools offer no prescription \cite{richardson1971, stone1975, barone2008, pressacco2007}.

Quantum mechanics followed a strikingly similar trajectory. Its canonical axiomatic formulation was set out by von Neumann \cite{vonneumann2018} only a year before Kolmogorov’s work; yet, much like its probabilistic counterpart, it lacks a single agreed-upon interpretation. This has inspired a vast programme of reconstruction attempts \cite{hardy2001b, hardy2013, chiribella2011a, masanes2011a, dakic2011, barnum2013, janotta2014, bretthorst1988, fuchs2002a, barrett2007a, clifton2003, fuchs2011, pawlowski2009, selby2021, wilce2018, piron1976a, goyal2008, mackey2004}. Because quantum mechanics and probability theory are so closely intertwined, many interpret the former as a generalised or non-commutative form of the latter. Central to this intersection is the language of quasi-probabilities, which has been extensively developed within quantum theory \cite{wigner1932, hiley1987a, kirkwood1933, moyal1949, glauber1963, gross2006, wootters1987, ferrie2008e, spekkens2008b, pashayan2015, zhu2016, schmid2024a, braaschjr.2022, surace2024a, lostaglio2023, levy2020}. Yet, quasi-probabilities themselves remain comparatively underexplored as independent objects; they are often treated as mere technical devices useful in quantum theory as well as other domains. 
A reconstruction of quasi-probability theory in its own right is missing.

If the probabilities that appear in quantum theory must be understood differently from those in the standard paradigm, and if quantum theory in fact requires quasi-probabilities, then we are led to think that quasi-probabilities are not merely technical tools for handling probabilities more conveniently. Rather, they may be something more, namely objects that call for a different kind of interpretation. In this work, we take a first step in that direction. Our aim is not to reconstruct the quasi-probabilities that arise in quantum mechanics; instead, we seek to understand quasi-probabilities in their own right. 
We avoid forcing quasi-probabilities into the classical paradigm of epistemic uncertainty. For this reason, we introduce a \textit{Universal Valuation} $V$, a way of assigning to each statement a definite value that is not restricted to zero or one. We view a universal valuation as quantifying a property of each statement, and we characterise it in the same spirit as one characterises a physical quantity.

To characterise $V$ we propose a set of motivated principles, inspired in part by ideas from physics and relying on the concept of Syntactic Locality. Syntactic Locality formalises a simple intuition: a statement is always understood relative to a larger context of statements, much as in physics any system can be viewed as a subsystem of a larger ambient system. In our setting this becomes a consistency requirement across embeddings. Whenever one syntactic universe embeds into a larger one, the valuation should restrict coherently along that embedding. This compatibility across hierarchies of systems places strong constraints on the behaviour of $V$. Starting from this we lay out five principles that allow us to fully characterise the structure of universal valuations $V$. This defines our programme.

The central technical outcome of the reconstruction and our main result is that the principles force an additive structure. Concretely, we prove (Theorem~1) that every admissible universal valuation admits a reparametrisation as a finitely additive representative on disjoint joins, unique up to an additive regraduation (gauge) freedom. We define a \emph{pre-probability} to be any such additive representative. We then identify quasi-probabilities as those pre-probabilities for which a canonical reparametrisation can be fixed. Finally, we recover finitely additive probabilities as exactly the quasi-probabilities that satisfy an additional stability property (stability under relativisation): the set of all probabilities is the set of all quasi-probabilities that behaves well when considered on local universes.

The same framework also yields as a second outcome of this work a coherent theory of conditionals: it resolves the usual difficulties in defining conditional quasi-probabilities and provides natural generalisations of the rule of total probability and of Bayes' theorem. This differs from the standard practice, where one simply applies the calculus of ordinary probability to an enlarged codomain that includes values outside $[0,1]$. Here the calculus is obtained together with the underlying structure, and when restricted to the classical domain it reproduces the familiar rules.

As a further outcome presented in the appendix, due to the finite additive structure, our reconstruction singles out rational probabilities. Allowing irrational-valued probabilities introduces additional descriptive complexity that is most naturally handled at the level of pre-probabilities. This is in line with the arguments of de Finetti  in favour of finite probability \cite{definetti2017d,dubins2014}.  We argue in favour of rational probabilities both on technical grounds and on interpretational grounds, and we use this to motivate a critique of infinite or hypothetical frequentism and of approaches that take irrational-valued probabilities as primitive. However, we also show that additional requirements of regularity recover a full treatment of irrational probabilities. 

Our reconstruction is inspired by both the logical-Bayesian school \cite{carnap1950, jeffreys1998a, cox2001, jaynes2003, caticha2009, knuth2005c} and the subjective-Bayesian tradition \cite{ramsey2016e, definetti1931, savage1954, jeffrey1990a, bernardo2001}.

\section{Contributions}
This paper develops a principled reconstruction of additive (quasi-)probabilistic calculi from a notion of Syntactic Locality and a set of principles on \enquote{universal valuations} of statements. Concretely, (i) we show that any admissible universal valuation $V$ admits a reparametrisation $R=\varphi\circ V$ that is finitely additive on disjoint joins (Theorem~1), that is, an additive representative which we call a \emph{pre-probability}. (ii) We characterise the resulting gauge freedom (Lemma~2) and introduce semantic frames as the minimal gauge-invariant data needed to compare representations across sub-universes. (iii) We provide a general conditioning/synchronisation framework for pre-probabilities and quasi-probabilities, yielding well-defined rules of total (quasi-)probability and a generalised Bayes' theorem even in settings where standard quasi-probability conditionals are ambiguous. (iv) We show that classical probabilities arise as those quasi-probabilities that are stable under relativisation; in this regime, conditioning reduces to the standard Kolmogorov/Bayes calculus. (v) The appendices contain proofs and additional discussions, including the consequences of assuming additional regularity assumptions and interpretational remarks.

\paragraph{Organisation.}
Section~II introduces universal valuations and states the reconstruction problem. Section~III presents universes and sub-universes and formalises Syntactic Locality. Section~IV gives the principles and derives the main structural constraints. Section~V proves the additivity theorem and introduces pre-probabilities and their gauge freedom. Section~VI develops conditioning and synchronisation, including generalised Bayes' theorem and total conditioning, and Section~VII  describes the probability (positive) regime via stability under relativisation. The appendices contain proofs and additional discussions, including regularity assumptions and interpretational remarks.

\section{Universal valuations and reconstruction of quasi-probability theory}

In this work we will occasionally refer to \textit{standard} probability theory, by which we mean the axiomatic framework introduced by Kolmogorov~\cite{kolmogorov1933}. For convenience, we recall Kolmogorov's three axioms below, since it will be useful to keep them in view throughout the reconstruction. This manuscript will ultimately recover finitely additive probability theory, thus Kolmogorov's axiomatic probability with countable additivity substituted by finite additivity in line with de Finetti \cite{regazzini2013a,howson2008a}. However we will recover the axioms in reverse logical order, starting from the third axiom and working back to the first. When referring to Kolmogorov's standard probability we use the symbol $\mathcal{P}$. We will also refer to \textit{standard} quasi-probability, denoted by $\mathcal{Q}$. By standard quasi-probability we mean the theory obtained by keeping Kolmogorov's axioms except that the first one (non-negativity) is weakened by allowing $\mathcal{Q}$ to take values in the whole $\mathbb{C}$.

For Kolmogorov a probabilistic model consists of a set $\Omega$ of all possible outcomes, a collection $\mathcal F$ of subsets of $\Omega$ (the \emph{events}) closed under complements and countable unions, and a function $\mathcal{P}:\mathcal F\to\mathbb R$ that assigns a real number to each event. Kolmogorov’s axioms require that:
\begin{enumerate}[label=\arabic*.]
	\item \label{ax:non-neg}Non-negativity: for every $A\in\mathcal F$, one has $\mathcal{P}(A)\ge 0$.
	\item Normalisation: $\mathcal{P}(\Omega)=1$.
	\item Countable additivity: if $A_1,A_2,\dots\in\mathcal F$ are pairwise disjoint, then
	\[
	\mathcal{P}\!\left(\bigcup_{i=1}^{\infty} A_i\right)=\sum_{i=1}^{\infty} \mathcal{P}(A_i).
	\]
\end{enumerate}
Standard quasi-probabilities are functions $\mathcal{Q}$ defined as $\mathcal{P}$, but omitting axiom \ref{ax:non-neg}.

In our reconstruction we begin from a small set of principles and progressively recover Kolmogorov’s axioms. 
In what we regard as a merit of this approach, our starting principles imply a weakening of Kolmogorov’s third axiom, replacing countable additivity with finite additivity:
\begin{enumerate}[label=3.\arabic*., start=1]
	\item Finite additivity: if $A_1,\dots,A_n \in \mathcal F$ are pairwise disjoint, then
	\[
	\mathcal P\!\left(\bigcup_{i=1}^{n} A_i\right)
	= \sum_{i=1}^{n} \mathcal P(A_i).
	\]
\end{enumerate}

As we already said, our reconstruction yields a natural hierarchy of theories, each contained in the preceding one, distinguished by which of Kolmogorov’s axioms they satisfy. We call the most general theory \emph{pre-probability}: it satisfies finite additivity 3.1., but need not satisfy axioms 2. and 1.. Imposing normalisation yields \emph{quasi-probability}: every quasi-probability is a pre-probability, but not conversely. Finally, when the theory is specialised so as to satisfy all three axioms, one recovers standard finite probability theory.

\section{Universes and sub-universes: the idea of Syntactic Locality}

Everything we can talk about can be expressed as a statement and in this work we are exactly interested in everything we can talk about.  This is why at the core of our reconstruction is the common notion of a \emph{statement}. To maintain this generality, we do not assume that statements necessarily correspond to facts about the world, nor that they just encode an agent’s information. They can be posed in entirely fictional settings and still count as statements in our sense. One is totally free to talk or write about a fantasy world. 
Even at this level of generality, statements come with formation rules and admit a systematic organisation. They can be combined into more complex statements, transformed by operations, and gathered into a universe of discourse. 
If an agent lists all admissible statements it can form and organise them, it might find that many turn out to be equivalent. Quotienting by this equivalence produces a finite structure, which we take to be the whole universe of what can be said, and in this formal sense, thought within the discourse~\cite{boole2003}.
This notion of universes of all possible statements is central to this work. We call these syntactic universes. 
We take classical propositional logic as our background formalism, since it provides a  simple and non-exotic set of rules for how statements are formed and how reasoning proceeds. Our focus will be on the most basic algebraic structure associated with this logic: the Boolean algebra of propositions~\footnote{Concretely, this Boolean algebra can be obtained as the Lindenbaum--Tarski algebra of classical propositional logic, formed by quotienting formulas by a purely syntactic equivalence relation.}.
Within this formalism we can formally characterise syntactic universes as follows
\begin{definition}[Syntactic universe]
	A Syntactic Universe (or simply a \emph{Universe}) is a set $\mathcal L$ of equivalence classes of statements such that
	\[
	(\mathcal L,\wedge,\vee,\neg,\bot,\top)
	\]
	is a Boolean lattice. We write $s\in\mathcal L$ for an element (a statement). The operations $\wedge,\vee,\neg$ are the algebraic counterparts of conjunction, disjunction, and negation, and $\bot,\top$ are the bottom and top elements. We denote by $\mathcal{L}_n$ a universe with $n\in \mathbb{N}$ atomic statements.
\end{definition}

The top element $\top$ is the join of all the statements in the universe, while the $\bot$ is the meet of all the statements of the universe. Different universes can have different top elements, as the statements of different universes can be different, however we will consider the $\bot$ element to be the same for every universe, in analogy to the empty set or the vacuum state. 
The crucial point that lets us call these structures \textit{universes} is that they are closed under their operations, every new statement formed with any possible combination of operations within the universe is contained inside the universe already. 
However, given a universe $\mathcal{L}$ one can find subsets of statements that form again a structure of Boolean lattice with the same conjunction and disjunction of $\mathcal{L}$ and contains the same bottom element of $\mathcal{L}$. We call these structures sub-universes

\begin{definition}[Sub-universe]
Let $\mathcal{L}$ be a syntactic universe. A \emph{syntactic sub-universe} of $\mathcal{L}$ is a subset
$\tilde{\mathcal{L}} \subseteq \mathcal{L}$ such that $\bot\in \tilde{\mathcal{L}}$, $\tilde{\mathcal{L}}$ is closed under the ambient
$\wedge$ and $\vee$, and there exists an element $\tilde\top\in \tilde{\mathcal{L}}$ and a map
$\tilde\neg:\tilde{\mathcal{L}}\to \tilde{\mathcal{L}}$ for which
\[
(\tilde{\mathcal{L}},\wedge,\vee,\tilde\neg,\bot,\tilde\top)
\]
is a Boolean lattice (with $\wedge,\vee$ the restrictions of those of $\mathcal{L}$).
\end{definition}

Every sub-universe can be understood as a universe in its own right. Hence we may regard any universe as embedded in a larger ambient universe, of which it captures only a fragment closed under the inherited conjunction and disjunction. This motivates a notion of \emph{Syntactic Locality}~\footnote{Syntactic Locality refers only to the embedding/restriction relations between Boolean universes of statements (coarse-grainings and principal ideals). It is not a locality assumption about spacetime, signalling, or causal structure. The term is intended to highlight the subsystem-like role of sub-universes in the logical calculus.}: a universe provides a local description relative to some ambient discourse.
 The terminology is inspired by the way physical theories treat subsystems via reductions and dilations, though the present notion is purely syntactic. 
Thus, an agent’s universe of statements should be seen as a local fragment of a larger discourse, possibly excluding statements that exist in the ambient universe but cannot be formulated within the agent’s own language.
The notion of Syntactic Locality is crucial. Every universe must always be thought of as a local universe of a bigger ambient universe. Every universe is just a (syntactically) local description of a big universe.
It can be useful for this to imagine there is a universal, possibly infinite, but atomic, Boolean universe $\mathbb{U}$, such that all the universes $\mathcal{L}$ one talks about are just local universes of $\mathbb{U}$. Even if we can imagine $\mathbb{U}$, because Syntactic Locality tells us that we consider only local universes always, every principle, theorem, lemma or proposition that appears in this work quantifies only over finite universes $\mathcal{L}$. In particular, because of Syntactic Locality, no axiom quantifies over all of $\mathbb U$.
 
\begin{figure}[h!]
	\begin{tikzpicture}
		\node[lattice node, fill=blue!25, line width=2pt] (bot) at (0,0) {$\bot$};
		
		\node[lattice node, fill=blue!25, line width=2pt] (a) at (-1.5, 1) {$a$};
		\node[lattice node, fill=blue!25, line width=2pt] (b) at (-0.5, 1) {$b$};
		\node[lattice node] (c) at (0.5, 1)  {$c$};
		\node[lattice node] (d) at (1.5, 1)  {$d$};
		
		\node[lattice node, fill=blue!25, line width=2pt] 			(ab) at (-2.5, 2) {$a\vee b$};
		\node[lattice node] 			(ac) at (-1.5, 2) {$a\vee c$};
		\node[lattice node]			(ad) at (-0.5, 2) {$a\vee d$};
		\node[lattice node]    			(bc) at (0.5, 2)  {$b\vee c$};
		\node[lattice node]			(bd) at (1.5, 2)  {$b\vee d$};
		\node[lattice node] 			(cd) at (2.5, 2)  {$c\vee d$};
		
		\node[lattice node] 			(abc) at (-1.5, 3) {\scalebox{0.6}{$a\vee b\vee c$}};
		\node[lattice node] 			(abd) at (-0.5, 3) {\scalebox{0.6}{$a\vee b\vee d $}};
		\node[lattice node] 			(acd) at (0.5, 3)  {\scalebox{0.6}{$a\vee c\vee d$}};
		\node[lattice node] 			(bcd) at (1.5, 3)  {\scalebox{0.6}{$b\vee c\vee d$}};
		
		\node[lattice node] 			(top) at (0, 4) {$\top$};
		
		\draw[lattice edge, very thick] (bot) -- (a);
		\draw[lattice edge, very thick] (bot) -- (b);
		\draw[lattice edge] (bot) -- (c);
		\draw[lattice edge] (bot) -- (d);
		
		\draw[lattice edge, very thick] (a) -- (ab); 	\draw[lattice edge] (a) -- (ac); 	\draw[lattice edge] (a) -- (ad);
		\draw[lattice edge, very thick] (b) -- (ab); 	\draw[lattice edge] (b) -- (bc); 	\draw[lattice edge] (b) -- (bd);
		\draw[lattice edge] (c) -- (ac); 	\draw[lattice edge] (c) -- (bc); 	\draw[lattice edge] (c) -- (cd);
		\draw[lattice edge] (d) -- (ad); 	\draw[lattice edge] (d) -- (bd); 	\draw[lattice edge,] (d) -- (cd);
		
		\draw[lattice edge] (ab) -- (abc); 	\draw[lattice edge] (ab) -- (abd);
		\draw[lattice edge] (ac) -- (abc); 	\draw[lattice edge] (ac) -- (acd);
		\draw[lattice edge] (ad) -- (abd); 	\draw[lattice edge] (ad) -- (acd);
		\draw[lattice edge] (bc) -- (abc); 	\draw[lattice edge] (bc) -- (bcd);
		\draw[lattice edge] (bd) -- (abd); 	\draw[lattice edge] (bd) -- (bcd);
		\draw[lattice edge] (cd) -- (acd); 	\draw[lattice edge] (cd) -- (bcd);
		
		\draw[lattice edge] (abc) -- (top);
		\draw[lattice edge] (abd) -- (top);
		\draw[lattice edge] (acd) -- (top);
		\draw[lattice edge] (bcd) -- (top);
	\end{tikzpicture}
	\caption{A universe $\mathcal{L}_4$ with highlighted in blue the relative (sub-)universe $\mathcal{L}_{a\vee b}$.}
	\label{fig:ex-sub-universe}
\end{figure}
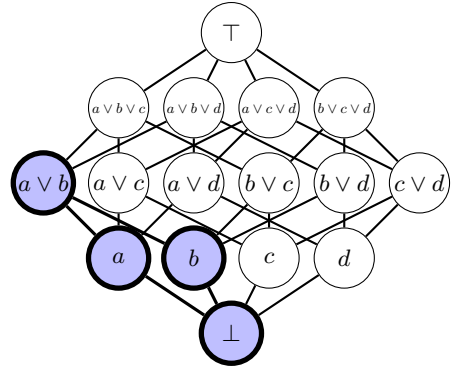

We distinguish two ways in which a universe may sit inside an ambient syntactic universe $\mathcal L$.

First, one may take a Boolean sub-lattice $\tilde{\mathcal L}\subseteq \mathcal L$ in the strict sense: $\tilde{\mathcal L}$ contains
$\bot$ and $\top$ and is closed under the ambient operations $\wedge$, $\vee$, and $\neg$. Such a $\tilde{\mathcal L}$ inherits the
entire Boolean structure of $\mathcal L$ by restriction. Conceptually, these sub-universes act as coarse-grainings of the ambient
universe, since they retain Boolean reasoning but forget distinctions present in $\mathcal L$.

Secondly, one may restrict $\mathcal L$ to the principal ideal generated by an element $t\in\mathcal L$. We define this ideal as
\[
\downarrow t \;:=\; \{\, s\in\mathcal L : s\wedge t = s \,\}.
\]
This set is automatically closed under the ambient $\wedge$ and $\vee$ and contains $\bot$, but it typically fails to be closed
under the ambient complement $\neg$, and therefore need not form a Boolean lattice with top element $t$. Nonetheless, $\downarrow t$
can be made into a Boolean lattice by endowing it with the relative complement
\[
\neg_t s = \neg(\neg t \vee s)
\qquad s\in\downarrow t,
\]
so that $(\downarrow t,\wedge,\vee,\neg_t,\bot,t)$ is Boolean. We call sub-universes arising in this way \emph{relative universes} (see FIG.~\ref{fig:ex-sub-universe}), we formalise them in the following definition

\begin{definition}[Relative universe]
	Let $\mathcal L$ be a universe and let $t\in\mathcal L$.
	The \emph{relative universe at $t$} is the structure
	\[
	\mathcal L_t \;:=\; \downarrow t \;=\; \{\, s\in\mathcal L : s\wedge t = s \,\},
	\]
	equipped with the operations $\wedge$ and $\vee$ inherited from $\mathcal L$, with bottom element $\bot$, top element $t$, and
	with \emph{relative complement} $\neg_t:\mathcal L_t\to\mathcal L_t$ defined by
	\[
	\neg_t s = \neg(\neg t \vee s)
	\qquad s\in\mathcal L_t,
	\]
	where $\neg$ denotes the complement in $\mathcal L$.
\end{definition}


\section{Universal Valuations of statements}
Given a universe $\mathcal{L}$, we call \textit{Universal Valuation} a function 
\[
V:\mathcal{L}\to \mC
\]
that assigns to every statement in the universe a complex number~\footnote{All the following can analogously be formulated for a valuation $\tilde{V}:\mathcal{L}\to\mR$ just by substituting $\mC$ with $\mR$. Since $(\mC,+,0)$ is isomorphic to $(\mR,+,0)$ the proof of Theorem~1 will be identical. In general, any group isomorphic to $(\mC,+,0)$ is suitable in the proof of Theorem~1. Thus one might have a representation additive on real numbers, but also, for example, additive on quaternions. We formulate the theorem referring to complex numbers as we have in mind quasi-probabilities in quantum mechanics and these are most generally represented by vectors of complex numbers, i.e. Kirkwood-Dirac distribution~\cite{lostaglio2023a,schmid2024} or any general frame~\cite{ferrie2008,ferrie2009,ferrie2011}.}. For each universe $\mathcal{L}$ we are interested in characterising the set $\Val(\mathcal{L})$  of admissible valuations. 

We call $V$ \textit{universal} valuation to highlight the idea, derived from the concept of Syntactic Locality, that this is indeed a restriction to local universes of a universal valuation on $\mathbb{U}$. For this reason we will always have that valuations are stable under restrictions, that is if for two universes $\mathcal{L}',\mathcal{L}$ we have that $\mathcal{L}'\subseteq \mathcal{L}$ and $V\in\Val(\mathcal{L})$, then $V|_{\mathcal{L}'}\in\Val(\mathcal{L}')$.

We want to characterise the set of admissible valuations that obey the following principles
\begin{enumerate}[label=\roman*]
	\item \label{dec-prin:classical-compatibility}\textbf{Compatibility with classical propositional logic}:
	The set of admissible valuations must reduce to the set of standard true–false valuations of classical propositional logic in the appropriate limit.

	\item \label{dec-prin:syntactic-locality}\textbf{Local deducibility}:
	There exists a deduction rule that allows one to compute the valuation of statements. This deduction rule allows one to compute the valuation of $s$ from the knowledge of the valuation of all the other statements in the universe. In every universe $\mathcal{L}$, the value assigned by $V$ to a statement is uniquely determined by the values assigned to all the other statements of $\mathcal L$.

	\item \label{dec-prin:universality}\textbf{Universality}:
	The deduction rules are universal: they depend only on the logical structure of the universe and are independent of how its elements are labelled or of the particular valuation $V$. 

	\item \label{dec-prin:realisability}\textbf{Maximum realisability}: The set of admissible valuations is the largest realisable set compatible with the principles.
	
	\item \label{dec-prin:symmetry-removes-freedom}\textbf{Symmetry removes valuational freedom}:  
	 Symmetry removes structure, and structure is the only source of valuational freedom. Hence, when atomic distinctions are erased, the value of a single non-bottom statement fixes the valuation uniquely.

\end{enumerate}

These principles are formalised precisely in the next subsection.

\subsection{Formalisation of the principles}

In this section we discuss and formalise mathematically the idea behind each principle. In this formalisation the principles are intertwined and are not independent. The formalisation of the idea behind each of them will depend also on the structure that earlier introduced principles dictate. For this reason the order in which the principles are introduced matters. 

\subsubsection{Compatibility with classical propositional logic}

The first principle we introduce is compatibility with classical propositional logic.

In the standard paradigm probability quantifies an agent's uncertainty about the property of statements to be either true or false. This tells us that we already recognise truth and falsity as quantification of property of a statement. Thus, among all the value of a valuation $V$ there must be two values that correspond to truth and falsity. Truth and falsity are a quantification of the property of statements. We require the universal valuation $V$ to acknowledge and include this classical limit. 

We formalise this requirement as follows.

On any universe $\mathcal{L}$, a truth assignment is a function
\[
\LL:\mathcal{L}\to\{\mTT,\mFF\},
\]
which declares each statement either true or false.

The truth assignments we consider include all classical Boolean truth assignments on a Boolean lattice, together with the degenerate assignment that maps every statement in the universe to $\mFF$. This degenerate assignment is usually excluded in classical logic, largely because one typically fixes a universe of discourse and does not treat arbitrary sub-universes as legitimate universes in their own right. In this work, however, locality plays a central role, and sub-universes must be regarded as proper universes. Once one allows restrictions to sub-universes of a Boolean lattice equipped with classical truth assignments, the all-$\mFF$ assignment appears naturally and must be included.

We say that  a set of valuations is compatible with classical propositional logic if there exist two distinguished numbers $\tau,\varphi\in \mC$ with $\tau\neq\varphi$ that play the roles of $\mTT$ and $\mFF$, and such that passing between truth assignments and $\{\tau,\varphi\}$-valued valuations is always possible and unambiguous (see FIG.~\ref{fig:classical-correspondence} for a pictorial representation of the correspondence between valuations and classical truth assignment).

\begin{principle}[Compatibility with classical propositional logic]
	A set of valuations is compatible with classical propositional logic if it respects these properties.
	\begin{enumerate}[label=(\roman*)]
  \item \textbf{Two-valued case.}
  Two-valued valuations $V:\mathcal L\to\{\tau,\varphi\}$ are in bijection with truth assignments
  $\LL:\mathcal L\to\{\mTT,\mFF\}$ via
  \begin{align*}
    \LL(s)=\mTT \ &\Longleftrightarrow\ V(s)=\tau,\\
    \LL(s)=\mFF \ &\Longleftrightarrow\ V(s)=\varphi,
  \end{align*}
		\item The valuation of the bottom element is always fixed  as in the classical propositional logic case:
		\[
		V(\bot)\coloneqq \varphi.
		\]
	\end{enumerate}
\end{principle}

This principle says that whenever a valuation happens to be $\{\tau,\varphi\}$-valued it is, in every respect, a classical truth assignment, and every classical truth assignment can be seen as a special case of our framework.
\begin{figure}[h!]
	\centering
	\begin{tikzpicture}
		\node[align=center] at (0,3.6) {Universal Valuation $V$};
		\node[lattice node, fill=red!25]   (Lbot) at (0,0) {$\bot$};
		
		\node[lattice node, fill=red!25]   (La)   at (-1,1) {$a$};
		\node[lattice node, fill=blue!25]  (Lb)   at (0,1)  {$b$};
		\node[lattice node, fill=red!25]   (Lc)   at (1,1)  {$c$};
		
		\node[lattice node, fill=blue!25]  (Lab)  at (-1,2) {$a\vee b$};
		\node[lattice node, fill=red!25]   (Lac)  at (0,2)  {$a\vee c$};
		\node[lattice node, fill=blue!25]  (Lbc)  at (1,2)  {$b\vee c$};
		
		\node[lattice node, fill=blue!25]  (Ltop) at (0,3) {$\top$};
		
		\draw[lattice edge] (Lbot) -- (La);
		\draw[lattice edge] (Lbot) -- (Lb);
		\draw[lattice edge] (Lbot) -- (Lc);
		
		\draw[lattice edge] (La) -- (Lab);
		\draw[lattice edge] (La) -- (Lac);
		
		\draw[lattice edge] (Lb) -- (Lab);
		\draw[lattice edge] (Lb) -- (Lbc);
		
		\draw[lattice edge] (Lc) -- (Lac);
		\draw[lattice edge] (Lc) -- (Lbc);
		
		\draw[lattice edge] (Lab) -- (Ltop);
		\draw[lattice edge] (Lac) -- (Ltop);
		\draw[lattice edge] (Lbc) -- (Ltop);
		
		\draw[<->, line width=0.8pt] (1.7,1.5) -- (2.5,1.5);
		
		\begin{scope}[xshift=4cm]
			\node[align=center] at (0,3.6) {Classical truth assignment $\LL$};
			\node[lattice node, fill=gray!25, align=center] (Rbot) at (0,0) {\(\bot\)\\\(\mFF\)};
			
			\node[lattice node, fill=gray!25, align=center] (Ra)   at (-1,1) {\(a\)\\\(\mFF\)};
			\node[lattice node, fill=gray!25, align=center] (Rb)   at (0,1)  {\(b\)\\\(\mTT\)};
			\node[lattice node, fill=gray!25, align=center] (Rc)   at (1,1)  {\(c\)\\\(\mFF\)};
			
			\node[lattice node, fill=gray!25, align=center] (Rab)  at (-1,2) {\(a\vee b\)\\\(\mTT\)};
			\node[lattice node, fill=gray!25, align=center] (Rac)  at (0,2)  {\(a\vee c\)\\\(\mFF\)};
			\node[lattice node, fill=gray!25, align=center] (Rbc)  at (1,2)  {\(b\vee c\)\\\(\mTT\)};
			
			\node[lattice node, fill=gray!25, align=center] (Rtop) at (0,3) {\(\top\)\\\(\mTT\)};
			
			\draw[lattice edge] (Rbot) -- (Ra);
			\draw[lattice edge] (Rbot) -- (Rb);
			\draw[lattice edge] (Rbot) -- (Rc);
			
			\draw[lattice edge] (Ra) -- (Rab);
			\draw[lattice edge] (Ra) -- (Rac);
			
			\draw[lattice edge] (Rb) -- (Rab);
			\draw[lattice edge] (Rb) -- (Rbc);
			
			\draw[lattice edge] (Rc) -- (Rac);
			\draw[lattice edge] (Rc) -- (Rbc);
			
			\draw[lattice edge] (Rab) -- (Rtop);
			\draw[lattice edge] (Rac) -- (Rtop);
			\draw[lattice edge] (Rbc) -- (Rtop);
		\end{scope}
	\end{tikzpicture}
	\caption{Correspondence between valuation and classical truth assignment. On the left in different colours are represented the values $\varphi$ and $\tau$ for $V$.}
	\label{fig:classical-correspondence}
\end{figure}
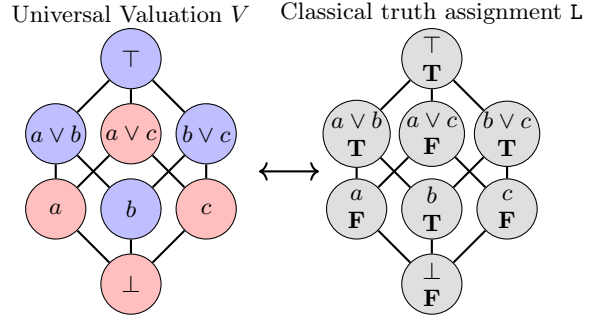

\subsubsection{Local Deducibility}

The principle of Local Deducibility says that, within a fixed syntactic universe, the valuation of any statement is determined by the valuations of all the other statements in the same universe.
Equivalently, an agent who knows the valuation of every statement in its universe except for a single statement $\tilde s$ has enough information to uniquely reconstruct the valuation of $\tilde s$ (see FIG.~\ref{fig:syntactic-locality} for a graphical representation of a use case of local deducibility).

In order for an agent to reconstruct the valuations of its distinguished statement there must be some other non-trivial statements in the universe, thus the principle of Local Deducibility holds whenever the universe has more than one atomic statement. 
Indeed, if $\mathcal L=\{\bot,\top\}$ an agent would have no other data to deduce the valuation of $\top$, in fact, from principle~\ref{dec-prin:classical-compatibility} $V(\bot)\coloneqq \varphi$ thus we do not include it as data for the deduction. In this case the value of the top element is free.

We formalise this principle as follows.
\begin{principle}[Local Deducibility]
	\label{prin:syntactic-locality}
	Let $\mathcal{L}$ be a universe described by a finite Boolean lattice with $n\ge 2$ atoms.
	For every distinguished statement ${\tilde s\in\mathcal L\setminus\{\bot\}}$ there exists a function, a \textbf{deduction rule}, $\Ded_{\mathcal L,\tilde s}$ such that, for every valuation $V\in\Val(L)$,
	\begin{equation}
		\label{eq:syntactic-locality}
		V(\tilde s)
		=
		\Ded_{\mathcal L,\tilde s}\!\left(\{V(s)\}_{s\in\mathcal L\setminus\{\tilde s,\bot\}}\right).
	\end{equation}
Furthermore, if $\mathcal{L}_a=\{\bot,a\}$ is a $1$-atom universe, then for every $x\in C$ there exists $V\in\Val(L_a)$ such that $V(a)=x$.
\end{principle}

In other words, the value assigned to $\tilde s$ is fixed by the valuation on all other statements (here we exclude explicitly the valuation of $\bot$ as a reminder that it is a fixed value).
Since the deduction rule $\Ded_{\mathcal L,\tilde s}$ is a function, this deduction is unambiguous.

This principle seems quite weak as it requires knowing the valuation of \textit{all} other statements in the universe. Inspired by physics, one could imagine that a more sensible principle would be asking for the valuation of a statement to be completely deducible from its nearest neighbour, for whatever definition of nearest neighbour one can imagine in the case of syntactic universes. However, the interplay of this principle, with the next one  (Universality) and Syntactic Locality will result in a deduction rule that is indeed \enquote{more local} than what this principle alone asks for.

\begin{figure}[h!]
	\begin{tikzpicture}
		\node[lattice node, fill=green!25] (bot) at (0,0) {$\bot$};               
		
		\node[lattice node, fill=green!25] (a)  at (-1,1) {$a$};
		\node[lattice node, fill=green!25] (b)  at (0,1)  {$b$};
		\node[lattice node, fill=green!25] (c)  at (1,1)  {$c$};
		
		\node[lattice node, fill=green!25]  (ab) at (-1,2) {$a\vee b$};
		\node[lattice node] (ac) at (0,2)  {$a\vee c$};
		\node[lattice node, fill=green!25]  (bc) at (1,2)  {$b\vee c$};
		
		\node[lattice node, fill=green!25] (top) at (0,3) {$\top$};
		
		\draw[lattice edge] (bot) -- (a);
		\draw[lattice edge] (bot) -- (b);            
		\draw[lattice edge] (bot) -- (c);
		
		\draw[lattice edge] (a) -- (ab);
		\draw[lattice edge] (a) -- (ac);
		
		\draw[lattice edge] (b) -- (ab);
		\draw[lattice edge] (b) -- (bc);
		
		\draw[lattice edge] (c) -- (ac);
		\draw[lattice edge] (c) -- (bc);
		
		\draw[lattice edge] (ab) -- (top);
		\draw[lattice edge] (ac) -- (top);        
		\draw[lattice edge] (bc) -- (top);
		
		\node[anchor=west, xshift=-8.4pt, yshift=12pt] at (ac.east) {\large\textbf{?}};
	\end{tikzpicture}
	\caption{From the principle of Local Deducibility one can deduce  $V(a\vee c)$ from the valuation of all the green statements.}
	\label{fig:syntactic-locality}
\end{figure}
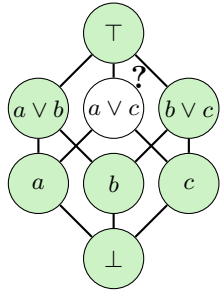

\subsubsection{Universality}

The principle of universality expresses that the rules of deduction depend only on the \emph{structure} of the syntactic universe, and not on arbitrary labels attached to statements.
For example we can imagine two syntactic universes $\mathcal{L}_X=\{\bot,a,b,\top_X\coloneqq a\vee b\}$ and $\mathcal{L}_Y=\{\bot,c,d,\top_Y\coloneqq c\vee d\}$. These two universes have the same structure, but the labels of the elements are different. We want the deduction rules to be independent, that is the valuation of $a$ should be deducible from the valuation of the other statements in $\mathcal{L}_X$ with the same deduction rule that one would use for computing the valuation of $c$ from the valuation of the other statements in $\mathcal{L}_Y$.

The notion of relabelling is captured by the notion of Boolean lattice isomorphism, thus we want for the rules of deduction to be invariant under Boolean lattice isomorphism. 
In finite Boolean lattices, statements at the same level (or rank) are isomorphic, so we expect them to obey the same deduction rule, while statements at different levels are not isomorphic and we expect them to obey different deduction rules.
Universality therefore identifies the deduction rule within each structural class of statements, that is it requires that the deduction rules of the principle~\eqref{dec-prin:syntactic-locality} depend only on the level of the statement in the Boolean lattice and not on the specific statement.

To make the interplay between universality and the idea of Local Deducibility explicit, in appendix \ref{app:global-local-universality} we split universality into two complementary requirements: a global one (within a fixed universe) and a local one (across isomorphic sub-universes).
Both can be summarised in the single principle of universality stating that the deduction rule depends only on the number $n$ of atoms of the universe (Boolean lattice isomorphisms do not change the number of atoms) and on the level $l$ of the distinguished statement (see FIG.~\ref{fig:levels} for a graphical representation of levels).

\begin{principle}[Universality]
	\label{prin:SL-U-rank}
	For each pair $(n,l)$ with $n\ge 2$ and $1\le l\le n$ there exists a deduction rule $\Ded_{n,l}$ such that whenever $\mathcal L$ is a Boolean lattice with $n$ atoms, $V\in \Val(\mathcal{L})$, and $\tilde s\in\mathcal L\setminus\{\bot\}$ at level $l$ of the lattice, one has
	\begin{equation}
		\label{eq:SL-U-rank}
		V(\tilde s)=\Ded_{n,l}\!\left(\{V(s)\}_{s\in\mathcal L\setminus\{\tilde s,\bot\}}\right).
	\end{equation}
	In particular, the deduction rule depends only on $n$ and on the level $l$ of the distinguished statement, and not on the specific universe or on the labels of its elements.
\end{principle}

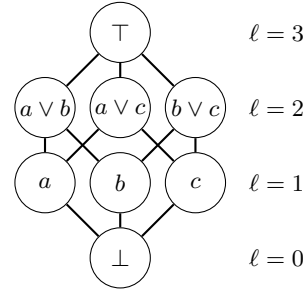
\begin{figure}[h!]
	\begin{tikzpicture}
		\node[lattice node] (bot) at (0,0) {$\bot$};
		
		\node[lattice node] (a) at (-1,1) {$a$};
		\node[lattice node] (b) at (0,1) {$b$};
		\node[lattice node] (c) at (1,1) {$c$};
		
		\node[lattice node] (ab) at (-1,2) {$a\vee b$};
		\node[lattice node] (ac) at (0,2) {$a\vee c$};
		\node[lattice node] (bc) at (1,2) {$b\vee c$};
		
		\node[lattice node] (top) at (0,3) {$\top$};
		
		\draw[lattice edge] (bot) -- (a);
		\draw[lattice edge] (bot) -- (b);
		\draw[lattice edge] (bot) -- (c);
		
		\draw[lattice edge] (a) -- (ab);
		\draw[lattice edge] (a) -- (ac);
		
		\draw[lattice edge] (b) -- (ab);
		\draw[lattice edge] (b) -- (bc);
		
		\draw[lattice edge] (c) -- (ac);
		\draw[lattice edge] (c) -- (bc);
		
		\draw[lattice edge] (ab) -- (top);
		\draw[lattice edge] (ac) -- (top);
		\draw[lattice edge] (bc) -- (top);
		
		\node[anchor=west] at (1.6,0) {$\ell=0$};
		\node[anchor=west] at (1.6,1) {$\ell=1$};
		\node[anchor=west] at (1.6,2) {$\ell=2$};
		\node[anchor=west] at (1.6,3) {$\ell=3$};
	\end{tikzpicture}
	\caption{The levels of a universe with three atomic statement labelled.}
	\label{fig:levels}
\end{figure}

\subsubsection{Maximum Realisability}
The principle of maximum realisability is meant to encode the fact that we want to consider all the possible valuations compatible with our principles. We want to be sure that if a valuation is compatible with our principle, then this valuation must be admissible. 
Since, by the idea of Syntactic Locality, we know that any syntactic universe and thus any valuation is always local in a bigger ambient universe, we have to approach the concept of maximum realisability from the point of view of local universes. For this purpose it is useful to rely on the following definition.
\begin{definition}[Consistent partial valuation]
	\label{def:consistent}
	Let $\mathcal{L}$ be a universe. A \textbf{partial valuation} on $\mathcal{L}$ is a map $p:\mathcal{S}\to \mC$ with $\mathcal{S}\subseteq \mathcal{L}$ not necessarily a sub-universe,
	and $\bot \in \mathcal{S}$. The partial valuation $p$ is \textbf{consistent} if:
	\begin{enumerate}
		\item $p(\bot)=\varphi$;
		
		\item for every sub-universe $\mathcal{L}'\subseteq \mathcal{L}$, whenever $ \mathcal{L}'\subseteq \mathcal{S}$ and $\mathcal{L}'\setminus\{\tilde s,\bot\}\neq \emptyset$,  for every $\tilde s\in  \mathcal{L}'\setminus\{\bot\}$ one has
		\[
		p(\tilde s)
		=
		\Ded_{\mathcal{L}',\tilde s}\!\Big(\{\,p(s)\,\}_{s\in \mathcal{L}'\setminus\{\tilde s,\bot\}}\Big).
		\]
	\end{enumerate}
\end{definition}

We call it a consistent partial valuation because we want to think of it as an \enquote{incomplete} specification of a complete valuation on a universe $\mathcal{L}$. In other words, we do not know the valuation of every statement in $\mathcal{L}$, but only of some of them. We denote by $\mathcal{S}\subseteq \mathcal{L}$ the subset of statements whose valuations are known.
It is consistent because we have no reason to believe that it fails to be the restriction of an admissible valuation on $\mathcal{L}$. Indeed, points~1 and~2 show that this partial assignment does not introduce any inconsistencies with the deduction rules we are able to check.

Once we have verified consistency, we may legitimately treat a consistent partial valuation as an \enquote{incomplete} specification of a complete valuation on $\mathcal{L}$. This is captured by the following principle, which is how we precisely formalise the notion of maximum realisability.

\begin{principle}[Maximum realisability]
	\label{prin:realisability}
	Let $\mathcal{L}$ be a syntactic universe. Every consistent partial valuation $p:S\to \mC$ on $\mathcal{L}$ extends to an admissible valuation $V\in\Val(L)$, that is,
	\[
	V|_{S}=p.
	\]
\end{principle}

That is, if there is no reason for a local agent within a local universe to not admit a valuation, then this valuation must be admissible. This principle is used to ensure that any assignment of values to a set of statements that does not violate deducibility constraints can be realised by some admissible valuation. Technically, it is invoked to guarantee surjectivity/invertibility properties of the binary composition rule $\mathsf{G}$ in the proof of the additivity theorem  appendix~\ref{app:study-of-G}, and to allow free choice of atomic values when constructing test valuations.

\subsubsection{Symmetry removes valuational freedom}
\label{subsec:rigidity-under-symmetry}

Local Deducibility and universality originate from the idea that a valuation depends only on distinctions that are present at the level of syntax.
Any structure that is not syntactically accessible is invisible to the valuation and cannot affect its assignments.

A particularly strong constraint arises when atomic statements are treated as indistinguishable.
If all atoms are assigned the same value, then no syntactic label remains that allows the valuation to tell them apart.
In this situation the valuation cannot access the internal combinatorial structure of the Boolean lattice.
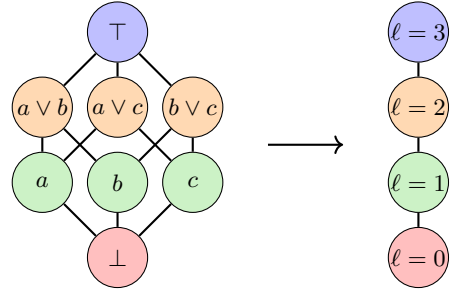
\begin{figure}[h!]
	\begin{tikzpicture}
		\node[lattice node, fill=red!25] (bot) at (0,0) {$\bot$};               
		
		\node[lattice node, fill=green!25] (a)  at (-1,1) {$a$};
		\node[lattice node, fill=green!25] (b)  at (0,1)  {$b$};
		\node[lattice node, fill=green!25] (c)  at (1,1)  {$c$};
		
		\node[lattice node, fill=orange!30] (ab) at (-1,2) {$a\vee b$};
		\node[lattice node, fill=orange!30] (ac) at (0,2)  {$a\vee c$};
		\node[lattice node, fill=orange!30] (bc) at (1,2)  {$b\vee c$};
		
		\node[lattice node, fill=blue!25] (top) at (0,3) {$\top$};
		
		\draw[lattice edge] (bot) -- (a);
		\draw[lattice edge] (bot) -- (b);            
		\draw[lattice edge] (bot) -- (c);
		
		\draw[lattice edge] (a) -- (ab);
		\draw[lattice edge] (a) -- (ac);
		
		\draw[lattice edge] (b) -- (ab);
		\draw[lattice edge] (b) -- (bc);
		
		\draw[lattice edge] (c) -- (ac);
		\draw[lattice edge] (c) -- (bc);
		
		\draw[lattice edge] (ab) -- (top);
		\draw[lattice edge] (ac) -- (top);        
		\draw[lattice edge] (bc) -- (top);
		
		\draw[->, line width=0.8pt] (2,1.5) -- (3,1.5);
		
		\begin{scope}[xshift=4cm]
			\node[lattice node, fill=red!25]    (l0) at (0,0) {$\ell=0$};
			\node[lattice node, fill=green!25]  (l1) at (0,1) {$\ell=1$};
			\node[lattice node, fill=orange!30] (l2) at (0,2) {$\ell=2$};
			\node[lattice node, fill=blue!25]   (l3) at (0,3) {$\ell=3$};
			
			\draw[lattice edge] (l0) -- (l1);
			\draw[lattice edge] (l1) -- (l2);
			\draw[lattice edge] (l2) -- (l3);
		\end{scope}
	\end{tikzpicture}
	\caption{Imposing a symmetry of the valuation with respect to all Boolean isomorphisms collapses the Boolean structure of universes into a chain structure.}
	\label{fig:collapse-to-chain}
\end{figure}

This loss of structure is captured by the following collapse result.

\begin{restatable}{lemma}{rankequality}
	\label{lem:atoms-to-ranks}
	Let $\mathcal{L}$ be a universe with $n\ge 2$ atoms and let $V:\mathcal L\to\mathbb C$ satisfy Local Deducibility and universality.
	If there exists $x\in\mathbb C$ such that
	\[
	V(a)=x \qquad \text{for all } a\in\mathrm{Atoms}(\mathcal L),
	\]
	then $V$ is constant on each level of the universe:
	for every $l=1,\dots,n$ there exists $v_l\in\mathbb C$ such that
	\[
	V(s)=v_l
	\]
	for every statement $s$ at level $l$.
\end{restatable}
The proof of the Lemma can be found in Appendix~\ref{proof:atoms-to-ranks}.

Lemma~\ref{lem:atoms-to-ranks} shows that once atomic distinctions are erased, the valuation can distinguish statements only by their level.
From the perspective of the valuation, all statements at the same level become indistinguishable: the lattice structure collapses from a Boolean lattice to a simple chain as represented in FIG.~\ref{fig:collapse-to-chain}.
Universality further implies that this hierarchy carries no intrinsic identity.
Any statement at a given level can be replaced by any other without changing the valuation.
Thus the valuation does not see \emph{which} statement occurs at a given level, but only that it occurs at that level.
After atomic distinctions have been removed, there should therefore remain no hidden freedom: fixing the value of a single statement should fix the valuation completely.
Allowing multiple valuations compatible with the same declared indifference would amount to introducing extra structure not supported by syntax.
Multiple symmetric valuations with the same $V(s)$ would amount to extra \enquote{labels} not fixed by the automorphisms that erase atomic identity.

Consider for example the chain structure in FIG.~\ref{fig:collapse-to-chain}. For the valuation $V$ all atomic statements are identical, thus we represent the structure of the universe as the chain on the right, where at level $l=1$ we have the statement $s$, at level $l=2$ we have the statement $s\vee s$, and at level $l=3$ we have the statement $s\vee s \vee s$. We expect at each level for the valuation to be just dependent on the level of the universe and the valuation of the statement $s$. Thus we expect to have essentially just one degree of freedom, that is the valuation of $s$. Allowing for more degrees of freedom would correspond to asking for some hidden structure. In the present approach, such hidden structure is excluded.
Once atoms are declared identical, the valuation must not retain any residual choice about how their combinations are valued.
This motivates the following rigidity principle.

\begin{principle}[symmetry removes valuational freedom]
	\label{prin:srigidity-under-symmetry}
	Let $\mathcal L$ be a finite Boolean lattice and let $s\in\mathcal L\setminus\{\bot\}$.
	For any $z\in\mathbb C$, there exists a unique valuation $V:\mathcal L\to\mathbb C$ with $V(s)=z$ such that
	$V$ assigns the same value to all atoms of $\mathcal L$ appearing in $s$.
\end{principle}

 Once we ignore distinctions between atoms, there is no further syntactic structure left to exploit, so fixing the value of just one statement already fixes all the rest and uniquely determines the valuation.

\section{Pre-probabilities}
In stating our principles we have not fixed any particular deducibility rule. Instead, we require only general structural constraints on admissible valuations. These constraints already determine a strong algebraic form: every admissible valuation can be reparametrised as an additive function on joins of statements. This is the central result of the paper, as it shows that universal valuations admit a finitely additive representation, and therefore exhibit a structure that partly mirrors that of probability functions.

We can now state the main theorem of the reconstruction: every universal valuation $V$ admits a finitely additive representation.

\begin{restatable}{theorem}{MainTheorem}
	\label{th:sum-rule}
	There exists a bijection $\phi : \mC \to \mC$, which we refer to as a \textbf{pre-probability mapping}, such that $R:\mathcal{L}\to\mC$ defined by $R = \phi \circ V $ satisfies the following conditions:
	\begin{enumerate}
		\item for any set of distinct indices of atomic statements $\{ i_j \}_{j}$,
		\begin{equation}
			R\!\left( \bigvee_{j} s_{i_j} \right) = \sum_{j} R(s_{i_j}),
		\end{equation}
		\item $R(\bot)=0$,
		\item for every $s\in \mathcal{L}$,
		\begin{equation}
			R(\neg s)= R(\top)-R(s).
		\end{equation}
	\end{enumerate}
	We call $R$ a \textbf{pre-probability}.
\end{restatable}
The proofs of the theorem can be found in Appendix~\ref{app:study-of-G}.

Whatever the underlying valuation, we can always map it to a pre-probability that behaves additively.
In particular, we can reason about valuations using the same algebraic rules of ordinary probability. For example, for any given syntactic universe $\mathcal{L}$ with valuation $V$ and pre-probability function $R$, the value of $R$ on any statement of the universe can be inferred from the vector of values of $R$ on the atomic statements
\begin{equation}
	\vec{r}=\left(R(s_1),R(s_2),\dots,R(s_d)\right).
\end{equation}

A central lemma of our main theorem is that, even if a pre-probability representation of $V$ always exists, this is in general not unique. This is stated in the next lemma.

\begin{restatable}{lemma}{afreedom}
	\label{lem:a-freedom}
	The pre-probability mapping $\phi$ is unique up to composition with an additive function, an additive gauge reparametrisation, that we call $\mbox{\textbf{$a$-function}}$, that is a bijection $a:\mC \to \mC$ satisfying the Cauchy functional equation
	\[
	a(x+y)=a(x)+a(y).
	\]
\end{restatable}
The proof of the Lemma can be found in Appendix~\ref{proof:a-freedom}.

The gauge freedom brought by the composition with $a$-function is a distinctive and fundamental property of the language of pre-probabilities. 
To maintain full generality we do not impose continuity or any other regularity assumption on $a$-functions, thus allowing even highly non-constructive $a$-functions~\footnote{In particular, most additive solutions of the Cauchy functional equation can be constructed via a choice of a Hamel basis of $\mathbb{R}$ (or $\mathbb{C}$ viewed as a $\mathbb{Q}$-vector space), which requires the axiom of choice. The axiom of choice is already invoked in the proof of Theorem~\ref{th:sum-rule}.}
.
We adopt this unrestricted setting because the valuation $V$ is taken to be a primitive object. Without additional structure on its codomain, there is a priori no canonical notion of regularity to require. For instance, it is not clear what \enquote{continuity} should mean for $V$ (and hence for admissible reparametrisations) beyond what is encoded by the lattice operations themselves. In this sense, refraining from regularity assumptions keeps the reconstruction purely algebraic. However, for completeness, in Section~\ref{sec:holomorphic} we discuss the effect of imposing additional regularity assumptions. Such assumptions may arise when one requires $V$ to carry further analytic meaning, and we explain how they bear on the results of this work.
 
Because of this gauge freedom, for any fixed universe $\mathcal{L}$ with valuation $V$, there are infinitely many admissible pre-probability mappings $\phi$, and hence infinitely many induced pre-probabilities $R=\phi\circ V$. In general, the pre-probability representation is therefore non-unique, and different choices of $\phi$ are related by composition with an $a$-function.

There is, however, a single exceptional case in which the pre-probability representation is unique. This is the degenerate valuation
\[
V(s)=\varphi \qquad \text{for all } s\in\mathcal{L},
\]
which corresponds to the degenerate classical-logic assignment in which every statement is false. In this case every induced pre-probability is identically zero, and composition with any $a$-function leaves it unchanged. For this reason we call this special case the \emph{invariant valuation}.

This multitude of pre-probabilistic representations for a unique universal valuation is crucial for many of the following results and in particular in the understanding and definition of conditionals and quasi-probabilities.

\subsection{Semantic frames and semantic dimension}
\label{subsec:semantic-frames}

The following subsection introduces additional structure useful for tracking regraduation freedom and comparing different pre-probability representation of the same valuation.

Lemma~\ref{lem:a-freedom} shows that a single valuation may admit many pre-probability representations, all related by composition with $a$-functions. 
All these different representations of the same $V$ share an underlying common invariant structure. For each universe and each valuation it is possible to single out a minimal collection of statements whose pre-probability values provide $\mathbb{Q}$-linearly independent reference values. 
This motivates the notions of semantic frame and semantic dimension.
The interpretation and broader role of semantic frames will be discussed in Sec.~\ref{sec:higher-semantic-dimension}. 

\begin{definition}[Semantic frames and semantic dimension]
	\label{def:gauge-frame-dim}
	Let $\mathcal{L}$ be a universe and let $R:\mathcal{L}\to\mathbb{C}$ be a pre-probability function.
	
	A finite set
	\[
	\SL{R}=\{s_1,\dots,s_m\}\subseteq\mathcal{L}
	\]
	is called a \textbf{semantic frame} for $(\mathcal{L},R)$ if it is a maximal set of statements such that for all rationals $q_1,\dots,q_m\in\mathbb{Q}$,
	\[
	\sum_{j=1}^{m} q_j\,R(s_j)=0
	\quad\Longrightarrow\quad
	q_1=\cdots=q_m=0 .
	\]
	The integer $m:=|\SL{R}|$ is called the \textbf{semantic dimension} of $(\mathcal{L},R)$.
\end{definition}

Semantic frames individuate the degree of freedom over which the $a$-functions act and the semantic dimension $m$ individuates the number of degrees of freedom for the gauge freedom of a valuation on a universe. We point out some properties of semantic frames in the following lemma. 

\begin{restatable}{lemma}{linindipendenti}
	\label{lem:ref-frame-from-lin-indep}
	Given a universe $\mathcal{L}$ with $n$ atomic statements, and a pre-probability $R:\mathcal{L}\to\mathbb{C}$, consider any $a$-function and a new pre-probability $R':=a\circ R$. Then:
	\begin{enumerate}
		\item If $s\in \SL{R}$ then $R(s)\neq 0$.
		\item All semantic frames $\SL{R}$ have the same semantic dimension.
		\item All semantic frames are independent of $a$. They depend only on the valuation $V$.
		\item There is always a semantic frame composed only of atomic statements and thus the semantic dimension $m$ is bounded as $m\le n$.
		\item \label{item:R-restriction-extension} If $R'$ satisfies
		\[
		R'(s_j)=R(s_j)
		\quad\text{for all } j=1,\dots,m,
		\]
		then
		\[
		R'(s)=R(s)
		\quad\text{for all } s\in\mathcal{L}.
		\]
	\end{enumerate}
\end{restatable}
The proof of the Lemma can be found in Appendix~\ref{proof:ref-frame-from-lin-indep}.

Property number 5. points out the idea that semantic frames individuate the degree of freedom of a pre-probabilistic representation. If two pre-probabilistic representations agree on all the elements of a semantic frame for a given universe, then these pre-probabilities coincide, thus any non-trivial $a$-function would act by changing the valuation of at least one element of a semantic frame.

\subsection{Embedding sub-universes in ambient universes}
\label{sec:communicating}

Consider two agents $X$ and $Y$, each equipped with a syntactic universe $\mathcal{L}_X$ and $\mathcal{L}_Y$, who wish to communicate the valuation of their statements. The very possibility of communication presupposes a larger universe $\mathcal{L}$, an \emph{ambient universe} in which both sub-universes are jointly situated. In other words, each agent may discover that what they initially took to be “the” universe is in fact only a sub-universe, and that there exists an ambient universe $\mathcal{L}$ into which both $\mathcal{L}_X$ and $\mathcal{L}_Y$ embed~\footnote{We remind the reader that we restricted our attention to settings in which all local universes arise from a single ambient universe, and their valuations are obtained as restrictions of a global valuation.};
 this is another instance of Syntactic Locality.

Now suppose that the two agents $X$ and $Y$ are equipped with pre-probabilities $R_{X}$ and $R_{Y}$ for their respective (sub-)universes $\mathcal{L}_X$ and $\mathcal{L}_Y$. We refer to their pre-probabilities as \textit{local pre-probabilities}, encoding in this definition the knowledge that their pre-probabilities are just a pre-probability mapping of the restrictions $V|_X$ and $V|_Y$ of the universal valuation on $\mathcal{L}$ to the sub-universes $\mathcal{L}_X$ and $\mathcal{L}_Y$.
In the sub-case when the sub-universe of an agent is a relative sub-universe $\mathcal{L}_s$, with $s\in \mathcal{L}$, we will refer to its local pre-probability $R_s$ as a \textit{relative pre-probability}.

For the two local agents $X$ and $Y$, the act of communicating the pre-probabilistic valuation of the statements of their own sub-universes therefore amounts to embedding these sub-universes into the ambient universe $\mathcal{L}$ in a coherent way, so that the valuation $R_X$ and $R_Y$ assigned locally agrees with the valuation $R$ of the corresponding statements viewed in the ambient universe. The central difficulty, then, is to find such a coherent embedding. In particular, coherence requires that the two agents first \emph{synchronise} their choice of local pre-probability representation, that is, they must ensure they are using the same pre-probability mapping.

\begin{definition}[Synchronisation of pre-probability valuations]
	Let $\mathcal{L}_X$ and $\mathcal{L}_Y$ be syntactic sub-universes of an ambient universe $\mathcal{L}$, with valuation $V$ on $\mathcal{L}$ and induced restrictions $V|_{X}$ and $V|_{Y}$.  
	Let $R_X$ and $R_Y$ be the local pre-probabilities used by agents $X$ and $Y$ on $\mathcal{L}_X$ and $\mathcal{L}_Y$.
	
	We say that agents $X$ and $Y$ \textbf{synchronise} their local pre-probabilities (relative to $\mathcal{L}$) if they define two a-functions $a_X$ and $a_Y$ such that their new  local pre-probabilities
	\[
		R'_X\coloneqq a_X \circ R_X, \qquad R'_Y \coloneqq a_Y \circ R_Y
	\]
	are such that
	\[
	R'_X=R|_{\mathcal{L}_X}, \qquad 	R'_Y=R|_{\mathcal{L}_Y},
	\]
	for some pre-probability $R$ of the ambient lattice $\mathcal{L}$.
\end{definition}

\begin{figure}[t]
	\centering
	\includegraphics[width=1.02\linewidth]{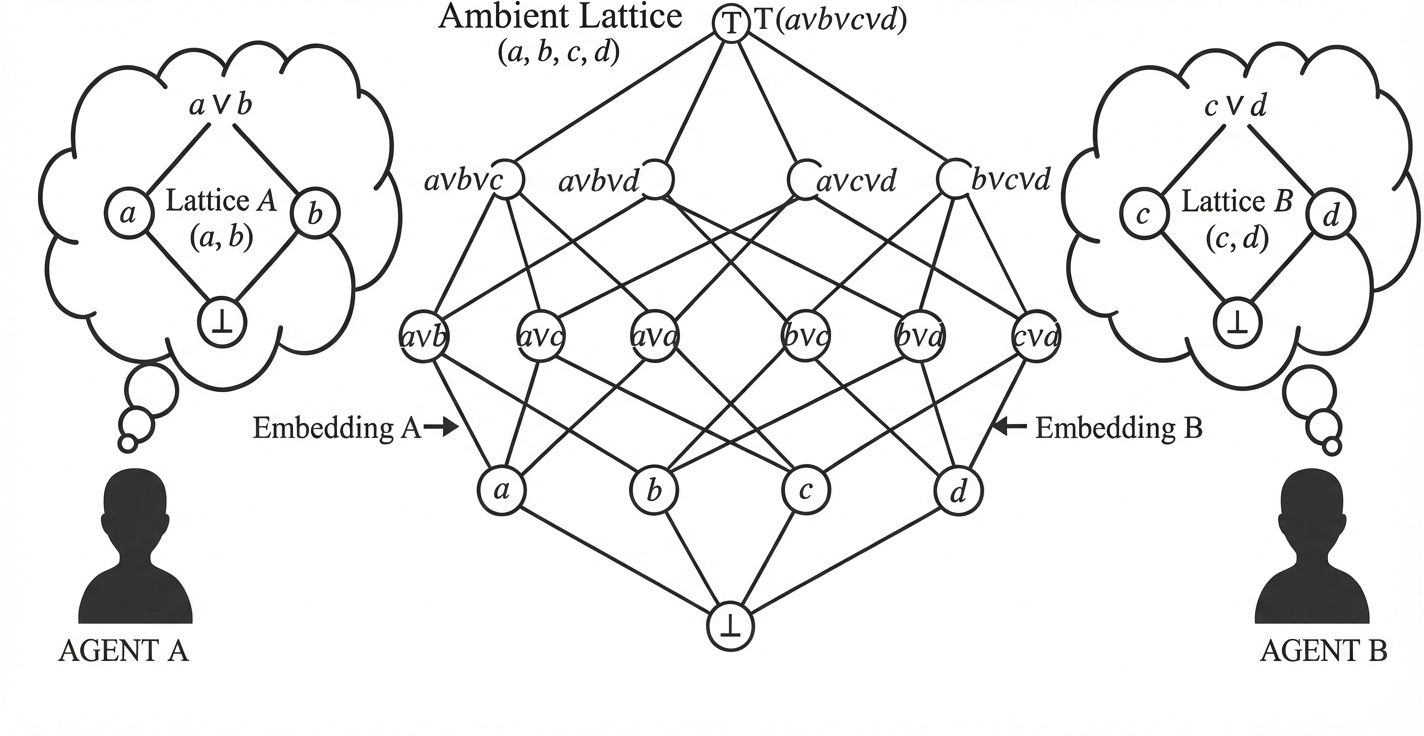}
	\caption{Two agents' sub-universes embedded in an ambient shared universe.}
	\label{fig:Boolean-embedding}
\end{figure}

If two agents can communicate we must suppose that they are indeed embedded in a common syntactic universe, thus that their valuations are synchronisable in principle. However, this just means that the two a-functions necessary for the synchronisation exist, there is no guarantee the two agents can find them. We ask what are the resources the two agents need to be able to synchronise their valuations.

If the agents have some knowledge of the value of the pre-probability $R$ on the ambient universe they can tune their local pre-probability representations to match that value. 
From point \ref{item:R-restriction-extension} of lemma~\ref{lem:ref-frame-from-lin-indep} one can directly prove the following lemma
\begin{lemma}
\label{lem:sychronise}
For an agent $A$, it is enough to know the ambient valuation of all the statements in any semantic frame
$\mathcal{S}_{\mathcal{L}_X}(R_X)$ in order to synchronise its local pre-probability $R_X$ with the pre-probability $R$ on the ambient universe.
\end{lemma}

\subsection{The Rule of Total Pre-Probability}
\label{sec:rtp-for-ppf}

Beyond synchronising the local pre-probabilities of distinct sub-universes, there are situations in which several agents, each working within a local universe, can collectively reconstruct the ambient universe together with its ambient pre-probability. This occurs when the agents' (sub-)universes jointly contain all atoms of the ambient universe.

An interesting case is when the sub-universes of the local agents are relative universes $\mathcal{L}_{s_j}$ of an ambient universe $\mathcal{L}$, such that
\[
s_{j_1}\wedge s_{j_2}=\bot \quad (j_1\neq j_2),
\qquad
\bigvee_j s_j=\top.
\]
Each agent is equipped with its relative pre-probability
\[
R_{s_j}:\mathcal L_{s_j}\to \mC
\]
with domain on its relative universe $\mathcal L_{s_j}$. We denote by $a_j(\cdot)$  the $a$-function that synchronises $R_{s_j}$ with the ambient pre-probability $R$.

\noindent
We are interested in reconstructing $R$ from the collection of relative pre-probabilities $R_{s_j}$. In order to do so, it is useful to define the \textit{localisation map} $[\cdot]_{s_j}:\mathcal L\to \mathcal L_{s_j}$ as 
\[
[s]_{s_j}\coloneqq s\wedge s_j.
\]
The localisation map allows us to extend the domain of the relative pre-probabilities $R_{s_j}$ to the whole ambient universe. We denote this extension with $R(s\mid s_j):\mathcal{L}\to \mC$ defined as
\begin{equation}
R(s\mid s_j)\coloneqq R_{s_j}\bigl([s]_{s_j}\bigr),
\end{equation}
and we call it \textit{conditional pre-probability}.

With these tools we can prove the following identity through which the agents become able to reconstruct the pre-probability $R$ of the whole ambient universe.

\begin{restatable}{proposition}{RTPP}
\label{prop:rule-total-pre-prob}
Under the assumptions above, the \textbf{rule of total pre-probability} is the identity
\begin{equation}
\label{eq:rule-total-pre-prob}
R(s)=\sum_j a_j\!\left(R(s\mid s_j)\right),
\end{equation}
where $a_j(\cdot)$ are the $a$-functions that synchronise the pre-probabilities of the agents.
\end{restatable}
The proof of the identity in the Proposition can be found in Appendix~\ref{proof:rule-total-pre-prob}.

\subsection{Generalised Bayes' Theorem}
In standard probability theory Bayes' theorem is used to \enquote{invert} conditional probabilities. It is a direct consequence of the standard definition of conditional probability and it is often written as
\begin{equation}
\label{eq:standard-Bayes}
\mathcal{P}(B=b\mid A=a)=\frac{\mathcal{P}(A=a\mid B=b)\mathcal{P}(B=b)}{\mathcal{P}(A=a)}.
\end{equation}

In the context of pre-probability one can derive a similar equation simply considering the synchronisation of different conditional pre-probabilities. The existence of a universal valuation plays a crucial role for the generalised Bayes' theorem. In fact different conditional pre-probabilities are indeed extensions via localisation map of restrictions to sub-universes of the same universal valuation $V$ possibly with different parametrisations of the pre-probabilities due to different choices of pre-probability mappings. The generalised Bayes' theorem is nothing more than an explicit recipe to synchronise conditional pre-probabilities with different representations. 

To display explicitly an identity analogous to~\eqref{eq:standard-Bayes}, it is convenient to work in a universe $\mathcal{L}_4$ with four atomic statements $\{a,b,c,d\}$. For readability, we will freely switch between the notation $\{a,b,c,d\}$ and the $AB$-labelling introduced in Appendix~\ref{sec:AB-labelling}. The notation $\{a,b,c,d\}$ is useful to think in terms of Boolean lattices, the notation with $AB$-labelling is useful to compare with the usual presentations of these concepts in standard probability theory.

Consider the two conditional pre-probabilities 
\begin{equation}
	\label{eq:conditionals}
	R(B=i_b \mid A=0), \quad R(A=i_a \mid B=1).
\end{equation}
From the definition of conditional pre-probabilities these correspond to
\[
	R_{a\vee b}(x), \quad R_{b\vee d}(y),
\]
where we restricted the domains to $x\in\{a,b\}$ and $y\in\{b,d\}$ to match the domains of the two conditional pre-probabilities in equation~\eqref{eq:conditionals}, and $R_{a\vee b},R_{b\vee d}$ are the relative pre-probabilities for two agents respectively equipped with the relative sub-universes $\mathcal{L}_{a\vee b}$ and $\mathcal{L}_{b\vee d}$.
Since we know that the universes of the two agents are actually sub-universes of $\mathcal{L}_4$, we know that their pre-probabilities can be synchronised to be the restrictions of the ambient universe's pre-probability $R$ to their respective sub-universes. In particular we know that there exist two $a$-functions $a_{A}$ and $a_{B}$ such that 
\[
a_A\circ R_{a\vee b} = R|_{a\vee b}, \quad a_B\circ R_{b \vee d}= R|_{b \vee d}.
\]
This implies that, using the AB-labelling, we can write
\[
	a_A(R(B=1\mid A=0))=a_B(R(A=0\mid B=1)).
\]
In a form reminiscent of Bayes' theorem, we can write an instance of what can be called  \textbf{generalised Bayes' theorem}
\begin{equation}
	R(B=1 \mid A=0)=a_A^{-1}(a_B(R(A=0\mid B=1))).
\end{equation}
We see that "inverting the conditionals" corresponds to synchronising conditional pre-probability representations. We know that both conditional pre-probabilities come from the same underlying universal valuation $V$, but they are in general expressed with two different representations as pre-probabilities. The generalised Bayes' theorem then is nothing more than a synchronisation of pre-probabilities corresponding to different agents each equipped with a representation of the universal valuation $V$ on a restriction of the ambient universe.

In what follows we show that, when specialised to classical probabilities, the theorem reduces exactly to Bayes' theorem.
Viewed in this way, Bayes' theorem will be easily applied to quasi-probabilities without running into the usual difficulties: since the result is essentially a synchronisation identity, it is always well-defined, even when some quasi-probabilities have value zero.

  \section{Quasi-probabilities}

\subsection{Restricting to universes with semantic dimension 1}

In what follows we restrict to universes $\mathcal{L}$ of semantic dimension $m=1$, leaving the discussion of higher semantic dimension to Sec.~\ref{sec:higher-semantic-dimension}. 
This restriction greatly simplifies the structure of admissible pre-probabilities. In semantic dimension $m=1$ all values of a pre-probability are $\mathbb{Q}$-proportional to any fixed nonzero value, as stated in the next lemma.

\begin{restatable}{lemma}{qpropsemanticone}
	\label{lem:q-proportionality}
	Let $\mathcal{L}$ be a universe of semantic dimension $1$ and let $R:\mathcal{L}\to\mathbb{C}$ be a pre-probability function.
	Then for every choice of $\hat{s}\in\mathcal{L}$ with $R(\hat{s})\neq 0$ and for every $s\in\mathcal{L}$ there exists a rational number $q(s)\in\mathbb{Q}$ such that
	\begin{equation}
		\label{eq:q-proportionality}
	R(s)=q(s)\,R(\hat{s}).
	\end{equation}
	Equivalently, all nonzero values of $R$ are proportional to one another with rational proportionality constants.
\end{restatable}
The proof of the Lemma can be found in Appendix~\ref{proof:q-proportionality}.

In particular, once a single nonzero reference value $R(\hat{s})$ is fixed, the entire pre-probability is determined by a rational function $q(\cdot)$ on $\mathcal{L}$.
Moreover, in this setting the freedom coming from composition with $a$-functions collapses to an overall scalar degree of freedom, as stated in the following lemma.

\begin{restatable}{lemma}{semanticquasiprob}
	\label{lem:semanticquasiprob}
	Let $\mathcal{L}$ be a universe of semantic dimension $m=1$, and let $R$ be a corresponding pre-probability function.  
	Then every $a$-function is necessarily a scalar multiplication: there exists a unique complex number $\alpha_a\in\mC$ such that
	\[
	a\circ R \;=\; \alpha_a\, R,
	\]
	i.e.\ for all $s\in\mathcal{L}$ one has $a(R(s))=\alpha_a R(s)$.
\end{restatable}
The proof of the Lemma can be found in Appendix~\ref{proof:semanticquasiprob}.

 \subsection{Canonical representation}

We can now invoke Lemma~\ref{lem:q-proportionality} to select a convenient canonical representative for pre-probabilities of semantic dimension~$m=1$.
By canonical we mean that the non-uniqueness induced by the $a$-freedom is eliminated: among all pre-probability representations related by composition with $a$-functions, we single out one that is uniquely defined.

To specify this canonical representative unambiguously, we observe that every universe contains a distinguished statement, namely the unique top element~$\top$. We adopt this element as our reference in the sense of Lemma~\ref{lem:q-proportionality} \footnote{Choosing $\top$ as the reference statement is natural because every universe has a top element, and it does not single out any particular statement from any particular universe. Any alternative choice would depend on the specific universe. For instance, requiring a specific statement $s\in\mathcal{L}$ to have pre-probability valued $1$ is not transferable to universes in which $s$ is absent. Likewise, requiring \enquote{the leftmost atomic statement} to have pre-probability valued $1$ is not invariant under Boolean lattice isomorphisms, which may permute the atomic statements.}.
In fact, whenever $R(\top)\neq 0$ we may take $\hat{s}=\top$ in Lemma~\ref{lem:q-proportionality} and describe the pre-probability valuation via the rational function $q$.

In this case we can use the remaining freedom coming from $a$-functions to \emph{fix} the value at the top element, thereby eliminating the ambiguity in the choice of pre-probability representation.
Concretely, we choose an $a$-function (equivalently, in semantic dimension~$m=1$, a scalar $\alpha_a$) and pass to the reparametrised pre-probability
\[
R' \coloneqq a\circ R \qquad (\text{equivalently } R'=\alpha_a R).
\]
Without loss of generality we impose~\footnote{This is just a convention: any nonzero value could be chosen in place of $1$, for instance $R'(\top)=-\sqrt{2}+i\pi$, and the discussion would be unchanged. The only value that cannot be imposed is $R'(\top)=0$, since no $a$-function can map a nonzero number to $0$. We choose $R'(\top)=1$ to align the canonical representative with the standard quasi-probability and probability normalisation.}
\[
R'(\top)=1.
\]
With this choice, Lemma~\ref{lem:q-proportionality} immediately gives $R'(s)\in\mathbb{Q}$ for every $s\in\mathcal{L}$.
More importantly, the condition $R'(\top)=1$ removes the $a$-freedom: there is no remaining freedom coming from composition with $a$-functions, and one obtains a single well-defined representative.
We call \textit{quasi-probability} this uniquely fixed canonical pre-probability.

\begin{definition}[Quasi-probability]
	\label{def:qpf}
	A \textbf{rational quasi-probability function} (in the following simply a \textbf{quasi-probability}) on a universe $\mathcal{L}$ is a pre-probability $Q:\mathcal{L}\to\mathbb{Q}\subseteq\mathbb{C}$ such that:
	\begin{enumerate}
		\item For any set of distinct indices $\{ i_j \}_j$ corresponding to atomic statements,
		\begin{equation}
			Q\!\left( \bigvee_{j} s_{i_j} \right) = \sum_{j} Q(s_{i_j}) .
		\end{equation}
		
		\item For every $s\in\mathcal{L}$,
		\begin{equation}
			Q(\neg s)=Q(\top)-Q(s).
		\end{equation}
		
		\item Either $Q$ is normalised over the atoms,
		\begin{equation}
			\sum_{s\in\mathrm{Atoms}(\mathcal{L})} Q(s)=1,
		\end{equation}
		or $Q$ is identically zero, i.e.\ $Q(s)=0$ for all $s\in\mathcal{L}$.
	\end{enumerate}
\end{definition}

There is not much difference between a pre-probability and a quasi-probability: a quasi-probability is simply the canonical representative within the $a$-freedom of pre-probabilities, obtained by fixing the value at the top element.

Notice that this canonical choice is only available when the top element is nonzero, or when one has the invariant valuation. Indeed, every $a$-function satisfies $a(0)=0$, so if a pre-probability $R$ has
\[
R(\top)=0,
\]
then no reparametrisation $a\circ R$ can satisfy ${(a\circ R)(\top)=1}$.
In this case there is no way to define a normalised quasi-probability associated with $R$, unless $R$ is identically zero (the case of the invariant valuation).
In this degenerate case the $a$-freedom is already trivial, one has $a\circ R=R$ for every $a$-function, so there is again a single well-defined representative with no residual freedom.
For this reason we include $Q\equiv 0$ as part of the canonical quasi-probability representation.
All the other cases one cannot have a quasi-probabilistic representation of the universal valuation $V$ and is restricted to pre-probabilities.

\section{Conditional quasi-probabilities}

\subsection{Standard conditional quasi-probabilities as extension of conditional probabilities and their problems}
In the standard construction of quasi-probabilities, the usual notion of conditioning can fail. Indeed, the standard approach extends classical probabilities beyond the interval $[0,1]$ while keeping the remaining definitions formally unchanged. 

Kolmogorov defines conditional probability as (see Section~1.4 of~\cite{kolmogorov1933})
\begin{equation}
	\mathcal{P}(B \mid A)\coloneqq \frac{\mathcal{P}(A,B)}{\mathcal{P}(A)}.
\end{equation}
Given a standard quasi-probability $\mathcal{Q}$, one then defines conditional quasi-probabilities by the same formula,
\begin{equation}
	\mathcal{Q}(B \mid A)\coloneqq \frac{\mathcal{Q}(A,B)}{\mathcal{Q}(A)}.
\end{equation}
However, this definition is not always meaningful, as the following example illustrates.

\begin{example}
	Consider two binary variables $A,B\in\{0,1\}$ and a quasi-probability distribution $\mathcal{Q}(A,B)$ given by
	\begin{align*}
		\mathcal{Q}(A=0,B=0)&=0.3, \quad &\mathcal{Q}(A=0,B=1)&=-0.3,\\
		\mathcal{Q}(A=1,B=0)&=0.6, \quad &\mathcal{Q}(A=1,B=1)&=0.4.
	\end{align*}
	The corresponding marginals are
	\begin{align*}
		\mathcal{Q}(A=0)&=0, \qquad &\mathcal{Q}(A=1)&=1,\\
		\mathcal{Q}(B=0)&=0.9, \qquad &\mathcal{Q}(B=1)&=0.1.
	\end{align*}
	Using the standard conditioning rule, we obtain for example
	\begin{align*}
		\mathcal{Q}(A=0 \mid B=0)&=\frac{0.3}{0.9}=\frac{1}{3}, \\
		\mathcal{Q}(A=1 \mid B=0)&=\frac{0.6}{0.9}=\frac{2}{3},\\
		\mathcal{Q}(B=0 \mid A=1)&=\frac{0.6}{1}=0.6, \\
		\mathcal{Q}(B=1 \mid A=1)&=\frac{0.4}{1}=0.4.
	\end{align*}
	But conditioning on $A=0$ is not defined, because $\mathcal{Q}(A=0)=0$:
	\begin{align*}
		\mathcal{Q}(B=0 \mid A=0)\ \text{is undefined},\\
		\mathcal{Q}(B=1 \mid A=0)\ \text{is undefined}.
	\end{align*}
	This is problematic: for instance, the joint value $\mathcal{Q}(A=0,B=0)=0.3$ is perfectly well-defined, yet the corresponding conditional quantities cannot be formed.
	
	In classical probability theory, events with ${\mathcal{P}(A)=0}$ can be handled consistently, so the breakdown above shows that importing the ratio definition of conditioning into the quasi-probabilistic setting can lead to genuine inconsistencies or undefined expressions.
\end{example}

\subsection{Instability of relativisation of quasi-probability and conditional quasi-probabilities}

With our reconstruction, the issues with quasi-probabilities and conditionals highlighted in the previous subsection disappear. Indeed, in our framework conditional quasi-probabilities are not obtained by extending classical conditional probabilities, but they are obtained as restrictions of the notion of conditional pre-probability, which is always well-defined. Within our reconstruction we understand that the failure of the standard definition of conditional quasi-probabilities comes from the fact that quasi-probabilities are not stable under relativisation in the following sense.

For pre-probabilities, restriction to sub-universes is well behaved: if a universe $\mathcal L$ is equipped with a pre-probability $R$, then the restriction of $R$ to the domain of any sub-universe is again a pre-probability.

This stability fails for quasi-probabilities. If $\mathcal L$ carries a quasi-probability $Q$, then the restriction of $Q$ to a relative universe $\mathcal L_{t}$ is, in general, not a valid quasi-probability, because quasi-probabilities must be normalised at the local top element $t$. Concretely, if $Q(t)=0$, then there is no $a$-function $a$ such that $a\!\bigl(Q(t)\bigr)=1$, and therefore the induced valuation on $\mathcal L_{t}$ cannot be parametrised as a quasi-probability. In this case one must instead work with the corresponding relative pre-probability on $\mathcal L_{t}$.

If, however, $Q(t)\neq 0$, the restriction can always be normalised to a quasi-probability by composing with the scalar $a$-function $a_{t}$ whose multiplicative factor is
\[
\alpha_{a_{t}}=\frac{1}{Q(t)}.
\]
This motivates the following definitions.

\begin{definition}[Relative quasi-probability]
\label{def:rel-quas-prob}
Let $\mathcal L$ be a universe equipped with a quasi-probability $Q$, and let $t \in\mathcal L$ satisfy $Q(t)\neq 0$. The \textbf{relative quasi-probability} at $t$ is the normalised restriction of $Q$ to the relative universe $\mathcal L_{t}$, defined by
\begin{equation}
\label{eq:relative-quasi-probability}
Q_{t}(s)\coloneqq \frac{Q(s)}{Q(t)}
\qquad (s\in \mathcal L_{t}).
\end{equation}
Equivalently, $Q_{ts}=a_{t}\circ \bigl(Q|_{\mathcal L_{t}}\bigr)$ with $a_{t}(z)=z/Q(t)$.
\end{definition}

\begin{definition}[Conditional quasi-probability]
	\label{def:cond-quas-prob}
	Let $Q_t$ be the relative quasi-probability at $t$. For any $s\in\mathcal L$, the \textbf{conditional quasi-probability} of $s$ given $t$ is defined by
	\begin{equation}
		\label{eq:cond-quasi-prob}
		Q(s\mid t)\coloneqq Q_t\bigl([s]_t\bigr).
	\end{equation}
\end{definition}

\subsection{The rule of total quasi-probability}

The rule of total pre-probability can be specialised to quasi-probabilities, making the synchronising $a$-functions explicit. However, because quasi-probabilities are not stable under relativisation, in some situations one is forced to revert to pre-probabilities. For this reason we distinguish three cases.

\begin{enumerate}
	\item \emph{Stable case:} the quasi-probability is stable under relativisation, so the analysis can be carried out entirely within quasi-probabilities.
	\item \emph{Mixed case:} conditioning (relativising) a quasi-probability yields a pre-probability that cannot be reparametrised to quasi-probability. This leads to a mixed rule of total quasi-/pre-probability, and it is precisely the regime in which the standard treatment of conditional quasi-probabilities breaks down.
	\item \emph{Invariant sub-universe:} the situation in which the valuation on a sub-universe is the invariant valuation.
\end{enumerate}

For each case we provide an explicit example in the $AB$-labelling.

\subsubsection{Case 1: Quasi-probability stable under relativisation}

We consider a universe $\mathcal L$ equipped with a quasi-probability $Q$, and a family of relative universes $\{\mathcal L_{s_j}\}_j$ such that
\[
 s_j\wedge  s_k=\bot \quad (j\neq k),
\qquad
\bigvee_j  s_j=\top,
\]
so that
\[
\sum_j Q( s_j)=Q(\top)=1.
\]
Assume that each relative universe $\mathcal L_{s_j}$ is equipped with a quasi-probability
$Q_{s_j}:\mathcal L_{\tilde s_j}\to \mC$. Thus simply composing with the localisation map we have a family of conditional quasi-probabilities $\{Q(s\mid s_j)\}_j$.

\begin{proposition}[Rule of total quasi-probability]
\label{def:rule-total-quasi-prob}
Under the assumptions above, the \textbf{rule of total quasi-probability} is the identity
\[
Q(s)=\sum_{j} Q(s\mid \tilde s_j)\,Q(\tilde s_j).
\]
\end{proposition}

\begin{example}
\label{ex:rtp-1}
We now consider a concrete instance of the rule of total quasi-probability in the ambient universe $\mathcal L_{4}$. Using the AB-labelling (see appendix~\ref{sec:AB-labelling}) the rule of total quasi-probability yields
\begin{align*}
Q((B=i_b))
&=Q(B=i_b\mid A=0)\,Q((A=0))+\nonumber\\
&+Q(B=i_b\mid A=1)\,Q((A=1)),\nonumber\\
Q((A=i_a))
&=Q(A=i_a\mid B=0)\,Q((B=0))+\nonumber\\
&+Q(A=i_a\mid B=1)\,Q((B=1)),
\end{align*}
which has the same formal structure as the standard rule of total probability.
\end{example}

\subsubsection{Case 2: Quasi-probability unstable under relativisation}
We now consider the case in which one selected relative universe,  say $\mathcal L_{s_k}$, is not equipped with a quasi-probability, but only with a pre-probability $R_{s_k}$ satisfying $R_{s_k}(\tilde s_k)=0$. To synchronise this local valuation with the ambient quasi-probability $Q$, one must choose a statement $s'_k\in \mathcal L_{s_k}$ such that $Q(s'_k)\neq 0$ (and hence $R_{s_k}(s'_k)\neq 0$). In this situation, we obtain a mixed rule of total quasi/pre probability of the form
\begin{equation*}
\label{eq:total-mixed}
Q(s)
=
\sum_{j\neq k} Q(s\mid s_j)\,Q(s_j)
+
R(s\mid s_k) \frac{Q(s'_k)}{R_{s_k}(s'_k)}.
\end{equation*}

\begin{example}
\label{ex:rtp-2}
Consider the same universe as in Example~\ref{ex:rtp-1}, but suppose now that $Q(A=1)=0$ while $Q((A=1,B=1))\neq 0$.
 
In this case the rule of total pre-probability yields
\begin{align*}
Q((B=i_b))
&= Q(B=i_b\mid A=0)\,Q((A=0)) \nonumber\\
& +R(B=i_b\mid A=1)\,
\frac{Q((A=1,B=1))}{R(B=1\mid A=1)}.
\end{align*}
\end{example}

\subsubsection{Case 3: Quasi-probability with sub-universe with invariant valuation}
As last case we examine the case in which the relative universe $\mathcal L_{s_k}$ is not equipped with a quasi-probability, but only with a pre-probability $R_{s_k}$ such that $R_{s_k}(s)=0$ for all $s\in \mathcal{L}_{s_k}$. 
Since the composition with $a$-function leaves $R_{s_k}$ invariant, then $R_{s_k}$ must be already synchronised, thus the rule of total pre-probability takes the form
\begin{align*}
\label{eq:total-mixed-0}
Q(s)
&=\sum_{j\neq k} Q(s\mid s_j) Q(s_j) + R(s\mid s_k)= \nonumber\\
&=\sum_{j\neq k} Q(s\mid s_j) Q(s_j).
\end{align*}
\begin{example}
Consider the same universe as in Example~\ref{ex:rtp-2}, but suppose now that 
\begin{align*}
Q((A=1))&=Q((A=1,B=1))=\nonumber \\
&=Q((A=1,B=0))= 0
\end{align*}. 
This implies that $Q((A=0))=1$ and in this case we have:
\begin{align*}
Q((B=i_b))
&= Q(B=i_b\mid A=0).
\end{align*}
\end{example}

\subsection{Generalised Bayes' theorem for quasi-probabilities}
Having a clear and coherent definition of conditional quasi-probabilities allows us to formulate a version of Bayes' theorem for quasi-probability.

As we saw in the case of pre-probability, Bayes' theorem is essentially a synchronisation principle for different conditional valuations. Many of the familiar difficulties with quasi-probabilities arise because they are not, in general, stable under relativisation. In other words, if we pass to a relative universe within a universe equipped with a quasi-probability, the resulting relative valuation, and hence the associated conditional valuation, may no longer be representable as a conditional quasi-probability. However, since we already know how to synchronise pre-probabilities, and we have a generalised Bayes' theorem for them, this obstacle disappears.

In what follows, as in the previous section, we present the formalisation of Bayes' theorem for quasi-probability in three cases. First, when the quasi-probability is stable under relativisation, we obtain a Bayes' theorem that is entirely analogous to the standard one. Second, when the quasi-probability on the universe under consideration is not stable, we must use a mixed form of Bayes' theorem that combines quasi-probabilities with pre-probabilities. Both of these forms are simply a specialisation of the generalised Bayes' theorem for pre-probabilities for these specific cases. As a third case we present the case with the invariant valuation.

\subsubsection{Case 1: Quasi-probability stable under relativisation}

Consider a universe $\mathcal{L}_4$ with four atomic statements $\{a,b,c,d\}$ for which the universal valuation $V$ admits a quasi-probability representation $Q$. For readability, we will freely switch between the notation $\{a,b,c,d\}$ and the $AB$-labelling introduced in Appendix~\ref{sec:AB-labelling}.

We assume $Q((A=0))=Q(a\vee b)\neq 0$ and $Q((B=1))=Q(b\vee d)\neq 0$.
Consider the two conditional quasi-probabilities 
\begin{equation}
	\label{eq:quasi-conditionals}
	Q(B=i_b \mid A=0), \quad Q(A=i_a \mid B=1).
\end{equation}
From the definition of conditional quasi-probabilities these correspond to
\[
Q_{a\vee b}(x), \quad Q_{b\vee d}(y),
\]
where we restricted the domains to $x\in\{a,b\}$ and $y\in\{b,d\}$ to match the domains of the two conditional quasi-probabilities in equation~\eqref{eq:quasi-conditionals}, and $Q_{a\vee b},Q_{b\vee d}$ are the relative quasi-probabilities for two agents respectively equipped with the relative sub-universes $\mathcal{L}_{a\vee b}$ and $\mathcal{L}_{b,d}$.
Since we know that the universes of the two agents are actually sub-universes of $\mathcal{L}_4$, we know that their quasi-probabilities can be synchronised to be the restrictions of the ambient universe's quasi-probability $Q$ to their respective sub-universes. In particular we know that there are two $a$-functions $a_{A}$ and $a_{B}$ such that 
\[
a_A\circ Q_{a\vee b} = Q|_{a\vee b}, \quad a_B\circ Q_{b \vee d}= Q|_{b \vee d}.
\]
This implies that, with AB-labelling, we can write
\[
a_A(Q(B=1\mid A=0))=a_B(Q(A=0\mid B=1)).
\]
In particular we know that $a_A(x)=Q((A=0))\cdot x$ and $a_B(x)=Q((B=1))\cdot x$.
Using the explicit form of the synchronising $a$-functions, the generalised Bayes' theorem
\begin{align*}
	Q(B=1 \mid A=0)=a_A^{-1}(a_B(Q(A=0\mid B=1))),
\end{align*}
for quasi-probabilities stable under relativisation takes the form
\begin{align*}
	Q(B=1 \mid A=0)=\frac{Q(A=0\mid B=1)Q((B=1))}{Q((A=0))}.
\end{align*}
That is exactly the form one has for standard Bayes' theorem.

\subsubsection{Case 2: Quasi-probabilities unstable under relativisation}
We consider the same setting as in Case~$1$, but now we assume that $Q(A=0)=Q(a\vee b)=0$.
In this situation we no longer have two conditional quasi-probabilities. Instead, we obtain one conditional pre-probability and one conditional quasi-probability,
\begin{equation}
	\label{eq:mixed-conditionals}
	R(B=i_b \mid A=0), \quad Q(A=i_a \mid B=1).
\end{equation}

Proceeding exactly as in the previous section, the generalised Bayes' theorem
\begin{align*}
	R(B=1 \mid A=0)=a_A^{-1}\bigl(a_B(Q(A=0\mid B=1))\bigr)
\end{align*}
simplifies to
\begin{align*}
	R(B{=}1\mid A{=}0)=a_A^{-1}\bigl(Q(A{=}0\mid B{=}1)\,Q(B{=}1)\bigr).
\end{align*}

Suppose there exists a statement $\tilde{s}\in\mathcal{L}_{a}$ such that $Q(\tilde{s})\neq 0$. Then
\begin{equation*}
	a_A^{-1}(x)=\frac{R_{a\vee b}([\tilde{s}]_{a\vee b})}{Q(\tilde{s})}\,x.
\end{equation*}
Equivalently, in the $AB$-labelling, suppose there is a value $\tilde{i}_b$ of $B$ such that $Q(A{=}0,B{=}\tilde{i}_b)\neq 0$. In that case,
\begin{equation*}
	a_A^{-1}(x)=\frac{R(B{=}\tilde{i}_b\mid A{=}0)}{Q(A{=}0,B{=}\tilde{i}_b)}\,x.
\end{equation*}

Therefore, in the mixed pre-probability/quasi-probability setting, the generalised Bayes' theorem becomes
\begin{equation*}
	R(B{=}1\!\mid\!A{=}0)\!=\!Q(A{=}0\!\mid\!B{=}1)\,Q((B{=}1))\,\tfrac{R(B{=}\tilde{i}_b \mid\!A{=}0)}{Q(A{=}0,B{=}\tilde{i}_b)}.
\end{equation*}
Of course, when $A$ and $B$ each have only two possible outcomes, the mixed pre-probability/quasi-probability case is essentially of limited practical use. However, once more outcomes are allowed, knowing a single reference value $R(B{=}\tilde{i}_b\!\mid\!A{=}0)$ is enough to recover all the other conditionals $R(B{=}i\mid A{=}0)$ from the quantities $Q(A{=}0\mid B{=}i)$.

The key point is the same as in Section~\ref{sec:communicating}: in semantic dimension~$m=1$, synchronising two valuations requires fixing a single reference value. For quasi-probabilities, we may always choose the valuation of the top element as this reference. For pre-probabilities, however, if the top element is assigned value zero, we must instead specify the valuation of some other element.

\subsubsection{Case 3:  Quasi-probability with sub-universe with invariant valuation}
This case is exactly as the case $2$, except that for the fact that for any  value $\tilde{i}_b$ of $B$ we have $Q(A=0,B=\tilde{i}_b)= 0$.
In this case the valuation $V$ on $\mathcal{L}_{a\vee b}$ is the invariant valuation and we have that 
\begin{equation*}
	Q(B=1\mid A=0)=0.
\end{equation*}

\section{Classical Probability}
\label{sec:classical-prob}

Throughout our reconstruction we have seen that \emph{pure} pre-probabilities, namely pre-probabilities that cannot be reparametrised as quasi-probabilities, arise precisely in universes where to the top element $\top$ is assigned the valuation $V(\top)=\varphi$, that is, where $R(\top)=0$. In the pure pre-probability case there is no canonical way to fix the pre-probability representation.

By contrast, whenever $V(\top)\neq \varphi$ the universal valuation can always be parametrised as a quasi-probability, which provides a canonical representation. However, we have also seen that quasi-probabilities are not stable under relativisation.

In this section we show that, for universes equipped with a quasi-probability that has the same sign on all statements, one can always parametrise the universal valuation as a genuine probability. Moreover, this probabilistic representation is stable under relativisation. In other words, restricting the set of admissible valuation to the one with the same sign on all statements of the universe corresponds to restricting to set of admissible valuations that are stable under relativisation~\footnote{The value of $0$ is always included, both in the valuations all negative and in the valuations all positive}.

We therefore conclude our reconstruction by constructing the probability representation, A brief comparison with Cox and Jaynes' classic derivation of probability can be found in appendix \ref{sec:Cox-Jaynes}.

\subsection{Classical probabilities and stability of relativisation}

For every universe with a valuation admitting a quasi-probability representation, if the sign of the quasi-probability for every statement of the universe is the same, we can define a (classical) probability function.

\begin{definition}[Probability function]
Let $\mathcal{L}$ be a universe. A \textbf{probability function} on $\mathcal{L}$ is a quasi-probability function
$P:\mathcal{L}\to\mathbb{Q}\subseteq\mC$ such that
\[
P(s)\ge 0 \qquad \text{for all } s\in\mathcal{L}.
\]
\end{definition}

Now consider a probability function $P$ on a universe $\mathcal{L}$, let $t\in\mathcal{L}$ and consider the relative sub-universe $\mathcal{L}_t$.
It is immediate to see that for probabilities relativisation is \emph{stable}.

\begin{definition}[Relative probability]
	\label{def:relative-probability}
	Let $P$ be a probability function on $\mathcal{L}$ and let $t\in\mathcal{L}$.
	The \textbf{relative probability} at $t$ is the map
	\[
	P_t:\mathcal{L}_t\to\mathbb{Q}\subseteq\mC
	\]
	defined by
	\[
	P_t(s)\coloneqq
	\begin{cases}
		\dfrac{P(s)}{P(t)}, & \text{if } P(t)>0,\\[1.2ex]
		0, & \text{if } P(t)=0,
	\end{cases}
	\qquad (s\in\mathcal{L}_t).
	\]
\end{definition}

When $P(t)>0$ this is the usual conditional assignment.
When $P(t)=0$, the definition $P_t\equiv 0$ comes as a consequence of positivity:
if $P(t)=0$, non-negativity together with additivity forces $P(s)=0$ for every $s\in\mathcal{L}_t$, therefore the restriction of $P$ to the whole relative universe is identically zero, the invariant valuation. Thus there is no freedom in the reparametrisation even when the top element is zero.
Thus, unlike the quasi-probabilistic case, passing to a relative universe never forces us out of the probabilistic representation.

\subsection{The rule of total probability}
\label{subsec:lotp}

Let $\mathcal{L}$ be a universe equipped with a classical probability $P$.
Let $\{\tilde s_j\}_j$ be a family of pairwise disjoint statements forming a partition of the top element,
\[
\tilde s_j\wedge \tilde s_k=\bot \quad (j\neq k),
\qquad
\bigvee_j \tilde s_j=\top,
\]
so that 
\[
\sum_j P(\tilde s_j)=P(\top)=1.
\]
In this case the law of total probability takes exactly the same form as the rule of total quasi-probability in Proposition~\ref{def:rule-total-quasi-prob}.

\begin{proposition}[Rule of total probability]
	\label{prop:lotp}
	Under the assumptions above, for every $s\in\mathcal{L}$ one has
	\[
	P(s)=\sum_{j} P(s\mid \tilde s_j)\,P(\tilde s_j).
	\]
\end{proposition}

Because of stability under relativisation we do not have mixed cases.

\subsection{Relative probabilities and conditional probabilities}

It is interesting to compare the notion of relative probabilities with the one of conditional probabilities. The relative probability $P_t(\cdot)$ is a function from the sub-universe $\mathcal{L}_t$ to $\mathbb{Q}$, while the conditional probability  $P(\cdot\mid t)$ is a function from the universe $\mathcal{L}$ to $\mathbb{Q}$. These two functions differ in their domain, but both take the same values on the sub-universe $\mathcal{L}_t$.  

\noindent
The relative probability $P_t$ is a restriction of the probability function to a sub-universe with a reparametrisation via $a$-function. From the point of view of the underlying universal valuation $V$, the relative probability is exactly a restriction of $V$ to the sub-universe $\mathcal{L}_t$. Thus probabilities and relative probabilities are exactly the same valuation, what changes is just a restriction of the domain. 

\noindent
Conditional probabilities do not describe the same universal valuation $V$ from which they originate. Conditional probabilities are the extension of relative probabilities to the domain of the full ambient universe through the localisation map. This extension gives a fixed probability value to all elements of the ambient lattice that belong to the same class specified by the localisation map (see FIG. \ref{fig:conditionals} for an example of classes and the extension of relative probabilities to conditional probabilities).

\section{Higher semantic dimension: normalised and top-zero semantic frames}
\label{sec:higher-semantic-dimension}

So far we have focused on universes of semantic dimension $m=1$, where all nonzero values of a pre-probability are rationally proportional to one another and, once we pass to quasi-probabilities, the normalisation $Q(\top)=1$ fixes the remaining $a$-freedom.
The picture changes in higher semantic dimension: when the values of a valuation span several $\mathbb Q$-independent directions in $\mathbb C$, a single scalar constraint cannot fix how those independent directions are represented.

Because of this, in the quasi-probability (and hence classical probability) setting, normalisation forces a structural asymmetry among semantic frames:
only exactly one semantic direction can carry the normalisation of $\top$, while every other independent direction must have top equal to $0$.
Equivalently, every quasi-probability canonically splits into one genuinely normalised component and several additional components that are necessarily pure pre-probabilities.

\paragraph{$\mathbb{Q}$-linear decomposition from a semantic frame}

Let $\mathcal{L}$ be a universe and let $R:\mathcal{L}\to\mathbb{C}$ be a pre-probability.
Fix a semantic frame
\[
\SL{R}=\{s_1,\dots,s_m\}\subseteq\mathcal{L},
\]
so that $\{R(s_k)\}_{k=1}^m$ are $\mathbb{Q}$-linearly independent and maximal with this property
(Definition~\ref{def:gauge-frame-dim}).

\begin{restatable}{lemma}{Uniquerationalcoordinates}
\label{lem:unique-q-coordinates}
For every statement $s\in\mathcal{L}$ there exist unique rationals $q_1(s),\dots,q_m(s)\in\mathbb{Q}$ such that
\begin{equation}
\label{eq:R-q-decomp}
R(s)=\sum_{k=1}^{m} q_k(s)\,R(s_k).
\end{equation}
\end{restatable}
The proof of the Lemma can be found in Appendix~\ref{proof:unique-q-coordinates}.

The functions $q_k:\mathcal{L}\to\mathbb{Q}$ can be interpreted as the rational coordinates of $R(s)$ with respect to the semantic frame.

\paragraph{Semantic components are pre-probabilities}

Using the coordinate functions we define the semantic components of $R$ as the parts supported on each independent frame value.

\begin{definition}[Semantic component of a pre-probability]
\label{def:semantic-components}
Let $\SL{R}=\{s_1,\dots,s_m\}$ be a semantic frame and let $q_k(\cdot)$ be the coordinate functions from Lemma~\ref{lem:unique-q-coordinates}.
For each $k\in\{1,\dots,m\}$ define
\begin{equation}
\label{eq:def-Rk}
R^{(k)}:\mathcal{L}\to\mathbb{C},
\qquad
R^{(k)}(s)\coloneqq q_k(s)\,R(s_k).
\end{equation}
We call $\{R^{(k)}\}_{k=1}^m$ the \textbf{semantic decomposition} of $R$ relative to the frame $\SL{R}$.
\end{definition}

By construction,
\begin{equation}
\label{eq:R-sum-components}
R(s)=\sum_{k=1}^{m} R^{(k)}(s)
\qquad \text{for all } s\in\mathcal{L}.
\end{equation}

Moreover, each $R^{(k)}$ is itself a pre-probability of semantic dimension $1$.

\paragraph{Normalisation singles out a unique semantic direction}

We now specialise to quasi-probabilities.
Let $Q:\mathcal L\to\mathbb C$ be a quasi-probability, so that $Q(\top)=1$.
The semantic decomposition \eqref{eq:R-sum-components} applied to $\top$ reads
\begin{equation}
\label{eq:top-decomp-general}
1
=
Q(\top)
=
\sum_{k=1}^m Q^{(k)}(\top)
=
\sum_{k=1}^m q_k(\top)\,Q(s_k),
\end{equation}
which is an equality between $1$ and a $\mathbb Q$-linear combination of the $\mathbb Q$-independent numbers $\{Q(s_k)\}_{k=1}^m$.

For quasi-probabilities one can choose the semantic frame so that $\top$ itself belongs to the frame, this forces the remaining directions to be top-zero.

\begin{restatable}{proposition}{uniquequasiprob}
\label{prop:normalised-frame-exists}
Let $\mathcal L$ be a universe and let $Q:\mathcal L\to\mathbb C$ be a quasi-probability with $Q(\top)=1$.
Then there exists a semantic frame
\[
\SL{Q}=\{\top,s_2,\dots,s_m\}\subseteq\mathcal L
\]
such that, for the semantic decomposition $\{Q^{(k)}\}_{k=1}^m$ relative to this frame, one has
\begin{equation}
\label{eq:one-normalised-rest-topzero}
Q^{(1)}(\top)=1,
\qquad
Q^{(k)}(\top)=0\ \ (k=2,\dots,m).
\end{equation}
In particular, $Q^{(1)}$ is a quasi-probability of semantic dimension $1$, while each $Q^{(k)}$ with $k\ge 2$ is a pure pre-probability.
The same conclusion holds \emph{a fortiori} for classical probabilities.
\end{restatable}
The proof of the Proposition can be found in Appendix~\ref{proof:normalised-frame-exists}.

Optional remarks on number-field choices (e.g. rational vs real components) and on frequency-like interpretations are deferred to Appendices~\ref{sec:rational} and~\ref{sec:irrational}.

\subsection{What happens with regularity conditions on $a$-functions?}
\label{sec:holomorphic}

Throughout this work we have allowed $a$-functions in their full algebraic generality, namely as additive automorphisms of $(\mathbb{C},+)$. We imposed no regularity assumptions beyond additivity and bijectivity because our starting point considers the universal valuation $V$ as a primitive object: without additional analytic structure on its codomain there is no canonical notion of \enquote{regularity} to require. In particular, our choice admits highly non-constructive solutions of Cauchy’s functional equation. Algebraically, this is exactly what makes the gauge freedom large enough to act independently along different $\mathbb{Q}$-directions in the value space, and hence what supports the discussion of higher semantic dimension.

In many interpretational settings one does want extra structure. If $V$ is meant to represent, for example, \enquote{information content}, \enquote{plausibility}, or another quantity for which complex analysis is intended to be meaningful, then it is natural to restrict admissible reparametrisations to those preserving that structure. Concretely, one may require the allowed gauge regraduations (the $a$-functions) to be holomorphic while still bijective and satisfying Cauchy’s functional equation. Under this holomorphic restriction the $a$-freedom collapses completely. Every holomorphic solution is necessarily of the form \cite{aczel1989,churchill2014}
\begin{equation}
	a(z)=c\,z,
	\qquad c\in\mathbb{C},
	\label{eq:a-cz}
\end{equation}
and bijectivity is equivalent to $c\neq 0$. Thus the only remaining gauge action is a uniform global rescaling and rotation applied to all values at once.

Equivalently, holomorphicity collapses the $a$-freedom to function acting as $z\mapsto cz$. In particular, this removes the possibility of regraduations that treat distinct $\mathbb{Q}$-independent semantic directions non-uniformly, for instance, fixing one semantic component while amplifying another (see for example appendix~\ref{sec:rational}). With holomorphic $a$-functions there is only a single constant $c$, so no separate action on different $\mathbb{Q}$-vector-space components singled out by a semantic decomposition is available.

A useful conceptual consequence is that, under holomorphicity, one may work in any semantic frame and, in particular, in the full semantic dimension without introducing frame-dependent gauge artefacts. Indeed, the residual regraduation $z\mapsto cz$ acts uniformly on the entire value space. The most important consequence is that one naturally reintroduce irrational numbers for probabilities, and full complex numbers for pre-probabilities and quasi-probabilities.

Importantly, since this restriction comes after the additivity theorem, this restriction does not alter Theorem~1, nor any of the algebraic consequences drawn from additivity and synchronisability. All results on pre-probabilities, the synchronisation problem, and the generalised Bayes' theorem remain valid: what changes is only the size and qualitative behaviour of the admissible gauge freedom (and hence the bookkeeping needed to compare representations).

The discussion of quasi-probabilities is likewise essentially unchanged at the level of gauge-invariant content. In semantic dimension $1$, holomorphicity means that (before fixing a normalisation) quasi-probabilities may be represented by arbitrary complex numbers, since the residual freedom is complex scaling; imposing a canonical normalisation such as $Q(\top)=1$ fixes $c$ and removes this freedom entirely. More generally, under holomorphicity the remaining gauge is purely global, so a single canonical normalisation suffices to eliminate the residual ambiguity even when working in higher semantic dimension.

Finally, if one proceeds to define \emph{probabilities} as those quasi-probabilities that are stable under relativisation and behave as genuine nonnegative weights, also in higher-semantic dimensions. Then one must additionally require that the originating quasi-probability be real and nonnegative. This rules out nontrivial cancellations (e.g.\ purely imaginary parts summing to zero at $\top$) that would otherwise allow spurious \enquote{top-zero} behaviours under relativisation. With this positivity requirement in place, the probability notion is recovered exactly as in our earlier definition, and all previously derived algebraic identities continue to hold.

The message here is, if one admits extra regularity conditions pre-probabilities, quasi-probabilities and probabilities lose the decomposition along different $\mathbb{Q}$-direction, the proofs of this work works analogously (since we mainly studied valuations and universes with semantic dimension $m=1$), and one recovers representations on the full field of complex/real numbers~\footnote{Even with additional regularity conditions, the results of this work can be formulated for a valuation $\tilde{V}:\mathcal{L}\to\mR$ just by substituting $\mC$ with $\mR$. Since $(\mC,+,0)$ is isomorphic to $(\mR,+,0)$ the proof of Theorem~1 will be identical. In general, any group isomorphic to $(\mC,+,0)$ is suitable in the proof of Theorem~1. Thus one might have a representation additive on real numbers, but also, for example, additive on quaternions.}.

\section{Conclusions}

When you flip a coin, most of the time you have no idea whether it will land heads or tails. Still, you can imagine a scenario in which you have enough information to predict perfectly the face on which it will land. Our theory of physics tells you how to do it. When you measure a quantum system, most of the time you have no idea what measurement result you will obtain, and there is no way for you to know. Our theory of physics tells you this too.

Both scenarios are modelled using probability. However, in the coin scenario, probability is clearly a tool for dealing with uncertainty, and it is hard to be convinced that it is a property of the coin and the flip. The fact that the flipped coin will land heads is either true or false, and this truth or falsity is a property of the coin and the flip. On the contrary, in the quantum scenario, one can feel compelled to think that probability is actually a property of the system and the measurement. The theory predicts this probability precisely, and we can test it experimentally. The statement about the fact that the outcome of the measurement \textit{will} be, let us say, $0$ is not simply true or false. It comes with a probability, and this probability looks like a property of the quantum system and the measurement. All information about the system and the measurement in its totality, yields a probability associated with statement about the fact that the measurement \textit{will} be $0$, it doesn't tell if it is true or false. 

Quantum mechanics suggests to us that we should treat its probabilities as properties of our systems. We take this suggestion seriously and try to develop a theory from it. We are used to statements having a property, a binary property that we call a truth value. Quantum theory pushes us to extend the range of this quantity, for every statement, beyond two values. Since we also know that quantum mechanics often speaks in terms of quasi-probabilities with complex numbers, we start from a quantity attached to statements, a Universal Valuation $V$, allowed to take complex values. This reasoning drives us to the direct formalisation of one fundamental property we want for this universal valuation $V$. We want the Universal Valuation $V$ to be compatible with classical logic, that is, with the assignment of truth values. At the end of the day, we do not want something exotic. We want a theory that incorporates what we already know about this property of statements, the universal valuation $V$, and we already know how it behaves when it takes two distinguished values that we are used to calling true and false.

Moreover, once we interpret this quantity as a property of statements, we can import concepts from physics. We may treat statements as \enquote{physical} objects, and the universal value $V$ as a \enquote{physical} property of these objects. Inspired by this analogy we introduce some principles that we expect a universal valuation $V$ should satisfy. Before doing so, however, we introduce the concept of Syntactic Locality. This is nothing more than deciding what it means for a system of statements to be close, what a subsystem is, and how to embed systems of statements in bigger systems. It is the concept of relativisation of systems. Inspired by Moore, we call these systems (syntactic) \enquote{universes}. On top of this it is also the idea that \textit{every} universe of discourse we take in consideration should always be thought of as a local part of a bigger universe. 

On top of the reduction to truth valuations, we introduce four additional principles. First, we assume that the universal value $V$ has structure, in the sense that its value can be deduced from information about other elements of the universe. We call this the principle of local deducibility. Second, we assume that these rules are invariant across universes. There is a unique structure for the universal valuation $V$, and it does not depend on the specific discourse. We call this the principle of universality. Then we ask for the set of universal valuations we characterise to be as large as possible, compatible with our principles. We do not want to include hidden constraints. Finally, we ask that the universal valuation $V$ behaves well under symmetries of the syntax.

From these principles, clearly inspired by basic ideas in physics, we obtain the main result of this work, which strongly characterises the universal valuation $V$. We find that any universal valuation $V$ can always be reparametrised as a function that is additive on statements. In other words, the universal valuation $V$ is essentially finitely additive. Additivity is the essential property of our theory of probability, and we obtain that the universal valuation $V$, this property we attach to every statement, can always be regarded as additive. Because of this similarity in structure with probability, we call these reparametrisations of the universal valuation $V$ \textit{pre-probabilities}. As a lemma of the main structure theorem, we also find that there is a huge freedom in choosing the pre-probability representation of the universal valuation $V$, and we characterise this freedom.

Studying this freedom yields two important results. The first concerns the possibility of fixing a canonical pre-probability parametrisation. We find that it is not always possible, but in some cases the freedom can be eliminated canonically. When this can be done, we call the canonically fixed representation a quasi-probability representation. In fact, in this canonical form the representation of the universal valuation $V$ resembles exactly the standard quasi-probabilities one is used to defining by relaxing Kolmogorov's positivity axiom, and possibly extending the domain to complex numbers. We find that quasi-probabilities, a language that appears almost unavoidably in quantum mechanics, emerge naturally and coherently once one tries to characterise the universal valuation $V$.

Whilst quasi-probabilities are usually treated as a useful calculation tool, often without meaning of their own, here we derive them from motivated principles, making them a theory in their own right. In this view, quasi-probabilities are a proper structured assignment of values to statements in syntactic universes. Furthermore, the framework we develop dissipates the problems that one generally encounters when defining conditionals for standard quasi-probabilities, and clarifies what we really mean by conditioning in a broader sense, not confined to standard classical probability. This allows us to obtain and define a generalised rule of total probability and a generalised Bayes' theorem, both for pre-probabilities and for quasi-probabilities. In this framework, probability appears as a special case of quasi-probability, namely those quasi-probabilities that are stable under relativisation. Said differently, these are the ones that behave properly and coherently when moving from universes to sub-universes, or to ambient (dilated) universes. Probabilities are quasi-probabilities equipped with a stability property. From this perspective, it is hardly surprising that, if one interprets quasi-probabilities not as such but merely as relaxations of ordinary probabilities, certain things fail. What fails, precisely, is this stability property, and one must address this by working with pre-probabilities.

The second result that comes from studying the freedom in choosing a pre-probabilistic representation for the universal valuation $V$ is the discovery of additional structures, which we call semantic frames.
In particular, in our framework, since we allow for highly non-regular gauge reparametrisations, quasi-probabilities and probabilities can naturally be taken to have values in the rational numbers rather than the irrational ones, with irrational numbers carrying additional structure.
This extra structure is equipped with a valuation that cannot be described by quasi-probabilities and probabilities themselves, but only via pre-probabilities.

This is a particular structure that resonates with the constructions of de Finetti~\cite{definetti1931b}. However, we noticed that as soon as we ask for additional regularity conditions (such as holomorphicity if we work over the field of complex numbers, or continuity at a single point if we work over the field of real numbers), this extra structure collapses and our reconstruction recovers quasi-probabilities and probabilities on the whole field of complex/real numbers, without distinctions between rationals and irrationals. One merit of this reconstruction is also to highlight the necessity of this extra assumption and the fact that the fundamental structure of probability and quasi-probability emerges even without it.

Finally we discuss in the appendices how rational versus irrational numerical realisations interact with certain frequentist intuitions,

From the foundational side, we stress that the hypothesis of a universal valuation $V$ is not intended to take a stance on the interpretation of quantum states. The framework concerns the structure of valuations on syntactic universes and the calculus they induce, it is just inspired by quantum mechanics, and it is independent of how one chooses to interpret the underlying quantum formalism.

Once a more refined theory of quasi-probability that incorporates pre-probabilities is in place, two natural directions for future work suggest themselves. The first is the development of an inference theory in the Bayesian sense. Bayesian inference prescribes how priors should be updated in the light of data, thereby providing a bridge between abstract probabilistic structures and experimental practice, a link that has not been addressed here. The quasi-probabilities considered in this work are entirely theoretical; an inference framework would clarify how they may be updated after an experiment and thus how they might be brought into contact with observations. Standard Bayesian inference rests on a tight relationship between conditional probabilities, as theoretical objects, and the operational act of updating a probabilistic description once outcomes are recorded. An analogous relationship can be sought for quasi-probabilities, although one must account for additional consistency requirements specific to the quasi-probabilistic setting, for which the quantum-mechanical literature may offer useful guidance.

A second direction concerns the fact that one is often less interested in unconstrained quasi-probability theories than in constrained ones. In quantum mechanics, for instance, not every quasi-probability vector is admissible, and the meaningful objects are precisely those that satisfy the relevant constraints. One may therefore view quasi-probabilities as experimentally valuable chiefly when constrained, since they can serve as a compact way of encoding admissibility conditions. From this perspective, a systematic study of \emph{constrained quasi-probabilities} should lead to constraint-preserving notions of conditionals and to a corresponding version of Bayes' theorem, tailored so that updating rules respect the structure of the admissible set.

\section{Acknowledgements}
I thank the many people  with whom I have discussed these ideas, including the quantum foundations group at Perimeter Institute, in particular Y\`{\i}l\`{e} Y{}\={\i}ng and Rob Spekkens for the last part of the appendix, the group of Borivoje Daki\'{c}, in particular Flavio Del Santo and Borivoje Daki\'{c}, the one of \v{C}aslav Brukner, in particular Luis Cortés Barbado and  \v{C}aslav Brukner; and Matteo Scandi, and Leonardo Vaglini, and Shintaro Minagawa.
This work has been done with the support of the French government under the France 2030 investment plan, as part of the Initiative d'Excellence d'Aix-Marseille Université-A*MIDEX, AMX-22-CEI-01.

 \bibliography{Classical-Conditionals.bib}

 \newpage
 
\appendix
\onecolumngrid

\section{Proof of theorem 1}
\label{app:study-of-G}

\subsection{Valuation of atoms can be freely chosen}

Before starting with the main proof it is useful to prove the following lemma.

\begin{lemma}
	\label{lem:atomic-richness}
	Fix $n\ge 1$ and values $x_1,\dots,x_n\in \mC$. Then for any universe $\mathcal{L}_n$ with atoms $a_1,\dots,a_n$ there exists an admissible valuation $V:\mathcal{L}_n\to \mC$
	such that $V(a_i)=x_i$ for all $i$.
\end{lemma}

\begin{proof}
	Consider $ \mathcal{S}:=\{\bot,a_1,\dots,a_n\}\subseteq \mathcal{L}_n $ and define a partial valuation $p:\mathcal{S}\to \mC$ by
	\[
	p(\bot)=\varphi,
	\qquad
	p(a_i)=x_i \quad (i=1,\dots,n).
	\]
	First of all we can see that $p$ is a consistent partial valuation.
	Indeed, $p(\bot)=\varphi$ holds by construction. Furthermore every sub-universe of $\mathcal{L}$ containing more than one element of $\mathcal{S}\setminus\{\bot\}$ cannot be included in $\mathcal{S}$ since $\mathcal{S}$ does not include any join of atoms.
	Therefore no deducibility constraint is imposed on the values $p(a_i)$, and $p$ is consistent.
	By Principle~\ref{prin:realisability}, $p$ extends to an admissible valuation $V\in\Val(L_n)$. In particular, $V(a_i)=p(a_i)=x_i$ for all $i=1,\dots,n$.
\end{proof}

\subsection{Behaviour of $V$ with respect to $\vee$: An outline of the proof of theorem \ref{th:sum-rule}}

Before starting we define two functions $\mathsf{G}:\mC^2\to \mC$ and $\mathsf{K}:\mC^2\to \mC$ such that for any sub-universe $\mathcal{L}_2$ of a universe $\mathcal{L}$ the following equalities hold
\begin{equation*}
	V(a\vee b)=\Ded_{2,2}(V(a),V(b))\coloneqq \mathsf{G}\big(V(a),V(b)\big), \qquad V(a)=\Ded_{2,1}(V(a\vee b), V(b))\coloneqq \mathsf{K}\big(V(a\vee b),V(b)\big)
\end{equation*}
The existence of these functions is guaranteed by the principle of local deducibility, and because of Universality this same function can be used for any $V$ on any sub-universe with two atomic statements.

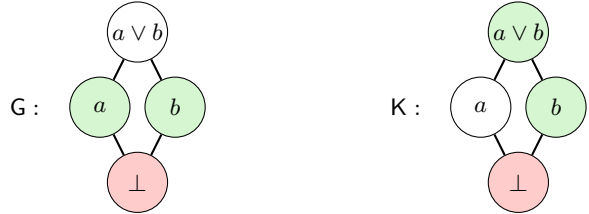
\begin{figure}[h!]
\centering
\begin{minipage}{0.3\textwidth}
\centering
\begin{tikzpicture}[baseline=(current bounding box.center)]
  \node (label) at (-1.5,1) {$\mathsf{G}:$};
  \node[lattice node, fill=red!20] (bot) at (0,0) {$\bot$};           
  \node[lattice node, fill=green!20] (a) at (-0.5,1) {$a$};
  \node[lattice node, fill=green!20] (b) at (0.5,1) {$b$};
  \node[lattice node] (top) at (0,2) {$a \vee b$};

  \draw[lattice edge] (bot) -- (a);                  
  \draw[lattice edge] (bot) -- (b);                  
  \draw[lattice edge] (a) -- (top); 
  \draw[lattice edge] (b) -- (top);     
\end{tikzpicture}
\end{minipage}
\begin{minipage}{0.3\textwidth}
\centering
\begin{tikzpicture}
	\node (label) at (-1.5,1) {$\mathsf{K}:$};
  \node[lattice node, fill=red!20] (bot) at (0,0) {$\bot$};           
  \node[lattice node] (a) at (-0.5,1) {$a$};
  \node[lattice node, fill=green!20] (b) at (0.5,1) {$b$};
  \node[lattice node, fill=green!20] (top) at (0,2) {$a \vee b$};

  \draw[lattice edge] (bot) -- (a);                  
  \draw[lattice edge] (bot) -- (b);                  
  \draw[lattice edge] (a) -- (top); 
  \draw[lattice edge] (b) -- (top);   	
\end{tikzpicture}
\end{minipage}
\caption{Because of the principle of local deducibility (principle \ref{dec-prin:syntactic-locality}), the \property of the white elements of the lattice is completely specified by the green elements of the lattice. Because of Universality (principle \ref{dec-prin:universality}) the way in which $a$ and $b$, specify $a\vee b$, and the way in which $b$ and $a\vee b,$ specify $a$ is general. We identify the relation between white and green elements with the function $\mathsf{G}$ for the case on the left, and with $\mathsf{K}$ for the case on the right.}
\label{fig:Functions-R-K}
\end{figure}

We will characterise the valuations $V$ via $\mathsf{G}$. Because the process of characterisation of $\mathsf{G}$ is long, we divide it into steps following the structure outlined in the following.\\

\begin{paragraph}{\textbf{Part 1: $(\mC,\mathsf{G},\varphi)$, is a cancellative, commutative, torsion-free monoid.}}
	
In the first part we will show that principles \eqref{dec-prin:classical-compatibility},\eqref{dec-prin:syntactic-locality},\eqref{dec-prin:universality}, and \eqref{dec-prin:realisability} imply that the function $\mathsf{G}$ has the following properties for every $x,y,z\in \mC$:
\begin{itemize}
\item $\GG{x,y}=\GG{y,x}$ \hfill \enquote{\textit{commutativity}}(see section~\ref{subsub:commutativity}),
\item $\GG{x,\GG{y,z}}=\GG{\GG{x,y},z}$ \hfill \enquote{\textit{associativity}}(see section~\ref{subsub:associativity}),
\item $\GG{x,\varphi}=x$ \hfill \enquote{\textit{neutral element}}(see section~\ref{subsub:neutral-element}),
\item $\GG{x,y}=\GG{x,z} \implies y=z$ \hfill \enquote{\textit{cancellative}}(see section~\ref{subsub:cancellativity}).
\end{itemize}
Essentially we will show that principles \eqref{dec-prin:classical-compatibility},\eqref{dec-prin:syntactic-locality},\eqref{dec-prin:universality} imply that \textbf{$(\mC,\mathsf{G},\varphi)$ is a cancellative commutative monoid}.
\noindent
With this structure in hand we will see that on every lattice $\mathcal{L}$, each $V$ is completely characterised by its value on the atomic statements of $\mathcal{L}$ .
Subsequently, requesting for $\mathsf{G}$ to satisfy principles~\eqref{dec-prin:symmetry-removes-freedom}, we will obtain that \textbf{$(\mC,\mathsf{G},\varphi)$ is a cancellative, commutative, torsion-free, divisible monoid}.\\
\end{paragraph}

\begin{paragraph}{\textbf{Part 2: Mapping of $(\mC,\mathsf{G},\varphi)$ into $(\mC,+,0)$}}

	 In this part we prove the main theorem, namely the existence and uniqueness up to group automorphisms of the pre-probability mapping $\phi$ that allows one to bijectively map every \property $V$ to a pre-probability, that is a valuation $R$ that is additive on the join of atomic statements
  \begin{equation}
    V \overset{\phi}{\mapsto} R \mid R(s_i\vee s_j)=R(s_i)+R(s_j) \qquad \mbox{with $s_i,s_j$ atomic statements.}
  \end{equation}
  To do so we show that we are in the condition of using the structure theorem of divisible groups to show that  $(\mC,\mathsf{G},\varphi)$ is isomorphic to $(\mC,+,0)$ proving the theorem and obtaining that the pre-probability mapping $\phi$ is a bijective homomorphism that maps $(\mC, \mathsf{G}, \varphi)$ to $(\mC, +,0)$.

\end{paragraph}

\subsubsection{\textbf{Part 1: $(\mC,\mathsf{G},\varphi)$, is a cancellative, commutative, torsion-free monoid.}}

\paragraph{\textbf{Commutativity}}
\label{subsub:commutativity}
We consider the two lattices $\mathcal{L}_2$ of FIG. \ref{lat:commutative}.

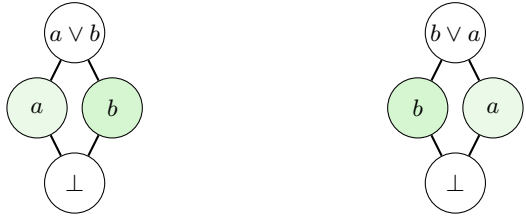
\begin{figure}[h]
\centering
\begin{minipage}{0.3\textwidth}
\centering
\begin{tikzpicture}
  \node[lattice node] (bot) at (0,0) {$\bot$};           
  \node[lattice node, fill=green!10] (a) at (-0.5,1) {$a$};
  \node[lattice node, fill=green!20] (b) at (0.5,1) {$b$};
  \node[lattice node] (top) at (0,2) {$a \vee b$};

  \draw[lattice edge] (bot) -- (a);                  
  \draw[lattice edge] (bot) -- (b);                  
  \draw[lattice edge] (a) -- (top); 
  \draw[lattice edge] (b) -- (top);     
\end{tikzpicture}
\end{minipage}
\begin{minipage}{0.3\textwidth}
\centering
\begin{tikzpicture}
  \node[lattice node] (bot) at (0,0) {$\bot$};           
  \node[lattice node, fill=green!20] (a) at (-0.5,1) {$b$};
  \node[lattice node, fill=green!10] (b) at (0.5,1) {$a$};
  \node[lattice node] (top) at (0,2) {$b \vee a$};

  \draw[lattice edge] (bot) -- (a);                  
  \draw[lattice edge] (bot) -- (b);                  
  \draw[lattice edge] (a) -- (top); 
  \draw[lattice edge] (b) -- (top);     
\end{tikzpicture}
\end{minipage}
\caption{The sub-universe on the right is obtained via the Boolean lattice isomorphism $\beta_s$ defined in equation~\eqref{eq:beta-commut}}
\label{lat:commutative}
\end{figure}

We define the Boolean lattice isomorphism
\begin{equation}
	\label{eq:beta-commut}
	\beta(\bot)=\bot, \qquad \beta(a)=b, \qquad \beta(b)=a,\qquad\beta(a\vee b)=a\vee b.
\end{equation}

Because of local deducibility and Universality we have
\begin{align*}
	V(a\vee b)=\mathsf{G}(V(a),V(b)), \qquad
	V(a\vee b)=\mathsf{G}(V(b),V(a)).
\end{align*}
It follows directly that
\begin{equation*}
	\mathsf{G}(V(a),V(b))=\mathsf{G}(V(b),V(a)), 
\end{equation*}
and since this has to be valid for every $V$ we have that \textbf{$\mathsf{G}$ is commutative}.

\paragraph{\textbf{Associativity}}
\label{subsub:associativity}
We consider the two sub-universes of the universe $\mathcal{L}_3$ shown in the left part of FIG.~\ref{fig:L3}. For these we have
\begin{align*}
	V(a\vee b\vee c) = \GG{V(a\vee b),V(c)},\\
	V(a\vee b\vee c) = \GG{V(a),V(b\vee c)}. 	
\end{align*}
We also consider the two sub-universes~\footnote{The principal ideals $\PrI{a\vee b}$and $\PrI{b\vee c}$.} shown in the right part of FIG.~\ref{fig:L3}, for these we have
\begin{align}
	V(a\vee b)=\GG{V(a),V(b)},\\
	V(b\vee c)=\GG{V(b),V(c)}.
\end{align}
Putting all this together we have
\begin{equation*}
	\GG{\GG{V(a),V(b)},V(c)}=\GG{V(a),\GG{V(b),V(c)}}.
\end{equation*}
This must be valid for every $V$, thus \textbf{$\mathsf{G}$ is associative}.
\begin{figure}[h!]
\centering
\begin{minipage}{0.2\textwidth}
	\begin{tikzpicture}
  \node[lattice node,line width=2pt] (bot) at (0,0) {$\bot$};               

  \node[lattice node, dotted]  (a)  at (-1,1) {$a$};
  \node[lattice node, dotted]  (b)  at (0,1)  {$b$};
  \node[lattice node,line width=2pt]  (c)  at (1,1)  {$c$};

  \node[lattice node,line width=2pt] (ab) at (-1,2) {$a\vee b$};
  \node[lattice node, dotted] (ac) at (0,2)  {$a\vee c$};
  \node[lattice node, dotted] (bc) at (1,2)  {$b\vee c$};

  \node[lattice node,line width=2pt] (top) at (0,3) {$\top$};

  \draw[lattice edge,very thick] (bot) -- (a);
  \draw[lattice edge, dotted] (bot) -- (b);            
  \draw[lattice edge,very thick] (bot) -- (c);

  \draw[lattice edge,very thick] (a) -- (ab);
  \draw[lattice edge, dotted] (a) -- (ac);

  \draw[lattice edge, dotted] (b) -- (ab);
  \draw[lattice edge, dotted] (b) -- (bc);

  \draw[lattice edge, dotted] (c) -- (ac);
  \draw[lattice edge,very thick] (c) -- (bc);

  \draw[lattice edge,very thick] (ab) -- (top);
  \draw[lattice edge, dotted] (ac) -- (top);        
  \draw[lattice edge,very thick] (bc) -- (top);
\end{tikzpicture}
\end{minipage}
\centering
\begin{minipage}{0.2\textwidth}
	\begin{tikzpicture}
  \node[lattice node,line width=2pt] (bot) at (0,0) {$\bot$};               

  \node[lattice node,line width=2pt]  (a)  at (-1,1) {$a$};
  \node[lattice node, dotted]  (b)  at (0,1)  {$b$};
  \node[lattice node, dotted]  (c)  at (1,1)  {$c$};

  \node[lattice node, dotted] (ab) at (-1,2) {$a\vee b$};
  \node[lattice node, dotted] (ac) at (0,2)  {$a\vee c$};
  \node[lattice node,line width=2pt] (bc) at (1,2)  {$b\vee c$};

  \node[lattice node,line width=2pt] (top) at (0,3) {$\top$};

  \draw[lattice edge, very thick] (bot) -- (a);
  \draw[lattice edge, dotted] (bot) -- (b);            
  \draw[lattice edge, very thick] (bot) -- (c);

  \draw[lattice edge, very thick] (a) -- (ab);
  \draw[lattice edge, dotted] (a) -- (ac);

  \draw[lattice edge, dotted] (b) -- (ab);
  \draw[lattice edge, dotted] (b) -- (bc);

  \draw[lattice edge, dotted] (c) -- (ac);
  \draw[lattice edge, very thick] (c) -- (bc);

  \draw[lattice edge, very thick] (ab) -- (top);
  \draw[lattice edge, dotted] (ac) -- (top);        
  \draw[lattice edge, very thick] (bc) -- (top);
\end{tikzpicture}
\end{minipage}
\centering
\begin{minipage}{0.2\textwidth}
	\begin{tikzpicture}
  \node[lattice node,line width=2pt] (bot) at (0,0) {$\bot$};               

  \node[lattice node,line width=2pt]  (a)  at (-1,1) {$a$};
  \node[lattice node,line width=2pt]  (b)  at (0,1)  {$b$};
  \node[lattice node, dotted]  (c)  at (1,1)  {$c$};

  \node[lattice node,line width=2pt] (ab) at (-1,2) {$a\vee b$};
  \node[lattice node, dotted] (ac) at (0,2)  {$a\vee c$};
  \node[lattice node, dotted] (bc) at (1,2)  {$b\vee c$};

  \node[lattice node, dotted] (top) at (0,3) {$\top$};

  \draw[lattice edge, very thick] (bot) -- (a);
  \draw[lattice edge, very thick] (bot) -- (b);            
  \draw[lattice edge, dotted] (bot) -- (c);

  \draw[lattice edge, very thick] (a) -- (ab);
  \draw[lattice edge, dotted] (a) -- (ac);

  \draw[lattice edge, very thick] (b) -- (ab);
  \draw[lattice edge, dotted] (b) -- (bc);

  \draw[lattice edge, dotted] (c) -- (ac);
  \draw[lattice edge, dotted] (c) -- (bc);

  \draw[lattice edge, dotted] (ab) -- (top);
  \draw[lattice edge, dotted] (ac) -- (top);        
  \draw[lattice edge, dotted] (bc) -- (top);
\end{tikzpicture}
\end{minipage}
\centering
\begin{minipage}{0.2\textwidth}
	\begin{tikzpicture}
  \node[lattice node, line width=2pt] (bot) at (0,0) {$\bot$};               

  \node[lattice node, dotted]  (a)  at (-1,1) {$a$};
  \node[lattice node, line width=2pt]  (b)  at (0,1)  {$b$};
  \node[lattice node, line width=2pt]  (c)  at (1,1)  {$c$};

  \node[lattice node, dotted] (ab) at (-1,2) {$a\vee b$};
  \node[lattice node, dotted] (ac) at (0,2)  {$a\vee c$};
  \node[lattice node, line width=2pt] (bc) at (1,2)  {$b\vee c$};

  \node[lattice node, dotted] (top) at (0,3) {$\top$};

  \draw[lattice edge, dotted] (bot) -- (a);
  \draw[lattice edge, very thick] (bot) -- (b);            
  \draw[lattice edge, very thick] (bot) -- (c);

  \draw[lattice edge, dotted] (a) -- (ab);
  \draw[lattice edge, dotted] (a) -- (ac);

  \draw[lattice edge, dotted] (b) -- (ab);
  \draw[lattice edge, very thick] (b) -- (bc);

  \draw[lattice edge, dotted] (c) -- (ac);
  \draw[lattice edge, very thick] (c) -- (bc);

  \draw[lattice edge, dotted] (ab) -- (top);
  \draw[lattice edge, dotted] (ac) -- (top);        
  \draw[lattice edge, dotted] (bc) -- (top);
\end{tikzpicture}
\end{minipage}
\caption{Four lattices $\mathcal{L}_3$ with highlighted sub-lattices, the two on the left are two Boolean-sub-lattices, the two on the right are two principal ideals.}
\label{fig:L3}
\end{figure}
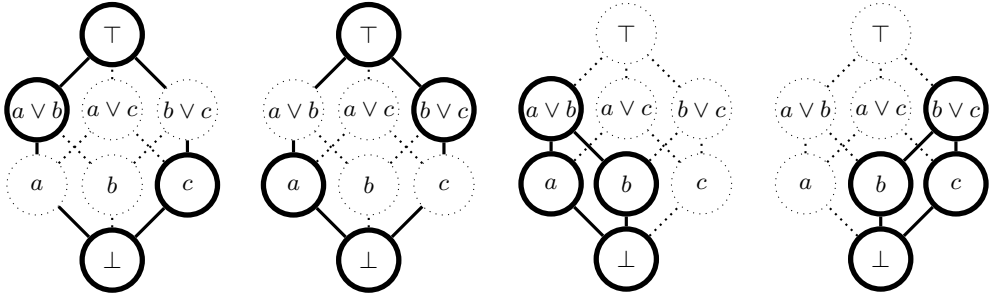

\paragraph{\textbf{Cancellativity}}
\label{subsub:cancellativity}

Consider the universe $\mathcal{L}_2=\{\bot,a,b,a\vee b\}$ with 
\begin{equation}\label{eq:swapK}
	V(b)=\mathsf{K}(V(a\vee b),V(a)).
\end{equation}

Set $x:=V(a)$ and $y:=V(b)$. By Lemma~\ref{lem:atomic-richness} with $n=2$, for every pair
$(x,y)\in \mC^2$ there exists a $2$-atom universe $\mathcal{L}_2=\{\bot,a,b,a\vee b\}$ with a valuation $V$
such that $V(a)=x$ and $V(b)=y$. Hence the identity above holds for all $x,y\in C$, i.e.
\begin{equation}
	\label{eq:KGidentity}
	\mathsf{K}(\mathsf{G}(x,y),x)=y\qquad\forall x,y\in C.
\end{equation}

Now suppose $\mathsf{G}(x,y)=\mathsf{G}(x,z)$. Applying $\mathsf{K}(\,\cdot\,,x)$ to both sides and using
\eqref{eq:KGidentity} yields $y=z$. Therefore \textbf{$\mathsf{G}$ is cancellative.}

\paragraph{\textbf{Neutral Element}}
\label{subsub:neutral-element}
We prove that $\varphi$ is a neutral element for $\mathsf{G}$.

First of all it is easy to see that $\varphi$ is the neutral element for $\mathsf{G}$ whenever we reduce to classical truth valuations. Let $\mathcal{L}_2=\{\bot,a,b,a\vee b\}$ be a $2$-atom universe. Consider the classical truth assignment on $\mathcal{L}_2$ with $a$ true and $b$ false. By Classical Compatibility
there exists an admissible valuation $V$ on $\mathcal{L}_2$ realising this assignment, namely
\[
V(\bot)=\varphi,\qquad V(a)=\tau,\qquad V(b)=\varphi,\qquad V(a\vee b)=\tau.
\]
Thus we have
\begin{equation}\label{eq:tauphi-neutral}
	\mathsf{G}(\tau,\varphi)=\tau,
\end{equation}
and by commutativity of $\mathsf{G}$ it also follows that $\mathsf{G}(\varphi,\tau)=\tau$.
Now fix an arbitrary $x\in\mathbb C$. By Lemma~\ref{lem:atomic-richness} with $n=3$, there exist a
$3$-atom universe $\mathcal{L}_3$ with atoms $a,b,c$ and an admissible valuation $W$ such that
\[
W(a)=x,\qquad W(b)=\varphi,\qquad W(c)=\tau.
\]
Using associativity of $\mathsf{G}$ we compute
\[
\mathsf{G}\bigl(\mathsf{G}(W(a),W(b)),\,W(c)\bigr)=\mathsf{G}\bigl(W(a),\,\mathsf{G}(W(b),W(c))\bigr).
\]
Substituting $W(a)=x$, $W(b)=\varphi$, $W(c)=\tau$ gives
\[
\mathsf{G}\bigl(\mathsf{G}(x,\varphi),\,\tau\bigr)=\mathsf{G}\bigl(x,\,\mathsf{G}(\varphi,\tau)\bigr).
\]
Using $\mathsf{\mathsf{G}}(\varphi,\tau)=\tau$ from \eqref{eq:tauphi-neutral}, we obtain
\begin{equation}\label{eq:step-neutral}
	\mathsf{G}\bigl(\mathsf{G}(x,\varphi),\,\tau\bigr)=\mathsf{G}(x,\tau).
\end{equation}

Using commutativity of $\mathsf{G}$, equation~\eqref{eq:step-neutral} can be rewritten as
\[
\mathsf{G}\bigl(\tau,\,\mathsf{G}(x,\varphi)\bigr)=\mathsf{G}(\tau,x).
\]
Using the cancellativity $\mathsf{G}$ we conclude
\[
\mathsf{G}(x,\varphi)=x.
\]
Finally, by commutativity also $\mathsf{G}(\varphi,x)=x$. Hence \textbf{$\varphi$ is a neutral element for $\mathsf{G}$}.

\paragraph{\textbf{Characterisation of $V$ from its value on atomic statements}}
\label{subsub:atomic-characterisation}
\begin{figure*}[h!]
	\centering
	\begin{tikzpicture}
\node[lattice node, line width=2pt, dotted] (bot) at (0,0) {$\bot$};

\node[lattice node, fill=green!20, line width=2pt, dotted] (a) at (-1.5, 1) {$a$};
\node[lattice node, fill=green!20, line width=2pt] (b) at (-0.5, 1) {$b$};
\node[lattice node, fill=green!20, line width=2pt] (c) at (0.5, 1)  {$c$};
\node[lattice node, fill=green!20] (d) at (1.5, 1)  {$d$};

\node[lattice node, line width=2pt] 			(ab) at (-2.5, 2) {$a\vee b$};
\node[lattice node, line width=2pt] 			(ac) at (-1.5, 2) {$a\vee c$};
\node[lattice node] 						(ad) at (-0.5, 2) {$a\vee d$};
\node[lattice node, line width=2pt, dotted] (bc) at (0.5, 2)  {$b\vee c$};
\node[lattice node] 						(bd) at (1.5, 2)  {$b\vee d$};
\node[lattice node] 						(cd) at (2.5, 2)  {$c\vee d$};

\node[lattice node, line width=2pt, dotted] (abc) at (-1.5, 3) {\scalebox{0.6}{$a\vee b\vee c$}};
\node[lattice node] 						(abd) at (-0.5, 3) {\scalebox{0.6}{$a\vee b\vee d $}};
\node[lattice node] 						(acd) at (0.5, 3)  {\scalebox{0.6}{$a\vee c\vee d$}};
\node[lattice node] 						(bcd) at (1.5, 3)  {\scalebox{0.6}{$b\vee c\vee d$}};

\node[lattice node] (top) at (0, 4) {$\top$};

\draw[lattice edge, very thick] (bot) -- (a);
\draw[lattice edge, very thick] (bot) -- (b);
\draw[lattice edge, very thick] (bot) -- (c);
\draw[lattice edge] (bot) -- (d);

\draw[lattice edge, very thick] (a) -- (ab); 	\draw[lattice edge, very thick] (a) -- (ac); 	\draw[lattice edge] (a) -- (ad);
\draw[lattice edge, very thick] (b) -- (ab); 	\draw[lattice edge, very thick] (b) -- (bc); 	\draw[lattice edge] (b) -- (bd);
\draw[lattice edge, very thick] (c) -- (ac); 	\draw[lattice edge, very thick] (c) -- (bc); 	\draw[lattice edge] (c) -- (cd);
\draw[lattice edge] (d) -- (ad); 				\draw[lattice edge] (d) -- (bd); 				\draw[lattice edge] (d) -- (cd);

\draw[lattice edge, very thick] (ab) -- (abc); 	\draw[lattice edge] (ab) -- (abd);
\draw[lattice edge, very thick] (ac) -- (abc); 	\draw[lattice edge] (ac) -- (acd);
\draw[lattice edge] (ad) -- (abd); 				\draw[lattice edge] (ad) -- (acd);
\draw[lattice edge, very thick] (bc) -- (abc); 	\draw[lattice edge] (bc) -- (bcd);
\draw[lattice edge] (bd) -- (abd); 				\draw[lattice edge] (bd) -- (bcd);
\draw[lattice edge] (cd) -- (acd); 				\draw[lattice edge] (cd) -- (bcd);

\draw[lattice edge] (abc) -- (top);
\draw[lattice edge] (abd) -- (top);
\draw[lattice edge] (acd) -- (top);
\draw[lattice edge] (bcd) -- (top);
	\end{tikzpicture}
	\caption{A universe $\mathcal{L}_4$. Starting from the knowledge of the value of $V$ on the atomic statements (in green), one can compute the value of $V$ on all the other statements. In particular in order to compute $V(a\vee b \vee c)$ one has to consider $\PrI{a \vee b \vee c}=\{\bot,a,b,c,a\vee b,a\vee c,b\vee c,a\vee b \vee c\}$ (highlighted with thicker lines) and a Boolean sub-lattice of this prime ideal (dotted).}
	\label{fig:L4-atomic-statements}
\end{figure*}
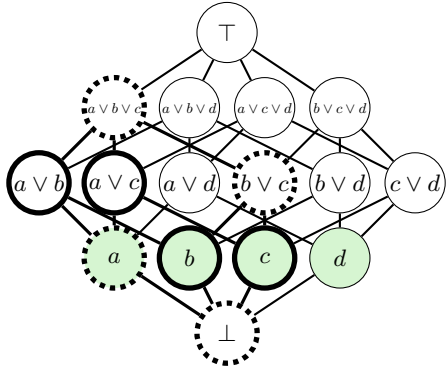
Consider any universe $\mathcal{L}$. We observe that the value of the \property $V(s_i)$  on the atomic statements $\{s_i\}_i$ completely determine the value of $V$ on every statement of the universe. Indeed, one can reconstruct the value of $V$ of any join of atomic statements directly from the value of $V$ on atomic statements. 
For instance, for two distinct atomic statements $s_i \neq s_j$, the value of $V(s_i \vee s_j)$ is simply given by $\GG{V(s_i),V(s_j)}$.
To determine the value of a join of three distinct atoms $s_i \vee s_j \vee s_k$, consider any sub-universe that forms a Boolean sub-lattice $\mathcal{L}_2$ of the prime ideal $\PrI{s_i\vee s_j\vee s_k}$ and use $\mathsf{G}$ again. 
As an example (see FIG.~\ref{fig:L4-atomic-statements}), to compute $V(a\vee b\vee c)$, consider the principal ideal $\PrI{a\vee b\vee c}$, and  focuss on its sub-universe $\{\bot,a,b\vee c, a\vee b \vee c\}$.
Within this sub-universe we have 
\begin{equation}
	V(a\vee b\vee c)=\GG{V(a),V(b\vee c)}=\GG{V(a),\GG{V(b),V(c)}},
\end{equation}
where $V(b\vee c)=\GG{V(b),V(c)}$ is computed from the sub-universe $\PrI{b\vee c}$. 
In this way, the same method extends to compute the value of $V$ on every element of any universe $\mathcal{L}$, using only its values on the atomic statements.
In general one would have that for every set of non-repeating indices of atomic statements $\{i_j\}_{j}$ the following holds
\begin{equation}
  \label{eq:decomposition-into-atomic}
	V\left(\bigvee_{j}s_{i_j}\right)=\GG{s_{i_1},\GG{s_{i_2},\GG{s_{i_3},\dots}}}.
\end{equation}

\paragraph{\textbf{Torsion-free and Divisibility}}

We define the following functions for every $n\in \mathbb{N}_{>0}$
\begin{equation}
\label{eq:def-GGn}
  \begin{cases}
    \GGn{n}{x}=\GG{x,\GGn{n-1}{x}},\\
    \GGn{1}{x}=x.
  \end{cases}
\end{equation}

Principle~\eqref{dec-prin:symmetry-removes-freedom} tells us that for every set of non-repeating indices of atomic statements $\{i_j\}_{j}$ and a value $y\in \mC$ there is just one $V$ such that $V(\bigvee_{j}s_{i_j})=y$ and $V(s_{i_j})=V(s_{i_{j'}})$ for every  $j,j'$.

Now fix $n\ge 1$ and let $s=\bigvee_{j=1}^n s_{i_j}$ be the join of $n$ distinct atoms.
Impose the symmetry condition that all these atoms receive the same valuation, i.e. $V(s_{i_j})=x$ for every $j$.
Then, by equation~\ref{eq:decomposition-into-atomic} and associativity, the value of the join depends only on $x$ and is given by the $n$-fold iterate of $\mathsf{G}$:
\[
V\!\left(\bigvee_{j=1}^n s_{i_j}\right)=\mathsf{G}^{(n)}(x).
\]
This means that  $V(s)=y$  is equivalent to $G^{(n)}(x)=y$.
Therefore, Principle~\ref{dec-prin:symmetry-removes-freedom} is exactly the statement
\begin{equation}
  \label{eq:property-v-G}
  \forall~ y\in \mC, n\in \mathbb{N}_{>0} ~ \exists !~ x\in \mC \mid \GGn{n}{x}=y.
\end{equation}
Recalling that $\varphi$ is the neutral element for $\mathsf{G}$, the property \eqref{eq:property-v-G} implies that \textbf{$\mathsf{G}$ is torsion-free}:
\begin{equation*}
  \forall~n\in \mathbb{N}_+, ~\GGn{n}{x}=\varphi \iff x=\varphi \qquad \mbox{\enquote{\textit{torsion-free}}}.
\end{equation*}
It follows directly from \eqref{eq:property-v-G} that \textbf{$\mathsf{G}$ is divisible}:
\begin{equation*}
  \forall~y\in \mC, \forall~ n\in\mathbb{N}_{>0}, \exists~ x\in\mC \mid \GGn{n}{x}=y \qquad \mbox{\enquote{\textit{divisibility}}.}
\end{equation*}

\subsubsection{\textbf{Part 2: Proof of the main theorem}}

\paragraph{\textbf{Properties of $\mathsf{G}$}}
We start this part recalling all the derived properties of $\mathsf{G}$

From principles~\eqref{dec-prin:classical-compatibility},\eqref{dec-prin:syntactic-locality},\eqref{dec-prin:universality},\eqref{dec-prin:realisability},\eqref{dec-prin:symmetry-removes-freedom} we have that the function $\mathsf{G}:\mC \times \mC \to \mC$ satisfies the following properties:
\begin{enumerate}
\item $\GG{x,y}=\GG{y,x}$ \hfill \enquote{\textit{commutativity}},
\item $\GG{x,\GG{y,z}}=\GG{\GG{x,y},z}$ \hfill \enquote{\textit{associativity}},
\item $\GG{x,\varphi}=x$ \hfill \enquote{\textit{neutral element}},
\item $\GG{x,y}=\GG{x,z}\implies y=z$ \hfill \enquote{\textit{cancellative}}.
\end{enumerate}
Together with property $(v)$ we have that $\mathsf{G}$ also satisfies:
\begin{enumerate}[resume]
\item \label{item:G-torsion-free}$ \forall~n\in \mathbb{N}_+, ~\GGn{n}{x}=\varphi \iff x=\varphi$ \hfill \enquote{\textit{torsion-free}}
\item \label{item:G-divisibility}$\forall~y\in \mC, \forall~ n\in\mathbb{N}_{>0}, \exists~ x\in\mC \mid \GGn{n}{x}=y$ \hfill \enquote{\textit{divisibility}}
\end{enumerate}

\paragraph{\textbf{Proof of the main theorem}}
\MainTheorem*

\begin{proof}
Because of equation~\eqref{eq:decomposition-into-atomic} proving Theorem~\ref{th:sum-rule} is equivalent to proving the following restatement of Theorem~1:
\begin{quote}
	There exists a bijection $\phi:\mC \to \mC$ that we call \textbf{pre-probability mapping} such that
	\begin{equation}
		\phi\left(\mathsf{G}(x,y)\right)=\phi(x)+\phi(y).
	\end{equation}
	Furthermore $\phi(\varphi)=0$.
\end{quote}

To prove it we start by fixing $y\in\mC$ and define the map
\[
\mathsf{G}_y:\mC \to \mC,\qquad \mathsf{G}_y(x):=\mathsf{G}(x,y).
\]
Cancellativity implies that $\mathsf{G}_y$ is injective. To see that it is surjective, consider an arbitrary $z\in\mC$  and consider the  universe $\mathcal{L}_2=\{\bot,a,b,a\vee b\}$ together with the subset
$\mathcal{S}:=\{\bot,b,a\vee b\}\subseteq\mathcal{L}_2$. Define a partial valuation $p:\mathcal{S}\to\mathbb{}\mathbb{C}$ by
\[
p(\bot)=\varphi,\qquad p(b)=y,\qquad p(a\vee b)=z.
\]
This partial valuation is consistent, so by Principle~\eqref{dec-prin:realisability} it extends to an admissible valuation
$V:\mathcal{L}_2\to\mC$. Setting $x:=V(a)$, we obtain that
\[
z \;=\; V(a\vee b) \;=\; \mathsf{G}(V(a),V(b)) \;=\; \mathsf{G}(x,y) \;=\; \mathsf{G}_y(x),
\]
so $\mathsf{G}_y$ is surjective as well. Hence for each $y\in\mathbb{}\mathbb{C}$ the map $\mathsf{G}_y$ is bijective,
and in particular, taking $z=\varphi$, for every $y\in\mathbb{}\mathbb{C}$ there exists an element $y^{-1}\in \mC$
such that $\mathsf{G}(y^{-1},y)=\varphi$. By commutativity, also $\mathsf{G}(y,y^{-1})=\varphi$, so every element has a two-sided
inverse and $(\mC,\mathsf{G},\varphi)$ is an abelian group.

Because of the divisibility and torsion-free of $\mathsf{G}$, $(\mC,\mathsf{G},\varphi)$ is a divisible torsion-free abelian group.

The next step is to prove that $(\mC,\mathsf{G},\varphi)$ is isomorphic to $(\mC,+)$ and thus that there is a group isomorphism
$\phi$ mapping $(\mC,\mathsf{G},\varphi)$ to $(\mC,+)$.
Using the structure theorem of divisible groups~\footnote{This proof requires the axiom of choice.}
(Theorem~3.1 in \cite{fuchs2015}) we know that, since $(\mC,\mathsf{G},\varphi)$ is a torsion-free divisible abelian group,
it is isomorphic to a direct sum of copies of $\mathbb{}\mathbb{Q}$, i.e.\ there exists a cardinal $\kappa$ such that
\[
(\mC,\mathsf{G},\varphi)\cong\left(\bigoplus_{\kappa}\mathbb{}\mathbb{Q},+_{.}\right),
\]
where $+_{.}$ denotes componentwise addition.  
At the same time $(\mC,+)$ is isomorphic to $\left(\bigoplus_{2^{\aleph_0}}\mathbb{}\mathbb{Q},+_{.}\right)$
(Exercise~3.2 in \cite{fuchs2015}). Therefore, to conclude that
\[
(\mC,\mathsf{G},\varphi)\cong (\mC,+),
\]
it suffices to show that $\kappa=2^{\aleph_0}$.

Since $ \left|\bigoplus_{\kappa}\mathbb{}\mathbb{Q}\right| =\kappa $ and  $(\mC,\mathsf{G},\varphi)\cong\left(\bigoplus_{\kappa}\mathbb{}\mathbb{Q},+_{.}\right)$, it follows that
\[
|\mC|=\left|\bigoplus_{\kappa}\mathbb{}\mathbb{Q}\right|=\kappa.
\]
Because $|\mC|=2^{\aleph_0}$, we conclude that $\kappa=2^{\aleph_0}$. Consequently,
\[
(\mC,\mathsf{G},\varphi)\cong\left(\bigoplus_{2^{\aleph_0}}\mathbb{}\mathbb{Q},+_{.}\right)\cong(\mC,+),
\]
and therefore there exists a group isomorphism $\phi:(\mC,\mathsf{G},\varphi)\to(\mC,+)$.

	All in all we obtain
	\begin{align*}
		(\mC,\mathsf{G},\varphi) \xrightarrow{\phi}(\mC,+,0).
	\end{align*}
	We call the bijection $\phi: \mC \to \mC$  a \textit{pre-probability mapping} with $\phi(\varphi)=0$ and define the \textit{pre-probability}
	\begin{equation}
		R\coloneqq \phi \circ V.	
	\end{equation}
	\noindent
	Because of the isomorphism just proved we have
	\begin{align*}
		\phi(\mathsf{G}(x,y))=\phi(x)+\phi(y).
	\end{align*}
	We just proved point $1.$ and $2.$. 
	Now we prove part $3.$
	For every $s\in \mathcal{L}$ we have
	\[
	\GG{V(s),V(\neg s)}=V(s\vee \neg s)=V(\top),
	\]
	mapping to pre-probability we have
	\[
	\phi\left(\GG{V(s),V(\neg s)}\right)=\phi(V(\top)),
	\]
	that is 
	\[
	\phi(V(s))+\phi(V(\neg s))=R(\top),
	\]
	that is
	\[
	R(s)+R(\neg s)=R(\top),
	\]
	and finally
	\[
	R(\neg s)= R(\top)-R(s).
	\]
\end{proof}

\newpage

\section{Proofs}

\subsection{Proof of Lemma~\ref{lem:atoms-to-ranks}}
\label{proof:atoms-to-ranks}
\rankequality*
\begin{proof}
	Fix $n\ge 2$ and assume $V$ satisfies local deducibility and universality.
	Assume moreover that $V$ is constant on atoms, that is there exists $x\in\mathbb C$ such that
	\[
	V(a)=x\qquad\text{for all }a\in\mathrm{Atoms}(\mathcal L).
	\]
	We prove by induction on the level $l$ of the statements that $V$ is constant on all elements of level $l$.
	
	First we show that the claim is valid for $l=1$.
	The level $1$ statements are the atoms, so the claim holds with $v_1=x$.

	Now we assume that for some $l\ge 2$ the claim holds for all levels $1,2,\dots,l-1$, that is for each $r<l$ there exists a number $v_r\in\mathbb C$
	such that $V(t)=v_r$ for every statement $t$ at level $r$.
	Consider a statement $t\in\mathcal L$ at level $l$ and the sub-universe of its principal ideal $\downarrow t$ (a sub-universe with $l$ atomic statements).
	Because of the principle of local deducibility and universality there is a deduction function $\Ded_{l,l}$ such that
	\begin{equation}
	V(t)=\Ded_{l,l}\bigl(
\underbrace{v_1,\dots,v_1}_{\binom{l}{1}},
\underbrace{v_2,\dots,v_2}_{\binom{l}{2}},
\dots,
\underbrace{v_{l-1},\dots,v_{l-1}}_{\binom{l}{\,l-1\,}}
\bigr)
	\end{equation}
	where the input data is coming from our inductive assumption.
	Now consider any other statement $\tilde{t}$ at level $l$ and proceed in the same way to obtain that
		\begin{equation}
	V(\tilde{t})=\Ded_{l,l}\bigl(
\underbrace{v_1,\dots,v_1}_{\binom{l}{1}},
\underbrace{v_2,\dots,v_2}_{\binom{l}{2}},
\dots,
\underbrace{v_{l-1},\dots,v_{l-1}}_{\binom{l}{\,l-1\,}}
\bigr).
\end{equation}
Since the input is the same, the inductive step is proved. Since we showed that for $l=1$ the claim is valid the lemma is proved.
	
\end{proof}

\subsection{Proof of Lemma~\ref{lem:a-freedom}}
\label{proof:a-freedom}
In theorem \ref{th:sum-rule} we proved the existence of a bijection on the valuation functions such that we obtain additivity. The theorem says nothing about the unicity of this bijection, and in fact this bijection in general is not unique at all. The freedom one has in the choice of $\phi$ is given in the following theorem.

\afreedom*
\begin{proof}
	Suppose that $\phi_1,\phi_2:(\mC,\mathsf{G},\varphi) \to(\mC,+)$ are two valid isomorphisms in the proof of theorem \ref{th:sum-rule}.  Then the composition $a_{2,1}:=\phi_2 \circ \phi_1^{-1}$ is an automorphism of $(\mC,+)$ and $\phi_2=a_{2,1}\circ \phi_1$. 
	Thus if $\phi$ is a pre-probability mapping, then for any automorphism $a:(\mC,+)\to (\mC,+)$, also $a\circ \phi$ is a pre-probability mapping.
	An automorphism $a$ of $(\mC,+)$ is any bijection from $\mC$ to $\mC$ that preserves the group structure, that is, such that $a(x,y)=a(x)+a(y)$ for all $x,y\in \mC$.
\end{proof}

\subsection{Proof of Lemma~\ref{lem:ref-frame-from-lin-indep}}
\label{proof:ref-frame-from-lin-indep}
\linindipendenti*
\begin{proof}
We begin by showing that any $a:\mathbb C\to\mathbb C$ satisfying Cauchy additivity
\[
a(x+y)=a(x)+a(y)\qquad(\forall x,y\in\mathbb C)
\]
is $\mathbb Q$-linear.

Setting $x=y=0$ gives $a(0)=a(0)+a(0)$, hence $a(0)=0$.
Also $0=a(x+(-x))=a(x)+a(-x)$, so $a(-x)=-a(x)$.

For $n\in\mathbb N$,
\[
a(nx)=a(\underbrace{x+\cdots+x}_{n\text{ times}})=\underbrace{a(x)+\cdots+a(x)}_{n\text{ times}}=n\,a(x),
\]
by induction using additivity. For $n\in\mathbb Z$ we use $a(-x)=-a(x)$ to extend this to negative integers:
$a(nx)=n\,a(x)$.
Now fix $q\in\mathbb N$. We know that
\[
a(qx)=q\,a(x).
\]
Replacing $x$ with $x/q$ and using additivity repeatedly yields
\[
a(x)=a\!\left(q\cdot \frac{x}{q}\right)=q\,a\!\left(\frac{x}{q}\right),
\]
hence $a(x/q)=a(x)/q$. Therefore, for any $p\in\mathbb Z$ and $q\in\mathbb N$,
\[
a\!\left(\frac{p}{q}x\right)=p\,a\!\left(\frac{x}{q}\right)=\frac{p}{q}\,a(x).
\]
So $a(rx)=r\,a(x)$ for all $r\in\mathbb Q$, that is, $a$ is $\mathbb Q$-linear.

\medskip
From now on, let $R:\mathcal L\to\mathbb C$ be a pre-probability and let $R':=a\circ R$ for some $a-function$ $a$.
Recall that a \emph{semantic frame} $\SL{R}$ is a maximal set of statements
$\{s_1,\dots,s_m\}\subseteq \mathcal L$ such that $\{R(s_1),\dots,R(s_m)\}$ is
$\mathbb Q$-linearly independent; its cardinality $m$ is the \emph{semantic dimension}.

\medskip
\textit{1.} If $s\in \SL{R}$ then $R(s)\neq 0$.

Indeed, if $R(s)=0$ then  $\{R(s)\}$ is $\mathbb Q$-linearly dependent (since, for example, $1\cdot R(s)=0$), contradicting that $\SL{R}$ is $\mathbb Q$-independent.

\medskip
\textit{2.} All semantic frames $\SL{R}$ have the same semantic dimension.

Let
\[
W_R:=\mathrm{span}_{\mathbb Q}\{R(s): s\in\mathcal L\}\subseteq \mathbb C.
\]
By maximality, the set $\{R(s): s\in\SL{R}\}$ is a $\mathbb Q$-basis of $W_R$.
Any two $\mathbb Q$-bases of the same vector space have the same cardinality, hence every semantic frame has size $\dim_{\mathbb Q}(W_R)$.

\textit{3.} All semantic frames are independent of $a$ (hence depend only on $V$).

Since $a$ is $\mathbb Q$-linear and injective, for any finite family $\{t_j\}$,
\[
\sum_j q_j R(t_j)=0 \ \Longleftrightarrow\ 
a\!\left(\sum_j q_j R(t_j)\right)=0 \ \Longleftrightarrow\
\sum_j q_j a(R(t_j))=0 \ \Longleftrightarrow\
\sum_j q_j R'(t_j)=0,
\]
for all $q_j\in\mathbb Q$. Thus $\{R(t_j)\}$ is $\mathbb Q$-independent if and only if
$\{R'(t_j)\}$ is $\mathbb Q$-independent. Therefore the maximal $\mathbb Q$-independent sets of statements, namely the semantic frames, are the same for $R$ and $R'$.
In particular, if $R$ and $R'$ are two pre-probabilities obtained from the same underlying valuation $V$ by choosing different pre-probability mappings, then they have the same semantic frames.

\textit{4.} There is always a semantic frame composed only of atomic statements, so $m\le n$.

Write the $n$ atoms of $\mathcal L$ as $s^{\mathrm{at}}_1,\dots,s^{\mathrm{at}}_n$.
Every $s\in\mathcal L$ is the join of the atoms below it, say
$s=\bigvee_{i\in I_s} s^{\mathrm{at}}_i$ and it holds that
\[
R(s)=\sum_{i\in I_s} R(s^{\mathrm{at}}_i).
\]
Hence $W_R=\mathrm{span}_{\mathbb Q}\{R(s^{\mathrm{at}}_1),\dots,R(s^{\mathrm{at}}_n)\}$.
Choose a maximal $\mathbb Q$-independent subfamily among the atomic values; let it correspond to atoms
$s^{\mathrm{at}}_{i_1},\dots,s^{\mathrm{at}}_{i_m}$. Their values form a $\mathbb Q$-basis of $W_R$, so the corresponding atoms form a semantic frame. In particular $m\le n$.

\textit{5.} If $R'(s_j)=R(s_j)$ for all $j=1,\dots,m$, then $R'(s)=R(s)$ for all $s\in\mathcal L$.

Let $\SL{R}=\{s_1,\dots,s_m\}$. Since $\{R(s_1),\dots,R(s_m)\}$ is a $\mathbb Q$-basis of $W_R$,
for each $s\in\mathcal L$ there exist unique $q_1,\dots,q_m\in\mathbb Q$ with
\[
R(s)=\sum_{j=1}^m q_j\,R(s_j).
\]
Using $\mathbb Q$-linearity of $a$ and the hypothesis $a(R(s_j))=R'(s_j)=R(s_j)$,
\[
R'(s)=a(R(s))
= a\!\left(\sum_{j=1}^m q_j R(s_j)\right)
= \sum_{j=1}^m q_j\,a(R(s_j))
= \sum_{j=1}^m q_j\,R(s_j)
= R(s).
\]
This holds for every $s\in\mathcal L$, as required.
\end{proof}

\subsection{Proof of Proposition~\ref{prop:rule-total-pre-prob}}
\label{proof:rule-total-pre-prob}
\RTPP*
\begin{proof}
	Fix an arbitrary statement $s\in\mathcal{L}$. Since $\bigvee_j s_j=\top$, we have
	\[
	s=s\wedge \top = s\wedge\Bigl(\bigvee_j s_j\Bigr).
	\]
	Because $\wedge$ distributes over $\vee$ in a Boolean lattice, this implies
	\begin{equation}
		\label{eq:decompose-s}
		s=\bigvee_j (s\wedge s_j)=\bigvee_j [s]_{s_j}.
	\end{equation}
	Moreover, the terms in this join are pairwise disjoint: if $j\neq k$ then
	\[
	(s\wedge s_j)\wedge(s\wedge s_k)=s\wedge (s_j\wedge s_k)=s\wedge \bot=\bot,
	\]
	so the family $\{[s]_{s_j}\}_j$ is a disjoint family of statements in $\mathcal{L}$.
	
	Since $R$ is a pre-probability, it is finitely additive over disjoint joins, hence applying $R$ to \eqref{eq:decompose-s} yields
	\begin{equation}
		\label{eq:R-decompose}
		R(s)=\sum_j R(s\wedge s_j)=\sum_j R([s]_{s_j}).
	\end{equation}
	
	Now use synchronisability. By assumption, for each $j$ there exists an $a$-function $a_j$ such that the reparametrised local pre-probability
	\[
	R'_{s_j}\coloneqq a_j\circ R_{s_j}
	\]
	agrees with the restriction of the ambient pre-probability to the relative universe $\mathcal{L}_{s_j}$, that is,
	\[
	R'_{s_j}(t)=R(t)\qquad \text{for all } t\in\mathcal{L}_{s_j}.
	\]
	In particular, for $t=[s]_{s_j}\in\mathcal{L}_{s_j}$ we obtain
	\[
	R([s]_{s_j})=R'_{s_j}([s]_{s_j})
	=a_j\!\left(R_{s_j}([s]_{s_j})\right)
	=a_j\!\left(R(s\mid s_j)\right),
	\]
	where in the last equality we used the definition $R(s\mid s_j)\coloneqq R_{s_j}([s]_{s_j})$.
	
	Substituting this identity into \eqref{eq:R-decompose} gives
	\[
	R(s)=\sum_j a_j\!\left(R(s\mid s_j)\right),
	\]
	which is exactly \eqref{eq:rule-total-pre-prob}.
\end{proof}

\subsection{Proof of Lemma~\ref{lem:q-proportionality}}
\label{proof:q-proportionality}
\qpropsemanticone*
\begin{proof}
Assume $(\mathcal L,R)$ has semantic dimension $1$ in the sense of Definition~\ref{def:gauge-frame-dim}.
Then there exists a semantic frame $\SL{R}=\{t\}$ consisting of a single statement $t\in\mathcal L$.
In particular $R(t)\neq 0$, since otherwise $\{t\}$ would fail to be $\mathbb Q$-linearly independent.

Fix any $\hat s\in\mathcal L$ with $R(\hat s)\neq 0$.
We first show that every $R(s)$ is a rational multiple of $R(t)$.
Let $s\in\mathcal L$.
If $R(s)=0$ we are done by taking the rational $q_t(s)=0$.
If $R(s)\neq 0$, then the set $\{t,s\}$ cannot be $\mathbb Q$-linearly independent, because the semantic dimension is $1$.
Hence there exist rationals $q_1,q_2\in\mathbb Q$, not both zero, such that
\[
q_1\,R(t)+q_2\,R(s)=0.
\]
Moreover $q_2\neq 0$: if $q_2=0$ then $q_1R(t)=0$ forces $q_1=0$ since $R(t)\neq 0$, a contradiction.
Therefore
\[
R(s)=-\frac{q_1}{q_2}\,R(t),
\]
so indeed $R(s)=q_t(s)\,R(t)$ for some $q_t(s)\in\mathbb Q$ (with the convention $q_t(s)=0$ when $R(s)=0$).

Apply the same argument to $\hat s$ in place of $s$ to obtain a rational $q_t(\hat s)\in\mathbb Q\setminus\{0\}$ such that
\[
R(\hat s)=q_t(\hat s)\,R(t).
\]
Combining the two relations gives, for every $s\in\mathcal L$,
\[
R(s)=q_t(s)\,R(t)=\frac{q_t(s)}{q_t(\hat s)}\,R(\hat s),
\]
and the coefficient $q(s):=\frac{q_t(s)}{q_t(\hat s)}$ is rational. This proves \eqref{eq:q-proportionality}.
\end{proof}

\subsection{Proof of Lemma~\ref{lem:semanticquasiprob}}
\label{proof:semanticquasiprob}
\semanticquasiprob*
\begin{proof}
	If $R\equiv 0$ then $a\circ R\equiv 0$ and the claim holds (take $\alpha_a=0$).  Assume $R$ is not identically zero and fix $\hat{s}\in\mathcal{L}$ with $R(\hat{s})\neq 0$.
	Define
	\[
	\alpha_a \;:=\; \frac{a\!\bigl(R(\hat{s})\bigr)}{R(\hat{s})}\in\mathbb{C}.
	\]
	Now let $s\in\mathcal{L}$. By Lemma~\ref{lem:q-proportionality} there exists $q(s)\in\mathbb{Q}$ such that
	$R(s)=q(s)\,R(\hat{s})$. Hence
	\[
	a\!\bigl(R(s)\bigr)
	\;=\; a\!\bigl(q(s)\,R(\hat{s})\bigr)
	\;=\; q(s)\,a\!\bigl(R(\hat{s})\bigr)
	\;=\; q(s)\,\alpha_a\,R(\hat{s})
	\;=\; \alpha_a\,R(s),
	\]
	which is exactly $a\circ R=\alpha_a R$.
	
	Uniqueness: if also $a\circ R=\beta R$, then evaluating at $\hat{s}$ gives
	$\alpha_a R(\hat{s})=\beta R(\hat{s})$, so $\alpha_a=\beta$ since $R(\hat{s})\neq 0$.
\end{proof}

\subsection{Proof of Lemma~\ref{lem:unique-q-coordinates}}
\label{proof:unique-q-coordinates}

\Uniquerationalcoordinates*
\begin{proof}
Consider the $\mathbb{Q}$-vector subspace of $\mathbb{C}$ generated by the set $\{R(t):t\in\mathcal{L}\}$.
By maximality, $\{R(s_k)\}_{k=1}^m$ is a $\mathbb{Q}$-basis for this subspace, hence every $R(s)$ admits an expansion of the form \eqref{eq:R-q-decomp}.
Uniqueness follows from $\mathbb{Q}$-linear independence of $\{R(s_k)\}_{k=1}^m$.
\end{proof}

\subsection{Proof of Proposition~\ref{prop:normalised-frame-exists}}
\label{proof:normalised-frame-exists}
\uniquequasiprob*
\begin{proof}
Since $Q(\top)=1\neq 0$, the singleton $\{\top\}$ is $\mathbb Q$-linearly independent.
Extend it to a maximal $\mathbb Q$-linearly independent set of frame statements, obtaining a semantic frame
\[
\SL{Q}=\{\top,s_2,\dots,s_m\}.
\]
Let $q_k(\cdot)$ be the unique rational coordinate functions from Lemma~\ref{lem:unique-q-coordinates} for this frame.
Applying \eqref{eq:R-q-decomp} to $\top$ gives
\[
Q(\top)=\sum_{k=1}^m q_k(\top)\,Q(s_k),
\qquad\text{with } s_1\equiv \top.
\]
But also
\[
Q(\top)=1\cdot Q(\top)+\sum_{k=2}^m 0\cdot Q(s_k).
\]
By uniqueness of the rational coordinates,
$q_1(\top)=1$ and $q_k(\top)=0$ for $k\ge 2$.
Using Definition~\ref{def:semantic-components} yields
\[
Q^{(1)}(\top)=q_1(\top)\,Q(\top)=1,
\qquad
Q^{(k)}(\top)=q_k(\top)\,Q(s_k)=0 \ \ (k\ge 2),
\]
as claimed.
\end{proof}

\section{AB-labelling}
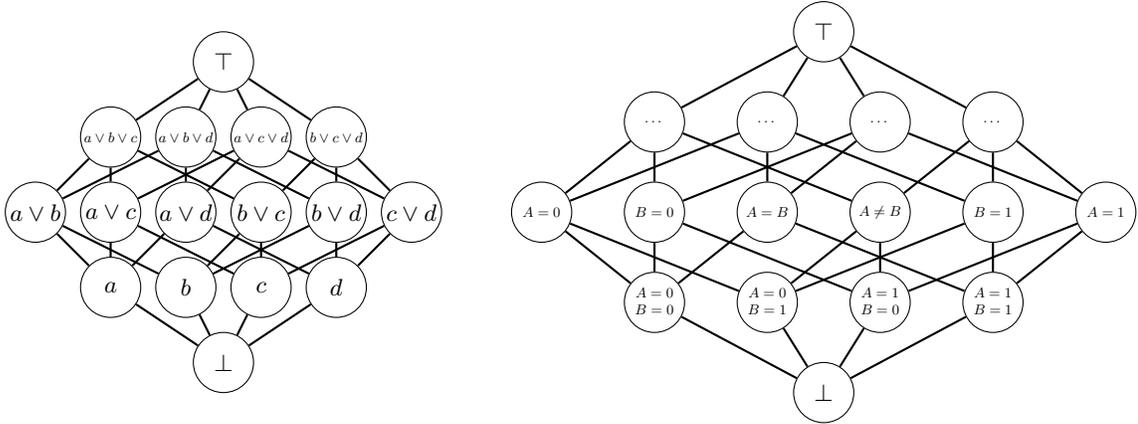
\begin{figure}[h!]
	\begin{minipage}{0.45\textwidth}
		\begin{tikzpicture}
			\node[lattice node] (bot) at (0,0) {$\bot$};
			
			\node[lattice node] (a) at (-1.5, 1) {$a$};
			\node[lattice node] (b) at (-0.5, 1) {$b$};
			\node[lattice node] (c) at (0.5, 1)  {$c$};
			\node[lattice node] (d) at (1.5, 1)  {$d$};
			
			\node[lattice node] 			(ab) at (-2.5, 2) {$a\vee b$};
			\node[lattice node] 			(ac) at (-1.5, 2) {$a\vee c$};
			\node[lattice node]			(ad) at (-0.5, 2) {$a\vee d$};
			\node[lattice node]    			(bc) at (0.5, 2)  {$b\vee c$};
			\node[lattice node]			(bd) at (1.5, 2)  {$b\vee d$};
			\node[lattice node] 			(cd) at (2.5, 2)  {$c\vee d$};
			
			\node[lattice node] 			(abc) at (-1.5, 3) {\scalebox{0.6}{$a\vee b\vee c$}};
			\node[lattice node] 			(abd) at (-0.5, 3) {\scalebox{0.6}{$a\vee b\vee d $}};
			\node[lattice node] 			(acd) at (0.5, 3)  {\scalebox{0.6}{$a\vee c\vee d$}};
			\node[lattice node] 			(bcd) at (1.5, 3)  {\scalebox{0.6}{$b\vee c\vee d$}};
			
			\node[lattice node] 			(top) at (0, 4) {$\top$};
			
			\draw[lattice edge] (bot) -- (a);
			\draw[lattice edge] (bot) -- (b);
			\draw[lattice edge] (bot) -- (c);
			\draw[lattice edge] (bot) -- (d);
			
			\draw[lattice edge] (a) -- (ab); 	\draw[lattice edge] (a) -- (ac); 	\draw[lattice edge] (a) -- (ad);
			\draw[lattice edge] (b) -- (ab); 	\draw[lattice edge] (b) -- (bc); 	\draw[lattice edge] (b) -- (bd);
			\draw[lattice edge] (c) -- (ac); 	\draw[lattice edge] (c) -- (bc); 	\draw[lattice edge] (c) -- (cd);
			\draw[lattice edge] (d) -- (ad); 	\draw[lattice edge] (d) -- (bd); 	\draw[lattice edge,] (d) -- (cd);
			
			\draw[lattice edge] (ab) -- (abc); 	\draw[lattice edge] (ab) -- (abd);
			\draw[lattice edge] (ac) -- (abc); 	\draw[lattice edge] (ac) -- (acd);
			\draw[lattice edge] (ad) -- (abd); 	\draw[lattice edge] (ad) -- (acd);
			\draw[lattice edge] (bc) -- (abc); 	\draw[lattice edge] (bc) -- (bcd);
			\draw[lattice edge] (bd) -- (abd); 	\draw[lattice edge] (bd) -- (bcd);
			\draw[lattice edge] (cd) -- (acd); 	\draw[lattice edge] (cd) -- (bcd);
			
			\draw[lattice edge] (abc) -- (top);
			\draw[lattice edge] (abd) -- (top);
			\draw[lattice edge] (acd) -- (top);
			\draw[lattice edge] (bcd) -- (top);
		\end{tikzpicture}
	\end{minipage}
	\begin{minipage}{0.5\textwidth}
		\begin{tikzpicture}
			\node[lattice node] (bot) at (0,0) {$\bot$};
			
			\node[lattice node] (a) at (-2.25, 1.2) {\scalebox{0.6}{$\begin{matrix}A=0\\B=0\end{matrix}$}};
			\node[lattice node] (b) at (-0.75, 1.2) {\scalebox{0.6}{$\begin{matrix}A=0\\B=1\end{matrix}$}};
			\node[lattice node] (c) at (0.75, 1.2)  {\scalebox{0.6}{$\begin{matrix}A=1\\B=0\end{matrix}$}};
			\node[lattice node] (d) at (2.25, 1.2)  {\scalebox{0.6}{$\begin{matrix}A=1\\B=1\end{matrix}$}};
			
			\node[lattice node] 			(ab) at (-3.75, 2.4) {\scalebox{0.6}{$A=0$}};
			\node[lattice node] 			(ac) at (-2.25, 2.4) {\scalebox{0.6}{$B=0$}};
			\node[lattice node]			(ad) at (-0.75, 2.4) {\scalebox{0.6}{$A=B$}};
			\node[lattice node]    			(bc) at (0.75, 2.4) {\scalebox{0.6}{$A\neq B$}};
			\node[lattice node] 			(bd) at (2.25, 2.4)  {\scalebox{0.6}{$B=1$}};
			\node[lattice node] 			(cd) at (3.75, 2.4)  {\scalebox{0.6}{$A=1$}};
			
			\node[lattice node] 			(abc) at (-2.25, 3.6) {\scalebox{0.6}{$\dots$}};
			\node[lattice node] 			(abd) at (-0.75, 3.6) {\scalebox{0.6}{$\dots$}};
			\node[lattice node] 			(acd) at (0.75, 3.6)  {\scalebox{0.6}{$\dots$}};
			\node[lattice node] 			(bcd) at (2.25, 3.6)  {\scalebox{0.6}{$\dots$}};
			
			\node[lattice node] 			(top) at (0, 4.8) {$\top$};
			
			\draw[lattice edge] (bot) -- (a);
			\draw[lattice edge] (bot) -- (b);
			\draw[lattice edge] (bot) -- (c);
			\draw[lattice edge] (bot) -- (d);
			
			\draw[lattice edge] (a) -- (ab); 	\draw[lattice edge] (a) -- (ac); 	\draw[lattice edge] (a) -- (ad);
			\draw[lattice edge] (b) -- (ab); 	\draw[lattice edge] (b) -- (bc); 	\draw[lattice edge] (b) -- (bd);
			\draw[lattice edge] (c) -- (ac); 	\draw[lattice edge] (c) -- (bc); 	\draw[lattice edge] (c) -- (cd);
			\draw[lattice edge] (d) -- (ad); 	\draw[lattice edge] (d) -- (bd); 	\draw[lattice edge,] (d) -- (cd);
			
			\draw[lattice edge] (ab) -- (abc); 	\draw[lattice edge] (ab) -- (abd);
			\draw[lattice edge] (ac) -- (abc); 	\draw[lattice edge] (ac) -- (acd);
			\draw[lattice edge] (ad) -- (abd); 	\draw[lattice edge] (ad) -- (acd);
			\draw[lattice edge] (bc) -- (abc); 	\draw[lattice edge] (bc) -- (bcd);
			\draw[lattice edge] (bd) -- (abd); 	\draw[lattice edge] (bd) -- (bcd);
			\draw[lattice edge] (cd) -- (acd); 	\draw[lattice edge] (cd) -- (bcd);
			
			\draw[lattice edge] (abc) -- (top);
			\draw[lattice edge] (abd) -- (top);
			\draw[lattice edge] (acd) -- (top);
			\draw[lattice edge] (bcd) -- (top);
		\end{tikzpicture}
	\end{minipage}
	\caption{On the left we use our standard notation for the atomic statements of the universe $\mathcal{L}_4$; on the right we show a different choice of atomic statements. The nodes marked with dotted labels have not been filled in, owing to lack of space.}
	\label{fig:switch-notation}
\end{figure}
\label{sec:AB-labelling}
For a universe $\mathcal{L}_4$ with four atomic statements $\{a,b,c,d\}$ the AB-labelling of statements is given by the following identifications (see Fig.~\ref{fig:switch-notation} for a clear graphical representation of these identifications)
\begin{align}
	\label{eq:AB-atoms}
	a=(A=0,B=0),\qquad b=(A=0,B=1), \nonumber\\
	c=(A=1,B=0),\qquad d=(A=1,B=1),
\end{align}
where $A$ and $B$ denote the outcomes of two binary experiments. It follows that
\begin{align}
	a\vee b =(A=0),\qquad c\vee d=(A=1),\nonumber\\
	a\vee c =(B=0),\qquad b\vee d=(B=1).
\end{align}
We also introduce the shorthand conditional notation (the analogous can be done with quasi-probabilities $Q$ and probabilities $P$)
\begin{align}
	R(B=i_b \mid A=i_a) &\coloneqq R((A=i_a,B=i_b)\mid (A={i_a})),\nonumber\\
	R(A=i_a \mid B=i_b) &\coloneqq R((A=i_a,B=i_b) \mid (B={i_b})),
\end{align}
with $i_a,i_b\in\{0,1\}$.\\

\section{Additional figures}

\begin{figure}[h!]
	\begin{minipage}{0.45\textwidth}
		\begin{tikzpicture}
			\node[lattice node, fill=red!40, line width=2pt] (bot) at (0,0) {$\bot$};
			
			\node[lattice node, fill=green!40, line width=2pt] (a) at (-1.5, 1) {$a$};
			\node[lattice node, fill=blue!40, line width=2pt] (b) at (-0.5, 1) {$b$};
			\node[lattice node, fill=red!20] (c) at (0.5, 1)  {$c$};
			\node[lattice node, fill=red!20] (d) at (1.5, 1)  {$d$};
			
			\node[lattice node, fill=orange!40, line width=2pt] 			(ab) at (-2.5, 2) {$a\vee b$};
			\node[lattice node, fill=green!20] 			(ac) at (-1.5, 2) {$a\vee c$};
			\node[lattice node, fill=green!20]			(ad) at (-0.5, 2) {$a\vee d$};
			\node[lattice node, fill=blue!20]    		(bc) at (0.5, 2)  {$b\vee c$};
			\node[lattice node, fill=blue!20] 			(bd) at (1.5, 2)  {$b\vee d$};
			\node[lattice node, fill=red!20] 			(cd) at (2.5, 2)  {$c\vee d$};
			
			\node[lattice node, fill=orange!20] 			(abc) at (-1.5, 3) {\scalebox{0.6}{$a\vee b\vee c$}};
			\node[lattice node, fill=orange!20] 			(abd) at (-0.5, 3) {\scalebox{0.6}{$a\vee b\vee d $}};
			\node[lattice node, fill=green!20] 			(acd) at (0.5, 3)  {\scalebox{0.6}{$a\vee c\vee d$}};
			\node[lattice node, fill=blue!20] 			(bcd) at (1.5, 3)  {\scalebox{0.6}{$b\vee c\vee d$}};
			
			\node[lattice node, fill=orange!20] 			(top) at (0, 4) {$\top$};
			
			\draw[lattice edge] (bot) -- (a);
			\draw[lattice edge] (bot) -- (b);
			\draw[lattice edge, very thick] (bot) -- (c);
			\draw[lattice edge, very thick] (bot) -- (d);
			
			\draw[lattice edge] (a) -- (ab); 	\draw[lattice edge, very thick] (a) -- (ac); 	\draw[lattice edge, very thick] (a) -- (ad);
			\draw[lattice edge] (b) -- (ab); 	\draw[lattice edge, very thick] (b) -- (bc); 	\draw[lattice edge, very thick] (b) -- (bd);
			\draw[lattice edge] (c) -- (ac); 	\draw[lattice edge, very thick] (c) -- (bc); 	\draw[lattice edge, very thick] (c) -- (cd);
			\draw[lattice edge] (d) -- (ad); 				\draw[lattice edge] (d) -- (bd); 				\draw[lattice edge, very thick] (d) -- (cd);
			
			\draw[lattice edge, very thick] (ab) -- (abc); 	\draw[lattice edge, very thick] (ab) -- (abd);
			\draw[lattice edge] (ac) -- (abc); 	\draw[lattice edge, very thick] (ac) -- (acd);
			\draw[lattice edge] (ad) -- (abd); 				\draw[lattice edge, very thick] (ad) -- (acd);
			\draw[lattice edge] (bc) -- (abc); 	\draw[lattice edge, very thick] (bc) -- (bcd);
			\draw[lattice edge] (bd) -- (abd); 				\draw[lattice edge, very thick] (bd) -- (bcd);
			\draw[lattice edge] (cd) -- (acd); 				\draw[lattice edge] (cd) -- (bcd);
			
			\draw[lattice edge, very thick] (abc) -- (top);
			\draw[lattice edge, very thick] (abd) -- (top);
			\draw[lattice edge] (acd) -- (top);
			\draw[lattice edge] (bcd) -- (top);
		\end{tikzpicture}
	\end{minipage}
	\begin{minipage}{0.45\textwidth}
		\begin{tikzpicture}
			\node[lattice node, fill=red!40, line width=2pt] (bot) at (0,0) {$\bot$};
			
			\node[lattice node, fill=red!20] (a) at (-1.5, 1) {$a$};
			\node[lattice node, fill=red!20] (b) at (-0.5, 1) {$b$};
			\node[lattice node, fill=green!40, line width=2pt] (c) at (0.5, 1)  {$c$};
			\node[lattice node, fill=blue!40, line width=2pt] (d) at (1.5, 1)  {$d$};
			
			\node[lattice node, fill=red!20] 			(ab) at (-2.5, 2) {$a\vee b$};
			\node[lattice node, fill=green!20] 			(ac) at (-1.5, 2) {$a\vee c$};
			\node[lattice node, fill=blue!20]			(ad) at (-0.5, 2) {$a\vee d$};
			\node[lattice node, fill=green!20]    		(bc) at (0.5, 2)  {$b\vee c$};
			\node[lattice node, fill=blue!20] 			(bd) at (1.5, 2)  {$b\vee d$};
			\node[lattice node, fill=orange!40, line width=2pt] 			(cd) at (2.5, 2)  {$c\vee d$};
			
			\node[lattice node, fill=green!20] 			(abc) at (-1.5, 3) {\scalebox{0.6}{$a\vee b\vee c$}};
			\node[lattice node, fill=blue!20] 			(abd) at (-0.5, 3) {\scalebox{0.6}{$a\vee b\vee d $}};
			\node[lattice node, fill=orange!20] 			(acd) at (0.5, 3)  {\scalebox{0.6}{$a\vee c\vee d$}};
			\node[lattice node, fill=orange!20] 			(bcd) at (1.5, 3)  {\scalebox{0.6}{$b\vee c\vee d$}};
			
			\node[lattice node, fill=orange!20] 			(top) at (0, 4) {$\top$};
			
			\draw[lattice edge, very thick] (bot) -- (a);
			\draw[lattice edge, very thick] (bot) -- (b);
			\draw[lattice edge] (bot) -- (c);
			\draw[lattice edge] (bot) -- (d);
			
			\draw[lattice edge, very thick] (a) -- (ab); 	\draw[lattice edge] (a) -- (ac); 	\draw[lattice edge] (a) -- (ad);
			\draw[lattice edge, very thick] (b) -- (ab); 	\draw[lattice edge] (b) -- (bc); 	\draw[lattice edge] (b) -- (bd);
			\draw[lattice edge, very thick] (c) -- (ac); 	\draw[lattice edge, very thick] (c) -- (bc); 	\draw[lattice edge] (c) -- (cd);
			\draw[lattice edge, very thick] (d) -- (ad); 				\draw[lattice edge, very thick] (d) -- (bd); 				\draw[lattice edge] (d) -- (cd);
			
			\draw[lattice edge] (ab) -- (abc); 	\draw[lattice edge] (ab) -- (abd);
			\draw[lattice edge, very thick] (ac) -- (abc); 	\draw[lattice edge] (ac) -- (acd);
			\draw[lattice edge, very thick] (ad) -- (abd); 				\draw[lattice edge] (ad) -- (acd);
			\draw[lattice edge, very thick] (bc) -- (abc); 	\draw[lattice edge] (bc) -- (bcd);
			\draw[lattice edge, very thick] (bd) -- (abd); 				\draw[lattice edge] (bd) -- (bcd);
			\draw[lattice edge, very thick] (cd) -- (acd); 				\draw[lattice edge, very thick] (cd) -- (bcd);
			
			\draw[lattice edge] (abc) -- (top);
			\draw[lattice edge] (abd) -- (top);
			\draw[lattice edge, very thick] (acd) -- (top);
			\draw[lattice edge, very thick] (bcd) -- (top);
		\end{tikzpicture}
	\end{minipage}
	\caption{A universe $\mathcal{L}_4$ with two sub-universe and classes of statements with the same conditional probability highlighted. On the left we consider the sub-universe $\mathcal{L}_{a\vee b}$, the relative probability $P_{a\vee b}$ is the restriction of $P$ on this sub-universe. Extending $P_{a \vee b}$ to a conditional probability on the whole $\mathcal{L}_4$ the values the conditional probability takes are extended as in picture, where the same colour means the same value. Statements of the same colour are statements that belong to the same equivalence class with respect to the equivalence relation $\sim_{a\vee b }$ defined as $s \sim_{a  \vee b} s' \iff s\wedge (a \vee b) = s'\wedge (a \vee b)$. On the right the same is done considering the sub-universe $\mathcal{L}_{c\vee d}$.}
	\label{fig:conditionals}
\end{figure}
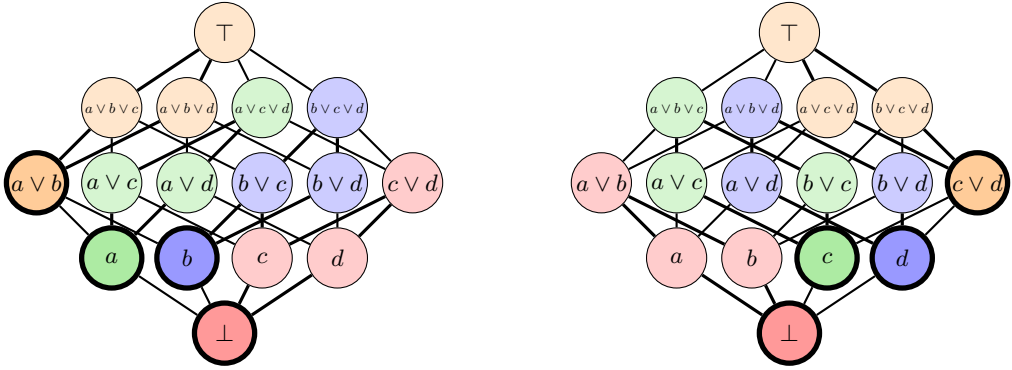

\section{Universality split}
\label{app:global-local-universality}
\subsection*{Global universality}
Let $\mathcal{L}$ be a universe with $n\ge 2$ atoms, and fix a distinguished statement $\tilde{s}\in\mathcal{L}\setminus\{\bot\}$.
By local deducibility there exists a function $F_{\mathcal{L},\tilde s}$ such that
\begin{equation}
	\label{eq:global-SL}
	V(\tilde{s})
	=
	F_{\mathcal{L},\tilde s}\!\left(\{V(s)\}_{s\in\mathcal{L}\setminus\{\tilde{s},\bot\}}\right).
\end{equation}
Global universality requires that the same deduction rule applies after any relabelling by a Boolean-lattice automorphism.
That is, for every Boolean-lattice automorphism $\phi:\mathcal{L}\to\mathcal{L}$ one has
\begin{equation}
	\label{eq:global-U}
	F_{\mathcal{L},\phi(\tilde s)}=F_{\mathcal{L},\tilde s},
\end{equation}
and consequently
\begin{equation}
	\label{eq:global-U-expanded}
	V\!\bigl(\phi(\tilde{s})\bigr)
	=
	F_{\mathcal{L},\tilde s}\!\left(\{V(\phi(s))\}_{s\in\mathcal{L}\setminus\{\tilde{s},\bot\}}\right).
\end{equation}
Since automorphisms preserve rank, Eq.~\eqref{eq:global-U} identifies the deduction rule within each level of the lattice, without forcing it to coincide across different levels.

\subsection*{Universality across sub-universes}
Let $\mathcal{L}_1$ and $\mathcal{L}_2$ be syntactic sub-universes of an ambient universe, and assume there exists a Boolean-lattice isomorphism
\[
\phi_{12}:\mathcal{L}_1\to\mathcal{L}_2.
\]
Let $V_{\mathcal{L}_1}$ and $V_{\mathcal{L}_2}$ denote the restrictions of $V$ to $\mathcal{L}_1$ and $\mathcal{L}_2$, respectively.
Fix a distinguished statement $\tilde{s}\in\mathcal{L}_1\setminus\{\bot\}$.
By local deducibility on $\mathcal{L}_1$ there exists a function $F_{\mathcal{L}_1,\tilde s}$ such that
\begin{equation}
	\label{eq:local-SL}
	V_{\mathcal{L}_1}(\tilde{s})
	=
	F_{\mathcal{L}_1,\tilde s}\!\left(\{V_{\mathcal{L}_1}(s)\}_{s\in\mathcal{L}_1\setminus\{\tilde{s},\bot\}}\right).
\end{equation}
Applying local deducibility to $\mathcal{L}_2$ with distinguished statement $\phi_{12}(\tilde s)$ gives
\begin{equation}
	\label{eq:local-SL2}
	V_{\mathcal{L}_2}\!\bigl(\phi_{12}(\tilde{s})\bigr)
	=
	F_{\mathcal{L}_2,\phi_{12}(\tilde s)}\!\left(\{V_{\mathcal{L}_2}(\phi_{12}(s))\}_{s\in\mathcal{L}_1\setminus\{\tilde{s},\bot\}}\right).
\end{equation}
Local universality requires that isomorphic sub-universes share the same deduction rule, i.e.
\begin{equation}
	\label{eq:local-U}
	F_{\mathcal{L}_2,\phi_{12}(\tilde s)}=F_{\mathcal{L}_1,\tilde s}.
\end{equation}
With \eqref{eq:local-U}, Eq.~\eqref{eq:local-SL2} becomes
\begin{equation}
	\label{eq:local-U-expanded}
	V_{\mathcal{L}_2}\!\bigl(\phi_{12}(\tilde{s})\bigr)
	=
	F_{\mathcal{L}_1,\tilde s}\!\left(\{V_{\mathcal{L}_2}(\phi_{12}(s))\}_{s\in\mathcal{L}_1\setminus\{\tilde{s},\bot\}}\right).
\end{equation}

\section{Universal valuations, Syntactic Locality, and the derivation of probability of Cox--Jaynes}
\label{sec:Cox-Jaynes}
Our derivation begins from a \emph{universal valuation} $V$, an object that will later specialise to a probability assignment (or, more generally, to a quasi-probability or pre-probability).
By contrast, the Cox--Jaynes approach takes as primitive what Cox calls a \enquote{degree of plausibility} \cite{cox2001} and Jaynes reads as an \enquote{extended logic}\cite{jaynes2003}: a numerical valuation of a statement \emph{relative to} a background condition, typically written schematically as
\[
w(A\mid B),
\]
and intended to behave like a conditional probability.\footnote{We keep the discussion informal here. In the Cox--Jaynes tradition, one imposes consistency constraints (sum and product rules, together with mild regularity assumptions) on $w(A\mid B)$, and then shows that it can be regraduated into an ordinary probability function.}
At first sight these starting points seem different: we start from an unconditioned valuation $V$, whereas Cox and Jaynes start from an explicitly relative valuation $w(\cdot \mid \cdot)$.

This difference is, however, largely a matter of presentation.
The reason is that our notion of \emph{Syntactic Locality} builds relativity into the formalism, but using a different language.
Syntactic Locality is the idea that any universe of statements can always be regarded as a sub-universe of a larger ambient universe.
This is inspired by a familiar reductionist physicists' perspective: when a description looks \enquote{intrinsic}, it is often more illuminating to see it as the restriction of a more global description after an embedding or dilation (one may think of extensions by ancillas, Naimark-type dilations, or more broadly the practice of passing to a larger system in which the relevant structure becomes canonical).

\paragraph{Sub-universes as relativisations}

Let $\mathsf{L}$ be the Boolean lattice of statements of some universe, and let $V$ be a valuation on $\mathcal{L}$.
Given a non-trivial statement $b\in\mathcal{L}$, there is a natural associated sub-universe obtained by restricting attention to what is compatible with $b$.
Concretely, one may represent this sub-universe as the principal ideal $\mathcal{L}_{b}$.

From a valuation on the ambient universe one obtains a \emph{relative valuation} on $\mathcal{L}_{b}$ by restriction of the domain.
Thus, even though our development defines conditional probabilities \emph{from} probabilities, the mechanism is the same: conditionality is simply what one obtains by passing to a sub-universe.

\paragraph{Why this is close to Cox and Jaynes}
It is possible to consider that our primitive object is not just $V$ but the family of all its relativisations.
This family is precisely the sort of object that Cox and Jaynes take as primitive, except that they choose to start the story already at the level of the relative quantities.
From a purely mathematical point of view, the two starting points are therefore much closer than they appear at first sight:
Cox and Jaynes begin from a conditional valuation $w(\cdot\mid\cdot)$ and reconstruct a global probability calculus, while
we instead begin from a global valuation $V$ and regard conditionality as the induced calculus on sub-universes.

Syntactic Locality is the notion that helps us do this.
Because every universe can be embedded into a larger one, any valuation we ever apply is, in practice, applied to some sub-universe selected by context or by an agent's informational situation.
The global object $V$ can be best understood as a book-keeping device for all coherent local valuations, rather than as a metaphysical claim that probabilities are written into the world as absolute numbers.

\paragraph{Relativity: information versus embedding}
The same reconciliation can be stated philosophically.
Cox and Jaynes place relativity at the core of their story, plausibilities are always conditioned on something, typically an agent's background information.
There is no agent-independent valuation of a statement, only valuations relative to what is known.
In our language, one might be tempted to read the \enquote{universal valuation} as an objective quantity, and hence as a departure from the Cox and Jaynes viewpoint.

Syntactic Locality softens this apparent tension.
It can be understood as telling us that what is universal is not a particular numerical assignment, but the structure that relates different relativisations.
Changing what we condition on is, for Cox and Jaynes, a change of information.
Changing which sub-universe we work in is, for us, a change of relativisation via embedding and restriction.
These are two descriptions of the same mathematical operation, and the coherence constraints one ought to impose are the same in spirit.
What matters is that the passage between contexts is well defined.

\paragraph{Pre-probabilities as objective coherence data}
This perspective also clarifies why we insist on working, at the global level, with pre-probabilities rather than probabilities.
We find, starting from first principles, that the global structure that organises all relativisations need not itself satisfy positivity, nor need it be uniquely fixed by a single pre-probability mapping, or, in the context of theory with semantic dimension $1$, by a single normalisation convention.
In fact, once one allows richer semantic structure, the global valuation may carry $a$-functions freedom while still inducing well defined relative probabilistic models on the sub-universes of interest.
In that sense, the universal valuation is not a pseudo-probability of the world, but the coherent rule that produces relative probabilities when a context is specified.

One may therefore speak of a form of objectivity that is weaker, and arguably more appropriate for foundational physics:
objectivity is not the claim that there is a unique context-free probability assignment, but the claim that there is an objective coherence structure governing how relative assignments must transform under changes of context.
A universal pre-probability is objective precisely in this limited sense.
It encodes the admissible rules of deducibility connecting different relative valuations, and hence it captures what can be shared between agents even when their sub-universes, and thus their conditionalisations, differ.

Finally, Cox and Jaynes highlight the relativity and then derive the global calculus; we instead start from a global calculus and then read conditionality as restriction.
Syntactic Locality shows that these are two views of the same mechanism.

As a final note, we observe that adding the monotonicity assumption of Jaynes to our quasi-probability representation singles out probability representations, but in our case monotonicity assumption is recovered without assuming it and without even assuming continuity.

\section{The case for rational probabilities}
\label{sec:rational}
Proposition~\ref{prop:normalised-frame-exists} clarifies how normalisation interacts with higher semantic dimension.
Even when $Q$ takes values spanning several $\mathbb Q$-independent directions in $\mathbb C$, only one direction can behave as a (quasi-)probability.
All remaining independent directions are necessarily top-zero and therefore remain at the level of pure pre-probabilities.

In particular, once we restrict attention to the genuinely normalised direction $Q^{(1)}$, we return to the semantic-dimension-$1$ situation.
The higher-dimensional features of $Q$ are not additional (quasi-)probabilistic degrees of freedom: they are extra $\mathbb Q$-independent pre-probabilistic components that cannot themselves be normalised and hence cannot be upgraded to quasi-probabilities or probabilities.

For this reason we restricted to semantic dimension $1$, and we interpret this structure as a clue that our reconstruction suggests that quasi-probabilities, and thus classical probabilities, should be rational-valued, or more conservatively, without further regularity axioms (see appendix \ref{sec:Cox-Jaynes} and section \ref{sec:holomorphic}), canonical normalisation fixes only the $\mathbb{Q}$-scalar freedom and irrational components behave like additional semantic coordinates. We emphasise that this suggestion does not arise from experimental considerations. For example, one might argue on frequentist grounds that probabilities should be rational because no finite agent can perform infinitely many trials. Here the request for rational values appears as a purely structural constraint of the reconstruction, it follows from the axioms governing valuations and their parametrisations, without any appeal to empirical limitations or measurement procedures. In this sense, theoretical and experimental viewpoints point in the same direction, both pointing to the possibility that, without additional regularity axioms on admissible regraduations, irrational components may behave as extra semantic coordinates rather than as canonically meaningful probabilistic degrees of freedom.

\paragraph{Residual freedom along the irrational directions}

When the semantic dimension is $m>1$ one may vary the irrational components without affecting the top element at all, because those components belong to top-zero semantic directions.

We make this explicit with a minimal example.

\begin{example}
\label{ex:irrational-prob}
Consider a two-atom universe with atoms $\{a,b\}$, and define
\[
P(a)=\frac{1}{2}-\frac{\sqrt{2}}{10},
\qquad
P(b)=\frac{1}{2}+\frac{\sqrt{2}}{10},
\]
so that $P(\top)=P(a\vee b)=1$ and $P(a),P(b)\ge 0$.
We consider a semantic decomposition of $P$ into a rational part plus an irrational part, $P = P_{\mathrm{rat}} + P_{\mathrm{irr}}$, such that
\begin{align*}
P_{\mathrm{rat}}(a)=P_{\mathrm{rat}}(b)=\frac12, \\
P_{\mathrm{irr}}(a)=-\frac{\sqrt2}{10},\ \ P_{\mathrm{irr}}(b)=\frac{\sqrt2}{10}.
\end{align*}
Then $P_{\mathrm{rat}}(\top)=1$ while $P_{\mathrm{irr}}(\top)=0$, so the entire irrational contribution sits in a top-zero semantic direction.

Now choose an additive $a$-function that fixes the rational direction and rescales the irrational one, for instance
\[
a(x+y\sqrt{2})\coloneqq x+2y\sqrt{2}
\qquad (x,y\in\mathbb{Q}),
\]
defined on the $\mathbb{Q}$-linear span of $\{1,\sqrt{2}\}$.
Define ${P'\coloneqq a\circ P}$. Then
\[
P'(\top)=a(P(\top))=a(1)=1,
\]
but
\[
P'(a)=\frac{1}{2}-\frac{\sqrt{2}}{5},
\qquad
P'(b)=\frac{1}{2}+\frac{\sqrt{2}}{5}.
\]
Thus $P'$ remains a probability assignment, but it corresponds to a different irrational parametrisation.
\end{example}

One might object that irrationals can be approximated arbitrarily well by rationals, so any additional semantic direction could be made negligible.
For instance, in Example~\ref{ex:irrational-prob} one could replace $\sqrt{2}$ with a truncated continuous fraction approximation obtaining an irrational error term
\[
\beta_n
\;:=\;
\sqrt{2}
\;-\;
\Bigl(
1
+
\cfrac{1}{
\underbrace{
2 + \cfrac{1}{
2 + \cfrac{1}{
\ddots + \cfrac{1}{2}
}}
}_{\text{$n$ levels}}
}
\Bigr),
\]
which can be made arbitrarily small by increasing $n$.
However, this conclusion is misleading, there always exists an $a$-function that fixes $1$ but magnifies the error by an arbitrarily large factor.
An additional semantic frame is therefore a genuine new degree of freedom, not a small correction.

\smallskip
\noindent
Finally, one could try to eliminate the ambiguity by imposing a separate normalisation on each semantic direction, effectively replacing a single quasi-probability by a collection of semantic-dimension-$1$ quasi-probabilities.
But then the resulting global parametrisation depends on extra choices (in particular, on a choice of semantic frame), and the canonicity is lost.
This is one reason why, within the present reconstruction, semantic dimension $1$ has been chosen as the natural regime in which to study (quasi-)probabilities.

\section{The case of irrational probabilities: against hypothetical/infinite frequentism}
\label{sec:irrational}
One of the most common ways of connecting probability theory with empirical data is the frequentist interpretation, according to which probabilities are identified with relative frequencies. Two main variants are usually distinguished: finite frequentism \cite{hajek1996} and hypothetical frequentism \cite{hajek2009} (also called infinite frequentism). The framework developed in this work is, in principle, compatible with finite frequentism, although we do not pursue this connection here. In fact our definition of probabilities and quasi-probabilities relies exclusively on rational valuations. This feature appears to be in tension with hypothetical frequentism, which typically assigns probabilities as limits of infinite relative frequencies and thereby allows for irrational values.

In order to make a case in favour of rational probabilities, I present below an argument against the hypothetical frequentist interpretation, by examining two cases in which the hypothetical frequentist definition of probability gives rise to particularly interesting problems and inconsistencies. 

Consider a completely biased coin that returns \textsf{tail} with certainty and \textsf{head} with probability zero. In hypothetical frequentism, a possible outcome of the experiment is an infinite sequence
\[
x:\mathbb{N}\to\{\textsf{head},\textsf{tail}\},
\]
and the probability assigned to an outcome $\textsf{out}\in\{\textsf{head},\textsf{tail}\}$ is defined as the limiting relative frequency
\begin{equation}
p_x(\textsf{out})
=\lim_{n\to\infty}\frac{1}{n}\sum_{i=1}^{n}\delta\!\bigl(x(i),\textsf{out}\bigr),
\end{equation}
where $\delta(u,v)=1$ if $u=v$ and $\delta(u,v)=0$ otherwise.

One admissible sequence is
\[
x(n)=\textsf{tail}\qquad \forall n,
\]
for which
\[
p_x(\textsf{head})=0,
\qquad
p_x(\textsf{tail})=1.
\]
Another admissible sequence is
\[
y(n)=
\begin{cases}
\textsf{head} & \text{if } \log_2(n)\in\mathbb{N},\\
\textsf{tail} & \text{otherwise},
\end{cases}
\]
so that \textsf{head} occurs only at times $n=2^k$. In this case, again,
\[
p_y(\textsf{head})
=\lim_{n\to\infty}\frac{\lfloor \log_2(n)\rfloor}{n}
=0,
\]
and therefore $p_y(\textsf{tail})=1$.
Both sequences of outcome $x(n)$ and $y(n)$ are proper of a coin with probability $0$ of returning $\textsf{head}$ and probability $1$ of returning $\textsf{tail}$. 

\subsection{An experiment with definite outcomes with well defined probabilities that leads to undefined probability}
We now consider the following experiment. A printer outputs one symbol per second, either $0$ or $1$. The printer has two internal states, denoted $\textsf{A}$ and $\textsf{B}$. If the printer is in state $\textsf{A}$, then at time $t\in\mathbb{N}$ it prints
\begin{equation}
f_{\textsf{A}}(t)=
\begin{cases}
0 & \text{if } t \text{ is even},\\
1 & \text{if } t \text{ is odd},
\end{cases}
\end{equation}
whereas in state $\textsf{B}$ it prints $0$ at every time step.

Before each print, a coin is flipped. If the outcome is \textsf{head}, the internal state of the printer is switched; if it is \textsf{tail}, the state is left unchanged. The printer starts in state $\textsf{A}$. We use the completely biased coin above, for which $p(\textsf{head})=0$ in the hypothetical frequentist sense, and ask for the probability that the printed symbol is $0$.

If the realised coin sequence is $x(n)=\textsf{tail}$ for all $n$, then the printer never switches state and the relative frequency of $0$ on the paper is $\tfrac{1}{2}$. However, if the realised sequence is $y$, the situation changes. The relative frequency of printed $0$s does not converge: it oscillates and remains strictly larger than $\tfrac{1}{2}$. Consequently, the hypothetical frequentist probability of printing $0$ does not exist.

Thus, within hypothetical frequentism, the assignment of probabilities to the coin outcomes leads to a situation in which the probability of a well-defined derived event fails to exist. This indicates an internal tension in the interpretation.

\subsection{An experiment with definite outcomes where local probabilities are well defined but without global probabilities}
We now slightly extend the experiment to show that the same mechanism can generate scenarios in which some probabilities exist while others do not. Suppose that the printer now outputs two symbols at each time $t$, one in a left column and one in a right column. The printer again has two internal states, $\textsf{A}$ and $\textsf{B}$. In state $\textsf{A}$ it prints
\[
\bigl(f^{L}_{\textsf{A}}(t),f^{R}_{\textsf{A}}(t)\bigr)=
\begin{cases}
(0,1) & \text{if } t \text{ is even},\\
(1,0) & \text{if } t \text{ is odd},
\end{cases}
\]
whereas in state $\textsf{B}$ it prints
\[
\bigl(f^{L}_{\textsf{B}}(t),f^{R}_{\textsf{B}}(t)\bigr)=
\begin{cases}
(0,0) & \text{if } t \text{ is even},\\
(1,1) & \text{if } t \text{ is odd}.
\end{cases}
\]
As before, the printer starts in state $\textsf{A}$ and switches state only upon the outcome \textsf{head}. We again consider the hypothetical frequentist scenario in which the realised coin sequence is $y$, for which $p_y(\textsf{head})=0$.

In this case, the relative frequencies of symbols in each column separately are well defined and in both columns the relative frequency of $0$ converges to $\tfrac{1}{2}$, and similarly for $1$. Thus, all local probabilities exist and are equal to $\tfrac{1}{2}$.

However, when one considers the joint outcomes $\bigl(f^{L}(t),f^{R}(t)\bigr)$, the relative frequencies fail to converge. The sparse but unbounded switches between the two internal states induced by the sequence $y$ produce persistent oscillations in the frequencies of the joint symbols $(0,0)$, $(0,1)$, $(1,0)$, and $(1,1)$. As a result, no hypothetical frequentist probability can be assigned to the global two-column output.

This example shows that, within hypothetical frequentism, the existence of well-defined local probabilities does not guarantee the existence of a well-defined global probability distribution.

\end{document}